\documentclass[conference]{IEEEtran}
\usepackage[utf8]{inputenc}
\usepackage[nolist]{acronym}
\usepackage{cite}
\usepackage{amsmath,amssymb,amsfonts, amsthm}
\usepackage{graphicx}
\usepackage{textcomp}
\usepackage{xcolor}
\usepackage{tikz}
\usepackage{pgfplots}
\usepackage{multirow}
\usepackage{makecell}
\usepackage{algorithm}
\usepackage{algpseudocode}
\usepackage{framed}
\usepackage{soul}
\usepackage{bm}
\usepackage{pdfcomment}
\usepackage{fancyhdr}
\usepackage{subcaption}

\usepgfplotslibrary{external}
\tikzset{every picture/.style={line width=0.75pt}}
\usetikzlibrary{positioning, shapes, trees, automata, calc, fit, automata}
\def\BibTeX{{\rm B\kern-.05em{\sc i\kern-.025em b}\kern-.08em
		T\kern-.1667em\lower.7ex\hbox{E}\kern-.125emX}}

\fancypagestyle{firststyle}
{
	\setlength{\headheight}{2cm} 
	\fancyhf[c]{This is the Accepted Version of the work. The definitive version was published in\\ \emph{Proc. ACM Meas. Anal. Comput. Syst.}, Vol. 4, 2, Article 27 (June 2020) ACM, New York, NY. \\ \href{https://doi.org/10.1145/3392145}{https://doi.org/10.1145/3392145}}
	\fancyfoot{}
}

\tikzset{every picture/.style={line width=0.75pt}}
\tikzset{every picture/.style={line width=0.75pt}}
\tikzset{cross/.style={cross out, draw=black, fill=none, minimum size=2*(#1-\pgflinewidth), inner sep=0pt, outer sep=0pt}, cross/.default={2pt}}
\usetikzlibrary{positioning, shapes, trees, automata, calc, fit, automata, shapes.misc, arrows, patterns}

\providecommand{\todo}[1]{}

\providecommand{\hze}{}
\providecommand{\h}{}
\providecommand{\T}{}
\providecommand{\rema}{}
\providecommand{\ant}{}
\providecommand{\hr}{}
\providecommand{\rr}{}

\renewcommand{\todo}[1]{\textcolor{red}{\underline{\textbf{\#TODO :}} #1}}

\renewcommand{\hze}{\mathcal{H}_{\text{TAI}}}
\renewcommand{\h}{\mathcal{H}}
\renewcommand{\T}{T_{\text{start}}}
\renewcommand{\rr}{\text{Reg}}
\renewcommand{\rema}{\textbf{Remark: }}
\renewcommand{\ant}{\textbf{Numerical Application to TSN: }}
\renewcommand{\hr}{\mathcal{H}_{\text{Reg}}}

\providecommand{\DIFdel}[1]{} %Don't show deleted text
\renewcommand{\DIFdel}[1]{\sout{#1}} %Don't show deleted text

\def\ben{\begin{equation*}}
\def\een{\end{equation*}}
\def\barr{\begin{array}}
\def\earr{\end{array}}
\newtheorem{proposition}{Proposition}

\newtheorem{definition}{Definition}

\newtheorem{lemma}{Lemma}

\begin{document}

	%%
%% The "title" command has an optional parameter,
%% allowing the author to define a "short title" to be used in page headers.
\title{On Time Synchronization Issues in Time-Sensitive Networks with Regulators and Nonideal Clocks}
\author{\IEEEauthorblockN{Ludovic Thomas}
	\IEEEauthorblockA{\textit{I\&C} \\
		\textit{EPFL}\\
		Lausanne, Switzerland \\
		ludovic.thomas@epfl.ch}
	\and
	\IEEEauthorblockN{Jean-Yves Le Boudec}
	\IEEEauthorblockA{\textit{I\&C} \\
		\textit{EPFL}\\
		Lausanne, Switzerland \\
		jean-yves.leboudec@epfl.ch}
}

\maketitle
\thispagestyle{firststyle}
% !TeX spellcheck = en_US
\begin{abstract}
	Flow reshaping is used in time-sensitive networks (as in the
	context of IEEE TSN and IETF Detnet) in order to reduce burstiness
	inside the network and to support the computation of guaranteed
	latency bounds. This is performed using per-flow regulators (such
	as the Token Bucket Filter) or interleaved regulators (as with IEEE
	TSN Asynchronous Traffic Shaping, ATS). The former use one FIFO
	queue per flow, whereas the latter use one FIFO queue per input
	port. Both types of regulators are beneficial as they cancel the
	increase of burstiness due to multiplexing inside the network.
	It was demonstrated, by using network calculus, that they
	do not increase the worst-case latency. However, the properties of
	regulators were established assuming that time is perfect in all
	network nodes. In reality, nodes use local, imperfect clocks. Time-sensitive networks exist in two flavours: (1) in non-synchronized networks,
	local clocks run independently at every node and their deviations
	are not controlled and (2) in synchronized networks, the deviations of
	local clocks are kept within very small bounds using for example a
	synchronization protocol (such as PTP) or a satellite based
	geo-positioning system (such as GPS). We revisit the properties of
	regulators in both cases. In non-synchronized networks, we show that
	ignoring the timing inaccuracies can lead to network instability
	due to unbounded delay in per-flow or interleaved regulators. We
	propose and analyze two methods (rate and burst cascade, and
	asynchronous dual arrival-curve method) for avoiding this problem. In
	synchronized networks, we show that there is no instability with
	per-flow regulators but, surprisingly, interleaved regulators
	can lead to instability. To establish these results, we develop a
	new framework that captures industrial requirements on clocks in
	both non-synchronized and synchronized networks, and we develop a toolbox
	that extends network calculus to account for clock imperfections.

\end{abstract}

% !TeX spellcheck = en_US
\section{Introduction}
	Time-sensitive networks support real-time applications in many industries such as automation~\cite{iecIECIEEE608022019}, avionics~\cite{AFDX,TTE}, space~\cite{ecssSpaceWireLinksNodes2008}, and automobile~\cite{ieeeDraftStandardLocal2019b}. Recent work in time-sensitive networks include  the \ac{TSN} task group of the \ac{IEEE} and the Detnet working group of the \ac{IETF}. Both aim to provide deterministic worst-case delay and jitter bounds with seamless reconfiguration~\cite{ieeeIEEEStandardLocal2018} and redundancy~\cite{ieeeIEEEStandardLocal2017}.
	
	Reshaping flows inside the network by means of traffic regulators helps achieve these objectives. Traffic regulators are hardware elements that, placed before a multiplexing stage, remove the increased burstiness due to interference with other flows in previous hops. They support higher scalability and efficiency of time-sensitive networks and enable the computation of guaranteed latency-bounds in networks with cyclic dependencies~\cite{wandeler2006performance,specht2016urgency,mohammadpourLatencyBacklogBounds2018}. They come in two types: the \acf{PFR} (also called ``per-flow shaper'')~\cite[Section 1.7.4]{leboudecNetworkCalculusTheory2001} and the \acf{IR}~\cite{leboudecTheoryTrafficRegulators2018}; the \ac{IR} processes flow aggregates and the \ac{PFR} processes each flow individually and requires one queue per flow.
	
	In both types, each flow has its own regulation parameter, usually in terms of burst and rate. Regulators then delay any packet whose release would violate the regulation parameter. A well-known example of a \ac{PFR} is the Linux's Token-Bucket Filter~\cite{wagner2001short}. Configured with a rate $r$ and a burst $b$, it makes sure that over any window of duration $t$, no more than $rt+b$ bits are released by the regulator. Hence, the evaluation of elapsed time is at the heart of the operation of any regulator. When a regulator can base its computations on an ideal clock, previous studies have established that it enjoys %positive properties. A fundamental one is
	the ``shaping-for-free'' property, i.e., a regulator that removes the burstiness increase caused by a \ac{FIFO} system does not increase the worst-case delay of flows \cite[Thms 1.5.2 and 1.7.3]{leboudecNetworkCalculusTheory2001},\cite[Thm 5]{leboudecTheoryTrafficRegulators2018}. This property is essential to the analysis of time-sensitive networks with regulators.
% for the \ac{IR}. It stipulates that when the regulator removes the burstiness incrase caused in a \ac{FIFO} system, then it does not increase the worst-case delay of flows.
	
	In reality, the clock used by a regulator is nonideal, and the clocks used by different devices in a network deviate slightly from true time. Time-sensitive networks are either synchronized or non-synchronized. In non-synchronized networks, local clocks run independently at every node and their deviations are not controlled.
In synchronized networks, the deviations %of local clocks
are kept within bounds, using a
time-synchronization protocol or a \acl{GNSS}. With time-synchronization methods such as the \acf{PTP}~\cite{ieeeIEEEStandardPrecision2008}, WhiteRabbit~\cite{moreiraWhiteRabbitSubnanosecond2009} or the \ac{GPS}~\cite{powersGPSGalileoUTC2004}, the clock deviation bound is $\sim$1$\mu$s or less; we call such cases ``tightly synchronized''. Here, the clock deviation is smaller than the latency requirements of network flows ($\sim$1ms for avionics systems), and tightly-synchronized networks are typically analyzed as if clocks would be ideal. Some other networks require time synchronization only for network management purposes and use a method such as the \acf{NTP}~\cite{murtaQRPp14CharacterizingQuality2006}, which provides a clock deviation bound of $\sim$100ms; we call such cases ``loosely synchronized''.

	Consider a flow of data shaped at some point in the network by a token-bucket filter with rate $r$ and burst $b$ (Figure~\ref{fig:prob:pfr}). After traversing some network element, say $S$, the flow typically becomes more bursty and no longer satisfies the burst tolerance $b$. Assume a per-flow regulator $R$ is applied to the flow at the output of $S$, in order to re-enforce the burst tolerance $b$, using a token-bucket filter with same rate $r$ and burst $b$. The token-bucket filter delays data of the flow that comes out of $S$ in too large bursts. However, if clocks are ideal, the shaping-for-free property means that the worst-case delay of the flow through $S$ and $R$ is the same as the worst-case delay through $S$ alone (i.e., late packets are not delayed by the regulator). Now assume the clocks are nonideal and the network is not synchronized. If the clock at the token-bucket filter $R$ is too slow, the true value of $r$ implemented by $R$ is slightly less than that of the source; this will lead to a slow, but steady, buildup of backlog at the input buffer of $R$, which might lead over time to an arbitrarily large delay or unexpected loss.

	This simple example suggests that clock nonidealities might significantly affect the delay analysis of time-sensitive networks with regulators. It has raised concerns and discussions in the ongoing standardization process of \ac{IEEE} \acs{ATS} \cite{ieeeDraftStandardLocal2019a}. In this paper, we provide theoretical foundations to the problem and we determine to what extend delay analyses are affected in non-synchronized and synchronized networks.
	%\textbf{Contributions:}
Our main contributions are: %\vspace{-0.5cm}%
	%\begin{itemize} %
		
$\bullet$ We propose a time model for non-synchronized and synchronized networks; it can be used for computation of latency bounds. Using the example of \ac{TSN}, we show that the model parameters can easily be obtained from industrial requirements.
		
$\bullet$ To compute latency bounds when clocks are nonideal, we propose a toolbox to be used with other network calculus results.
		
$\bullet$ In non-synchronized networks, we show that the configuration of regulators must be adapted to take into account the clock imperfections. If not adapted, we prove that regulators can yield unbounded latencies or unexpected packet losses.
		
$\bullet$ For non-synchronized networks, we refine and provide a formal justification for the method proposed in \cite[Annex V.8]{ieeeDraftStandardLocal2019a} for configuring the regulators and for avoiding the problem mentioned above. This rate- and burst-cascade method increases the rate and burst tolerance at every regulator along the path of a flow; it requires that the parameters of a regulator depend on the position of the regulator along the flow path, thus it adds complexity to the control plane. It applies to both \ac{PFR} and \ac{IR}. We propose an alternative method, \acf{ADAM}, that uses the same regulator parameters at all re-shaping points on the flow path thus makes the control plane simpler; it applies to \ac{PFR} only. We also compare the delay bounds obtained with each method, but we leave to future work the practical evaluation of the two methods.
		
$\bullet$ In synchronized networks, %\todo{synchronous ou synchronized ?},
we compute a bound on the delay penalty imposed by \ac{PFR}s. In tightly-synchronized networks, this penalty is small compared to latency bounds, and the current practice of ignoring it is adequate. In contrast, in loosely-synchronized networks, we show an example where the delay penalty can be significant thus should be taken into account.

$\bullet$ The conclusions are very different in synchronized networks with \ac{IR}s. We show that, even in tightly-synchronized networks, \ac{IR}s can yield unbounded delay or unexpected loss if the residual clock inaccuracies are not accounted for. The method of rate and burst cascade can be used to avoid this problem.	
%\end{itemize} %
	
In Sections~\ref{sec:related} and~\ref{sec:prob}, we present the related work and provide the necessary background on time-sensitive networks, regulators and network calculus. We introduce our assumptions and our time and network models in Section~\ref{sec:sys-model}. In Section~\ref{sec:toolbox}, we detail our toolbox of network calculus results. We then analyze regulators in non-synchronized and synchronized networks in Sections~\ref{sec:manag-async} and \ref{sec:manag-sync}, respectively. We make our conclusive remarks in Section~\ref{sec:conclu}. Proofs of propositions are available in Appendix~\ref{sec:appendix}.
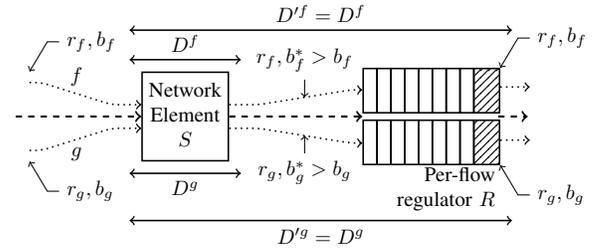
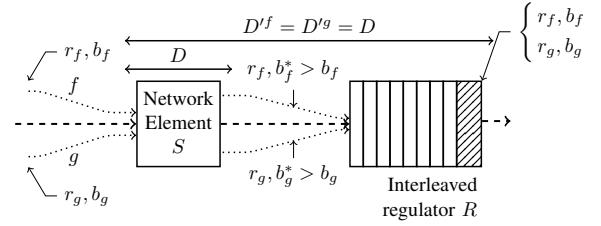
\begin{figure}\centering\begin{minipage}{0.9\linewidth}\centering %
		\resizebox{\linewidth}{!}{% !TeX spellcheck = en_US
% !TeX rootfile = article.tex
\begin{tikzpicture} %
\tikzstyle{mn} = [draw, minimum height=1.5cm]
\tikzstyle{nn} = [draw, minimum height=0.75cm, anchor=south west]
\tikzstyle{ns} = [draw, minimum height=0.75cm, anchor=north west]
\tikzstyle{fil} = [draw, pattern=north east lines, text width=0.2cm]
\tikzstyle{hfil} = [fil, minimum height=0.75cm]
\tikzstyle{ffil} = [fil, minimum height=1.5cm]
\tikzstyle{cor} = [xshift=-\pgflinewidth]

%SOURCE
\node[minimum height=1.5cm] at (0,0) (source) {};
\node[minimum height=0.75cm, anchor=north west] at ([cor]source.north east) (fs1) {};
\node[minimum height=0.75cm, anchor=south west] at ([cor]source.south east) (fs2) {};

%FIFO
\node[mn] at (3,0) (fifo) {\makecell[c]{Network\\Element\\$S$}};
\draw[->, line width=1pt, dashed] (source) -- (fifo);

%VAL AT SOURCE
\node[anchor=south west] at ([xshift=0.5cm, yshift=0.2cm] fs1.north east) (fs1val) {$r_f,b_f$};
\node[anchor=north west] at ([xshift=0.5cm, yshift=-0.2cm] fs2.south east) (fs2val) {$r_g,b_g$};
\draw[->, line width=0.5pt] (fs1val.west) -- ++(-0.2cm,0) -- (fs1.north east);
\draw[->, line width=0.5pt] (fs2val.west) -- ++(-0.2cm,0) -- (fs2.south east);

%PFR
\node[ns] at (6,-0.05) (p) {};
\node (spfr) at (6,0) {};
\node[ns] at ([cor]p.north east) (p) {};
\node[ns] at ([cor]p.north east) (p) {};
\node[ns] at ([cor]p.north east) (p) {};
\node[ns] at ([cor]p.north east) (p) {};
\node[ns] at ([cor]p.north east) (p) {};
\node[ns] at ([cor]p.north east) (p) {};
\node[ns] at ([cor]p.north east) (p) {};
\draw let \p1=(p.north east) in node at (\x1,0) (epfr) {};
%\node (epfr) at (p.north east) {};

\node[nn] at (6,0.05) (p) {};
\node[nn] at ([cor]p.south east) (p) {};
\node[nn] at ([cor]p.south east) (p) {};
\node[nn] at ([cor]p.south east) (p) {};
\node[nn] at ([cor]p.south east) (p) {};
\node[nn] at ([cor]p.south east) (p) {};
\node[nn] at ([cor]p.south east) (p) {};
\node[nn] at ([cor]p.south east) (p) {};

\node[hfil, cor, anchor=north west] at ([yshift=-0.05cm]epfr.center) (fpfr2) {};
\node[hfil, cor, anchor=south west] at ([yshift=0.05cm]epfr.center) (fpfr1) {};
\node[anchor=north east] at ([xshift=0.5cm,yshift=0.16cm]fpfr2.south west) {\makecell[r]{Per-flow\\regulator $R$}};
\node[anchor=south west] at ([xshift=0.5cm, yshift=0.2cm] fpfr1.north east) (fpfr1val) {$r_f,b_f$};
\node[anchor=north west] at ([xshift=0.5cm, yshift=-0.2cm] fpfr2.south east) (fpfr2val) {$r_g,b_g$};
\draw[->, line width=0.5pt] (fpfr1val.west) -- ++(-0.2cm,0) -- (fpfr1.north east);
\draw[->, line width=0.5pt] (fpfr2val.west) -- ++(-0.2cm,0) -- (fpfr2.south east);

\draw[->, line width=1pt, dashed] (fifo) -- (spfr.center);
\draw[->, line width=1pt, dashed] ([xshift=0.4cm]epfr.center) -- ++(0.5cm,0);

%FLOWS
\draw[->, rounded corners=0.1cm, dotted] ([yshift=0.2cm] fs1.east) -- ++(0.2cm,0) -- ([yshift=0.2cm, xshift=-0.5cm] fifo.west) node[above, pos=0.5] {$f$}-- ([yshift=0.2cm]fifo.west);
\draw[->, rounded corners=0.1cm, dotted] ([yshift=-0.2cm] fs2.east) -- ++(0.2cm,0) -- ([yshift=-0.2cm, xshift=-0.5cm] fifo.west) node[below, pos=0.5] {$g$}-- ([yshift=-0.2cm]fifo.west);

\draw[<-, rounded corners=0.1cm, dotted] ([yshift=0.45cm] spfr.center) -- ++(-0.2cm,0) -- ([yshift=0.2cm, xshift=0.5cm] fifo.east) node[pos=0.5] (targetf) {} -- ([yshift=0.2cm]fifo.east);
\draw[<-, rounded corners=0.1cm, dotted] ([yshift=-0.45cm] spfr.center) -- ++(-0.2cm,0) -- ([yshift=-0.2cm, xshift=0.5cm] fifo.east) node[pos=0.5] (targetg) {} -- ([yshift=-0.2cm]fifo.east);

\draw[->, dotted] ([yshift=0.45cm] fpfr1.south east) -- ++(0.5cm,0);
\draw[->, dotted] ([yshift=-0.45cm] fpfr1.south east) -- ++(0.5cm,0);

\node[anchor=south] at ([yshift=0.3cm] targetf) (inter1val) {$r_f,b^*_f>b_f$};
\node[anchor=north] at ([yshift=-0.3cm] targetg) (inter2val) {$r_g,b^*_g>b_g$};
\draw[->, line width=0.5pt] (inter1val.south) -- (targetf.center);
\draw[->, line width=0.5pt] (inter2val.north) -- (targetg.center);

%DELAY
\draw[<->] ([xshift=-0.2cm, yshift=0.2cm]fifo.north west) -- ([xshift=0.2cm, yshift=0.2cm] fifo.north east) node[above, pos=0.5] {$D^f$};
\draw[<->] let \p1=([xshift=-0.2cm, yshift=0.7cm]fifo.north west), \p2=([xshift=0.2cm] fpfr1.north east) in (\x1,\y1) -- (\x2,\y1) node[above, pos=0.5] {$D'^f = D^f$};
\draw[<->] ([xshift=-0.2cm, yshift=-0.2cm]fifo.south west) -- ([xshift=0.2cm, yshift=-0.2cm] fifo.south east) node[below, pos=0.5] {$D^g$};
\draw[<->] let \p3=([xshift=-0.2cm, yshift=-1cm]fifo.south west), \p4=([xshift=0.2cm] fpfr2.south east) in (\x3,\y3) -- (\x4,\y3) node[below, pos=0.5] {$D'^g = D^g$};
%\node{$D^{E+ \text{PFR}}_f = D^{E}_{f}$ and $D^{E+\text{PFR}}_g = D^{E}_{g}$};

\end{tikzpicture}} %
		\vspace{-0.2cm}
		\subcaption{\label{fig:prob:pfr}The \acl{PFR} uses one FIFO queue per flow (2~here). The shaping is for free for each individual flow.} %
	\end{minipage}\\\begin{minipage}{0.9\linewidth}\centering %
		%\vspace{-0.2cm}
		\resizebox{\linewidth}{!}{% !TeX spellcheck = en_US
% !TeX rootfile = article.tex
\begin{tikzpicture} %
\tikzstyle{mn} = [draw, minimum height=1.5cm]
\tikzstyle{n} = [draw, minimum height=1.5cm, anchor=west]
\tikzstyle{fil} = [draw, pattern=north east lines, text width=0.2cm]
\tikzstyle{hfil} = [fil, minimum height=0.75cm]
\tikzstyle{ffil} = [fil, minimum height=1.5cm]
\tikzstyle{cor} = [xshift=-\pgflinewidth]

%SOURCE
\node[minimum height=1.5cm] at (0,0) (source) {};
\node[minimum height=0.75cm, anchor=north west] at ([cor]source.north east) (fs1) {};
\node[minimum height=0.75cm, anchor=south west] at ([cor]source.south east) (fs2) {};

%FIFO
\node[mn] at (3,0) (fifo) {\makecell[c]{Network\\Element\\$S$}};
\draw[->, line width=1pt, dashed] (source) -- (fifo);

%VAL AT SOURCE
\node[anchor=south west] at ([xshift=0.5cm, yshift=0.2cm] fs1.north east) (fs1val) {$r_f,b_f$};
\node[anchor=north west] at ([xshift=0.5cm, yshift=-0.2cm] fs2.south east) (fs2val) {$r_g,b_g$};
\draw[->, line width=0.5pt] (fs1val.west) -- ++(-0.2cm,0) -- (fs1.north east);
\draw[->, line width=0.5pt] (fs2val.west) -- ++(-0.2cm,0) -- (fs2.south east);

%PFR
\node[n] at (6,0) (p) {};
\node (spfr) at (p.west) {};
\node[n] at ([cor]p.east) (p) {};
\node[n] at ([cor]p.east) (p) {};
\node[n] at ([cor]p.east) (p) {};
\node[n] at ([cor]p.east) (p) {};
\node[n] at ([cor]p.east) (p) {};
\node[n] at ([cor]p.east) (p) {};
\node[n] at ([cor]p.east) (p) {};
\node (epfr) at (p.east) {};

\node[ffil, cor,  anchor=west] at (epfr.center) (fpfr) {};
\node[anchor=north east] at ([xshift=0.5cm]fpfr.south west) {\makecell[r]{Interleaved\\regulator $R$}};
\node[anchor=south west] at ([xshift=0.5cm, yshift=0.2cm] fpfr.north east) (fpfrval) {	$\left\lbrace\begin{aligned}
		&r_f,b_f\\
		&r_g,b_g\\
	\end{aligned}\right.$};
\draw[->, line width=0.5pt] (fpfrval.west) -- ++(-0.2cm,0) -- (fpfr.north east);

\draw[->, line width=1pt, dashed] (fifo) -- (spfr.center);
\draw[->, line width=1pt, dashed] (fpfr1.south east) -- ++(0.5cm,0);

%FLOWS
\draw[->, rounded corners=0.1cm, dotted] ([yshift=0.2cm] fs1.east) -- ++(0.2cm,0) -- ([yshift=0.2cm, xshift=-0.5cm] fifo.west) node[above, pos=0.5] {$f$}-- ([yshift=0.2cm]fifo.west);
\draw[->, rounded corners=0.1cm, dotted] ([yshift=-0.2cm] fs2.east) -- ++(0.2cm,0) -- ([yshift=-0.2cm, xshift=-0.5cm] fifo.west) node[below, pos=0.5] {$g$}-- ([yshift=-0.2cm]fifo.west);

\draw[<-, rounded corners=0.1cm, dotted] ([yshift=0.1cm] spfr.center) -- ++(-0.2cm,0) -- ([yshift=0.5cm, xshift=0.5cm] fifo.east) node[pos=0.5] (targetf) {} -- ([yshift=0.5cm]fifo.east);
\draw[<-, rounded corners=0.1cm, dotted] ([yshift=-0.1cm] spfr.center) -- ++(-0.2cm,0) -- ([yshift=-0.5cm, xshift=0.5cm] fifo.east) node[pos=0.5] (targetg) {} -- ([yshift=-0.5cm]fifo.east);

\node[anchor=south] at ([yshift=0.3cm] targetf) (inter1val) {$r_f,b^*_f>b_f$};
\node[anchor=north] at ([yshift=-0.3cm] targetg) (inter2val) {$r_g,b^*_g>b_g$};
\draw[->, line width=0.5pt] (inter1val.south) -- (targetf.center);
\draw[->, line width=0.5pt] (inter2val.north) -- (targetg.center);

%DELAY
\draw[<->] ([xshift=-0.2cm, yshift=0.2cm]fifo.north west) -- ([xshift=0.2cm, yshift=0.2cm] fifo.north east) node[above, pos=0.5] {$D$};
\draw[<->] let \p1=([xshift=-0.2cm, yshift=0.7cm]fifo.north west), \p2=([xshift=0.2cm] fpfr1.north east) in (\x1,\y1) -- (\x2,\y1) node[above, pos=0.5] {$D'^f = D'^g = D$};

\end{tikzpicture}} %
		\vspace{-0.2cm}
		\subcaption{\label{fig:prob:ir}The \acl{IR} uses only one FIFO queue. If the network element $S$ is \acs{FIFO} for the aggregate $\lbrace f,g\rbrace$, then the shaping is for free for the aggregate and both flows have the same delay bound $D'^f=D'^g=D$ through $S$ and $R$. } %
	\end{minipage} %
	\caption{\label{fig:prob}Shaping-for free property of two type of regulators in a network with ideal clocks.}
\end{figure}
	%\vspace{-0.2cm}

% !TeX spellcheck = en_US
\section{Related work}
\label{sec:related}

The modeling of clock nonidealities benefits from a solid background in time metrology. In~\cite{ituDefinitionsTerminologySynchronization1996}, the \ac{ITU} defines fundamental notions and models for clocks used in synchronization networks. These models are further detailed in reference documents such as~\cite{ituTimingRequirementsSlave2004} for the \ac{ITU} and \cite{ieeeIEEEStandardDefinitions2009} for the \ac{IEEE}. Correspondingly, industrial requirements have been defined for clocks to be used in networks. They put constraints and bounds on the clock characteristics defined in the above documents. For instance, for a clock to be used in a synchronized \ac{TSN} network, it must meet the requirements of \cite[Annex B.1]{ieeeIEEEStandardLocal2011}.

Many technologies have been developed to perform the time-synchronization of a network. The most common are the use of an external time source such as a \acf{GNSS}~\cite{powersGPSGalileoUTC2004}, and the use of time-synchronization protocols such as \ac{NTP}~\cite{rfc5905}, \ac{PTP}~\cite{ieeeIEEEStandardPrecision2008}, generalized \ac{PTP}~\cite{ieeeIEEEStandardLocal2011} and WhiteRabbit~\cite{moreiraWhiteRabbitSubnanosecond2009}. Other technologies are tailored for wireless sensor networks~\cite{wuClusterBasedConsensusTime2015,liuJointTimeSynchronization2016}. Each comes with various performance analyses: we can cite~\cite{murtaQRPp14CharacterizingQuality2006} for \ac{NTP}, \cite{dierikxWhiteRabbitPrecision2016} for WhiteRabbit. However, the design of a ``good'' time-synchronization protocol remains an open issue~\cite{frerisFundamentalLimitsSynchronizing2011,ridouxjulienPrinciplesRobustTiming2010}, and each protocol proposition adds to the time-metrology domain by identifying limits of previous protocols~\cite{gengExploitingNaturalNetwork2018a,veitchRobustSynchronizationSoftware2004}.

The analyses mention that the precision of time-synchronization protocols depends on the latency and jitter of synchronization messages and of control data. The latency and jitter bounds of time-sensitive networks were studied in numerous occasions using network calculus. For \ac{TSN}, we can cite~\cite{zhaoComparisonTimeSensitive2018,bouillard2018deterministic,specht2016urgency}. TSN provides many building blocks to provide guaranteed delay bounds. The use of regulators with \acf{ATS} is one of them. Regulators have been studied in~\cite{mohammadpourLatencyBacklogBounds2018,specht2016urgency,wandelerPerformanceAnalysisGreedy2006a}. Other building blocks include \ac{CQF} \cite{IEEEISOIEC2019a}, \ac{CBS} \cite{IEEEStandardLocal2010} or \ac{TAS} \cite{iso-iec-ieeeISOIECIEEE2018}. Choosing the best set of building blocks for a specific network is an open question~\cite{navetUsingMachineLearningto2019}, and several studies have compared their performance~\cite{nasrallahPerformanceComparisonIEEE2019}. In this paper, we focus on regulators, thus on \ac{ATS}.   

Interestingly, the reciprocal aspect, the effect of the clock and synchronization nonidealities on the network performances, appears to be much less studied. For example, the above-mentioned studies always assume that clocks are perfect in the network or that time distribution is perfect~\cite{nasrallahPerformanceComparisonIEEE2019}. Even~\ac{ns-3}~\cite{nsnamNs3NetworkSimulator2011} has a unique time-base for simulating network events, and the simulation of clock behavior and of time-synchronization protocols requires work-arounds such as the one in \cite[Section III]{maruyamaNS3BasedIEEE2015}. In this paper, we assess the effect of clock nonidealities on the performance bounds of regulators (\ac{ATS}) and we leave to future work the evaluation of their effect on other TSN building blocks.

In \cite[Section 7.1]{bergerRelevanceAdversarialQueueing2014}, the authors consider the nonidealities of clocks to show that the adversarial traffic generation described in \cite{bhattacharjeeInstabilityFIFOArbitrarily2005a,andrewsInstabilityFIFOPermanent2009} does not induce unbounded latencies under realistic network assumptions. However, their traffic model is limited to non-bursty flows \cite[Section 3]{bergerRelevanceAdversarialQueueing2014} and their results are obtained through simulations. In this paper, we focus on obtaining upper-bounds on worst-case latencies, and finding the worst-case using simulations is known to be an intractable issue~\cite{bouillard2018deterministic,phanComposingFunctionalStateBased2007}.

A seminal work on applying network calculus on networks with nonideal clocks for obtaining worst-case upper-bounds lies in \cite{daigmorteTraversalTimeWeakly2016,daigmorteEvaluationAdmissibleCAN2017, daigmorteReducingCANLatencies2017}. The authors are interested in a bandwidth management method that spreads the time at which frames are scheduled on the \ac{CAN} bus. Such bandwidth management uses offsets between the time instants at which frames are scheduled. They note that such scheduling of the frames across different nodes would require a time-synchronization mechanism between the network nodes. However, they show that a ``weak synchronization'' of the nodes (with a $1$ms precision bound) already provides significant performance improvements~\cite[Section 5.3]{daigmorteTraversalTimeWeakly2016}. Their time-model~\cite[Section 3.1]{daigmorteTraversalTimeWeakly2016} is hence limited to synchronized networks (including loosely-synchronized networks). It doesn't consider bounds on the clock frequency offset, but only on the clock time-error bound (called the ``phase bound'' in their paper).
In the present paper, we are interested in both synchronized and non-synchronized networks, and we show that taking into consideration the bounds on the frequency offsets helps having tight delay bounds. Last, the authors of \cite{daigmorteTraversalTimeWeakly2016,daigmorteEvaluationAdmissibleCAN2017, daigmorteReducingCANLatencies2017} compute tight arrival curves, service curves and latency bounds of periodic flows on a \ac{CAN} bus that employs the above-mentioned offset management method~\cite[Section 4.3.1]{daigmorteTraversalTimeWeakly2016}. Then they adapt the results to take into account the phase bounds (the time-error bounds)~\cite[Section 4.3.2]{daigmorteTraversalTimeWeakly2016}. As a consequence, their analysis of clock nonidealities is limited to a specific network model, with specific service and arrival curves. In this paper, we are interested in computing the effect of clock nonidealities given any service curve, arrival curve or latency bound.

The present paper is also motivated by discussions with industrial partners, specifically in the context of \ac{TSN}. The ongoing draft for \ac{ATS} mentions the possible consequences of clock nonidealities when deploying regulators \cite[Annex V.8]{ieeeDraftStandardLocal2019a} and proposes some solutions that would benefit from theoretical foundations, as proposed in this paper.

% !TeX spellcheck = en_US
\vspace{-0cm}
\section{Time-Sensitive Networks With Regulators}
\label{sec:prob}

Here, we provide some background that is required by the rest of the paper. In time-sensitive networks, delays at network elements have to be bounded in worst case, not in average. To this end, network calculus is often used \cite{cruzCalculusNetworkDelay1991,cruzCalculusNetworkDelay1991a,leboudecNetworkCalculusTheory2001,changPerformanceGuaranteesCommunication2000}. This framework uses cumulative functions such as $A(t)$, where $A(t)$ is the total number of bits observed at some observation point between an arbitrary time reference $0$ and time $t$. Traffic flows are assumed to be bounded by arrival curve constraints, namely, constraints of the form: $\forall t\geq s\geq 0, R(t)-R(s)\leq  \alpha(t-s)$ (the function $\alpha$ is called ``arrival curve''). A frequently used function  is $\gamma_{r,b}$ defined by $\gamma_{r,b}(t)=rt+b$ for $t>0$ and $\gamma_{r,b}(t)=0$ for $t\leq 0$. It corresponds to a flow that is limited to a rate $r$ and a burst $b$ (``leaky-bucket'' arrival curve).

The service offered by a network element is also assumed to be lower bounded by a condition of the form
$\forall t\geq0: D(t)\geq (A\otimes \beta)(t)$ where $A$ [resp. $D$] is the input [resp. output] cumulative function, the function $\beta$ is called ``service curve'' and the symbol $\otimes$ is the min plus convolution, such that $(A\otimes \beta)(t)= \inf_{s\geq 0}\left(A(s)+\beta(t-s)\right)$. By Reich's formula~\cite{norros1994storage}, a network element reduced to a single server queue with output rate $R$ offers the service curve $\beta(t)=Rt$. If in addition the server can take vacations for durations upper bounded by $T$ per busy period, the system offers a service curve $\lambda_{R,T}(t)=\left|R(t-T)\right|^+$($=\max(0,R(t-T))$), called ``rate-latency'' service curve. A FIFO network element that guarantees a delay upper bounded by $D$ offers the service curve $\delta_D$ defined by $\delta_D(t)=0$ for $t\leq D$ and $\delta_D(t)=+\infty$ for $t>D$. The concatenation of network elements that each offers a service curve also offers a service curve equal to the min-plus convolution of service curves. Many schedulers, such as Deficit Round Robin or the Credit Based Shaper of IEEE TSN, are characterized by rate-latency service curves
\cite{boyer2012deficit,daigmorte2018modelling}.

Classic network calculus results give delay and backlog bounds at a network element, given some arrival-curve and service-curve constraints. They also give bounds on the burst of the output, i.e., arrival curves for the output flows \cite{leboudecNetworkCalculusTheory2001,mohammadpour2019improved}.

Time-sensitive networks can be per-flow networks or per-class networks. In the former case, schedulers are per-flow, e.g., there is one Deficit Round Robin queue per flow. It follows that service-curve properties apply to individual flows. In contrast, in class-based networks, schedulers offer a service guarantee (captured by a service curve) to the aggregate of all flows that belong to one specific class; inside a class, the scheduler is FIFO. Providing delay bounds in class-based network is more complicated than in per-flow networks. In particular, we need to compute good bounds on the burst increase that can occur at every hop \cite{charnyDelayBoundsNetwork2000}. Practical solutions for complex class-based networks almost all require that flows are re-shaped individually inside the network.

Flow shaping (or re-shaping) is performed by regulators, that are either per-flow (\ac{PFR}) or interleaved (\ac{IR}). A PFR, configured with arrival curve $\sigma$ for flow $f$, makes sure that its output satisfies the arrival curve constraint $\sigma$ (also called ``shaping curve''). If the input data of flow $f$ arrives too fast, the packets are stored in the \ac{PFR} buffer (with one FIFO queue per flow), until the earliest time when it is possible to release the packet without violating the arrival curve constraint. 

Note that a regulator, as defined above, controls the arrival curve of any individual flow at its output. As a consequence, \acf{ATS}~\cite{ieeeDraftStandardLocal2019a} (the standardization of \ac{PFR} and \ac{IR}) differs from other types of shapers that act on a per-class basis. For instance, \acf{TAS}~\cite{iso-iec-ieeeISOIECIEEE2018} is a scheduler that controls, using gates, when classes may access the link, depending on a global distributed schedule. It cannot control the arrival curve of a single flow that continues to suffer from the contention with other flows of the same class. It is also worth noting that, as regulators only need to measure elapsed time, they are insensitive to constant time offsets. On the contrary, some schedulers such as \ac{TAS} require the nodes to be time-synchronized.

The input-output characterization of a \ac{PFR} with concave shaping curve $\sigma$  is well understood:  it can be modelled as the sequence of two virtual systems, a fluid greedy shaper followed by a packetizer~\cite{le2002some}. The fluid greedy shaper is similar to the PFR but operates on individual bits: it releases fractions of a packet as soon as possible. The packetizer receives the bit-by-bit output of the former and stores it until a full packet can be released. The former is a min-plus linear system, characterized by the relation $D(t)=(A\otimes\sigma)(t)$ where $A$ [resp. $D$] is the input [resp. output] cumulative function of the fluid greedy shaper. The latter does not increase the per-packet delay bound. It can be ignored in latency calculations \cite{chang1998general}. In particular, a PFR with concave shaping curve $\sigma$ is a network element that offers $\sigma$ as service curve.

An IR is similar to a \ac{PFR} with one large difference. All packets of all flows are stored in a single FIFO queue, the packet at the head of the queue is released at the earliest time when it is possible without violating the arrival curve constraint for this flow, and packets of other flows wait until they appear at the head of the queue~\cite{specht2016urgency}. \ac{PFR}s and \ac{IR}s have state information per flow but \ac{IR}s have a single FIFO queue thus are preferred in the context of IEEE TSN, which is typically per-class \cite{ieeeDraftStandardLocal2019a}. Unlike \ac{PFR}s, no service curve characterization appears to be known for \ac{IR}s.

When a flow is served in a network element, its burst typically increases, and the increase is large when the flow shares a network element with many bursty flows in the same class. Worse, the more bursty the competing flows are, the larger the burst increase is. The increased burst leads to larger delays and backlogs in downstream nodes, which is the ingredient for a cascade effect and can even lead to instability when there are cyclic dependencies \cite{charnyDelayBoundsNetwork2000, andrewsInstabilityFIFOPermanent2009}. This does not occur if flows are reshaped by regulators at some or every node \cite{thomasCyclicDependenciesRegulators2019}. In contrast, if regulators are used, the burst of the output is known and can be imposed to be the same as at the source, which enables us to find good delay bounds.

The regulator, however, is itself a queuing system and its impact on delay should be accounted for.
Here, an essential property of regulators is used, called ``shaping-for-free''. For a \ac{PFR}, it can be stated as follows \cite[Theorem~3]{leboudecTheoryTrafficRegulators2018}. Consider flows, such as $f$ on Figure~\ref{fig:prob:pfr}, which satisfy arrival-curve constraints $\sigma_f$ (one per flow) and are served in a network element $S$, which is FIFO for packets inside every flow; after $S$, the flows are processed by a \ac{PFR} with the same arrival curves (i.e., $\sigma_f$ for flow $f$). Then, for  every flow $f$, the worst-case delay, $D^{f}$, for any packet of the flow through $S$ is equal to the worst-case delay, $D'^{f}$ through $S$ and the \ac{PFR}. In other words, re-shaping does not increase the worst-case delay of the previous hop (but reduces the worst-case delay at the next hop). If the \ac{PFR} is replaced by an \ac{IR}, there is a similar result for the worst-case delay over all flows that are processed by the \ac{IR} (Figure~\ref{fig:prob:ir}), assuming $S$ is FIFO for all packets of all flows inside a given class \cite[Theorem~5]{leboudecTheoryTrafficRegulators2018}. The shaping-for-free property is established assuming all clocks are ideal.

In reality, the clocks at the source and at different \ac{PFR}s or \ac{IR}s along the path of a flow are likely to be different. We say that the regulators in a time-sensitive network are ``non-adapted'' if we ignore the clock deviations and apply the shaping-for-free property. As we show in Sections~\ref{sec:manag-async} and \ref{sec:manag-sync}, this can lead to severe problems, in both non-synchronized and synchronized networks.

% !TeX spellcheck = en_US
% !TeX rootfile = article.tex
\section{System model}\label{sec:sys-model}
	We first propose a framework for modeling a clock within a device. We derive this framework for non-synchronized networks and synchronized networks. We show how the clocks requirements within \ac{TSN} are easily captured under our time model. Finally, we detail the network model under consideration. Notations for the whole paper are available in Table~\ref{tab:notations}.

% !TeX spellcheck = en_US
% !TeX rootfile = article.tex
\begin{table} %
	\caption{\label{tab:notations}Notation} %
	\resizebox{\linewidth}{!}{
	\begin{tabular}{l|l}
		$t$ & The measure of a time instant \\
		$\mathcal{H}_i$ & The clock of a device or the true time (\acs{TAI})\\
		$d_{g\rightarrow i}$ & The relative time-function between $\h_g$ and $\h_i$\\
		$T_{\text{start}}$ & When any of the clocks shows $T_{\text{start}}$, all other clocks have \\ & positive values and no device has sent any bit yet.\\
TAI & International Atomic Time (true time)\\
		\hline
		$\eta$ & The timing-jitter bound\\
		$\rho$ & The clock-stability bound\\
		$\Delta$ & The time-error bound in a synchronized network\\
		\hline
		$R$ [resp $R^{*}$]& The cumulative arrival function of a flow\\ & at the input [resp output] of a device \\
		$\alpha$ [resp $\alpha^{*}$] &  An arrival curve for a flow \\ & at the input [resp output] of a device\\
		$\beta$ & The service provided by a system (series of devices) to a flow \\
		$D'_k$ [resp $D_k$] & An upper-bound on the delay of a flow through its $k$-th hop \\ & [resp excluding the regulator of the $k$-th hop] \\
		\hline
		$R^{\h_i}, \alpha^{\h_i},$ & We use the super-script to denote the clock used to observe \\ $\beta^{\h_i}, D^{\h_i} $& one of the previous notions. \\ & Example: $\alpha^{\h_i}$ is the arrival curve as observed with clock $\h_i$.\\
		\hline
		$\gamma_{r,b}(t)=rt+b, t>0$ & Leaky-bucket arrival curve of rate $r$, burst $b$ \\
\hspace{.95cm}$= 0, t\leq 0$&\\
		$\delta_D(t)=+\infty, t>D$ & Service curve of a variable $D$-bounded delay  \\
\hspace{.95cm}$= 0, t\leq D$&\\
		$\lambda_{R,T}(t)=\left|R(t-T)\right|^+$ & Rate-latency service curve of rate $R$ and latency $T$\\
        \hline
		$\mathcal{R}_k(r)$ [resp $\mathcal{Q}_k(b)$] & The smallest rate $\geq r$ [resp burst $\geq b$] that can be implemented \\ & by regulator $\text{Reg}_k$.\\ 
		\hline
		$|x|^+$ & $\max(0,x)$\\
		$a\wedge b$ [resp $a\vee b$] & $\min(a,b)$ [resp $\max(a,b)$] \\
	\end{tabular}}
\end{table}
 %
% !TeX spellcheck = en_US
% !TeX rootfile = article.tex
\subsection{General Time-Model}
\label{sec:time-model}
	
	We denote with $\hze$ the true time, i.e. the \ac{TAI}. We assume that it represents a continuous quantity. When reading the time indicated by a clock $\mathcal{H}_i$ in the network, only a subset of values are readable, due to the precision of the clock logic. We assume that this clock logic enforces the accessible values to increase when the clock is read over the course of the true time. %The red curve in Figure~\ref{fig:general-time} represents the accessible values of the clock $\h_i$ as a function of the true time.
	With this assumption, and since the clock output is observable only at discrete time instants, it is possible to find a continuous, strictly increasing function $h_i(t)$ of the true time, which returns the value that the clock would display at true time $t$ if it had infinite precision (Figure~\ref{fig:general-time}). Accessing the values of clock $\h_i$ corresponds to sampling the function $h_i$ at the discrete time instants where the clock logic does a transition.
	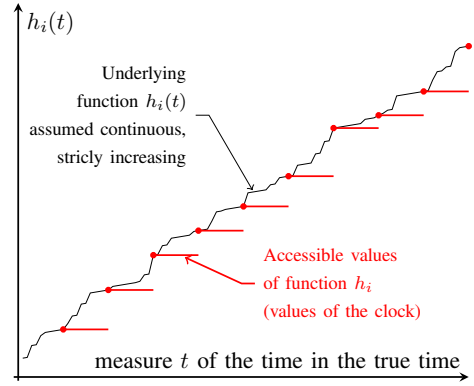
\begin{figure}\centering %
		\resizebox{0.7\linewidth}{!}{% !TeX spellcheck = en_US
% !TeX rootfile = article.tex
\begin{tikzpicture}[samples=200] %

	\pgfplotsset{ticks=none}

	\begin{axis}[xlabel=measure $t$ of the time in the true time, ylabel=$h_i(t)$,xmin=-0.1,xmax=10,ymin=-0.1,ymax=10, axis lines=center] %, width=7cm, height=6cm]
	
	\addplot[no marks, line width=0.4pt] table {general1.dat};
	\addplot[jump mark left, red, mark=*, mark size=1pt] table {general2.dat};

	\node at (axis cs:2,7) (legendTrue) {\makecell[r]{\footnotesize{Underlying}\\\footnotesize{ function $h_i(t)$}\\\footnotesize{assumed continuous,}\\\footnotesize{stricly increasing}}};
	\node[red] at (axis cs:7.2,2.5) (legendAccess) {\makecell[l]{\footnotesize{Accessible values}\\\footnotesize{of  function $h_i$}\\\footnotesize{(values of the clock)}}};
	
	\draw[->, line width=0.4pt] (legendTrue.east) -- ++(0.2cm,0) -- (axis cs:5.2,5);
	\draw[->, red] (legendAccess.west) -- ++(-0.2cm,0) -- (axis cs:3.7,3.2);
	\end{axis} %
\end{tikzpicture}}%
		\caption{\label{fig:general-time} Function $h_i$ (Section \ref{sec:time-model}). Only a subset of values are readable. We assume that the underlying function $h_i$ is continuous, strictly increasing.} %
	\end{figure} %
	This allows us to introduce the relative time function, which will be useful in Section~\ref{sec:toolbox}, as follows. %Many results of the toolbox presented in Section~\ref{sec:toolbox} require the notion of the relative time function that we define as follows:
	\begin{definition}[Relative time function $d_{g\rightarrow i}$]\label{def:relative-time-function}
		For any two clocks $\h_g,\h_i$ we define the relative time function from $\h_g$ to $\h_i$ as $d_{g\rightarrow i} = h_i\circ h_g^{-1}$, where $\circ$ represents the composition of functions.
	\end{definition}
	
	$d_{g\rightarrow i}(t)$ is the value that clock $\h_i$ would have when clock $\h_g$ shows time $t$ if they both had infinite precision. Note that $h_i$ and $h_g$ are both continuous, strictly increasing. Hence, $d_{g\rightarrow i}$ is also continuous, strictly increasing. We note $d_{g\rightarrow i}^{-1}$ its inverse. Obviously, $d_{g \rightarrow i}^{-1}$ is also equal to $d_{i \rightarrow g}$.
	
	The time-metrology literature uses the time-error function~\cite[Section 4.5.13]{ituDefinitionsTerminologySynchronization1996}, which is equal to the time difference between two clocks $\h_i$ and $\h_g$. With our notation, the time-error function between $\h_g$ and $\h_i$ is $d_{g\rightarrow i}(t) - t$ for any $t$ measured with $\h_g$.
	
	For two clocks $\h_g,\h_i$, $d_{g\rightarrow i}(0)$ is the value of the clock $\h_i$ when the clock $\h_g$ shows zero. We denote by $T_{\text{start}}$ the maximum of this value for any pair $(\h_g,\h_i)$, where each of $\h_g,\h_i$ represents a clock in the network or the \ac{TAI}. We assume that, for any pair of clocks, no device transmits any bit in the network until $d_{g \rightarrow i}$ reaches $T_{\text{start}}$. We believe this is not a limiting assumption because we can assume that all devices start with a rough estimation of the true time, hence of each other. Consequently, $T_{\text{start}}$ could be in the magnitude of hours or days, whereas the origin of time on network devices usually refers to several years in the past.
	
	In this paper, we consider two time-models: %
	
	$\bullet$ The non-synchronized time-model: Clocks are free-running and do not interact with each other, but constraints on their stability can be formulated.
	
	$\bullet$ The synchronized time-model: In addition to the stability requirements on the clocks, a time-synchronization algorithm (such as NTP or PTP), or an external time source (such as a \ac{GNSS}) is employed. It distributes a time reference to all devices so that their local time matches with each other within a specified bound.

\subsection{The Non-Synchronized Time-Model}
\label{sec:time-model:async}

	We consider a clock $\h_i$ in the network. We first assume that this clock does not interact with any other. This corresponds to the \emph{free-running mode} defined in~\cite[Section 4.4.1]{ituDefinitionsTerminologySynchronization1996}. We assume that, when compared to the true time $\hze$, the clock behaves as per the time-error model provided in \cite[Annex I.3]{ituDefinitionsTerminologySynchronization1996}, reported below:
	\begin{equation}
		h_i(t) - t = x_{0,i} + t y_{0,i}(T)  + \frac{D_{i}}{2} t^2 + \psi_i(t)
	\end{equation}
	With $x_{0,i}$ the initial time offset of $\h_i$ (relative to the true time), $y_{0,i}(T)$ the frequency offset at constant temperature $T$, defined relative to the true frequency of 1 second per second, $D_{i}$ the average frequency drift of clock $\h_i$ caused by its aging, and $\psi_i$ is a random noise component.

	As the clock is free-running, $x_{0,i}$ can take any value. Hence, it is impossible to put constraints on the time-error function itself but we can constrain its evolution. Take $s<t$, then
	\begin{equation}\label{eq:error-model-ev}
		h_i(t) - h_i(s) - (t-s) = (t-s)y_{0,i}(T) + \frac{D_i}{2}(t^2-s^2) + \psi_i(t) - \psi_i(s)
	\end{equation}
	
	The first term is linear with $(t-s)$ and depends on $y_{0,i}(T)$.
	
	\textbf{Remark:} Some time-metrology studies, including~\cite{gengExploitingNaturalNetwork2018a}, denote with ``clock drift'' the linear evolution of the time-error function, \emph{i.e.} $y_{0,i}(T)$. However, we decide here to remain consistent with the definitions of~\cite{ituDefinitionsTerminologySynchronization1996} that defines the \emph{frequency offset} as the linear evolution of the time-error function and the \emph{frequency drift} as the second-order evolution of the time-error function.
	
	In industrial requirements, the frequency offset is usually bounded by a value that depends on the temperature conditions. We note $y_{\max,i}(T)$ the bound on the frequency offset of clock $\h_i$ at constant temperature $T$ and $\rho_{1,i}=\max_{\lbrace T\in\mathcal{T}\rbrace}y_{\max,i}(T)$ its highest value over the whole range of temperatures $T\in\mathcal{T}$ that the network is expected to encounter. It corresponds to term $a_1+a_2$ in \cite[Section 11.2.1]{ituTimingRequirementsSlave2004}.
	
	The second term is in second order with $t-s$ and depends on the relative aging of the clocks.
	
	\textbf{Remark:} \cite{ituDefinitionsTerminologySynchronization1996} defines two main contributors to the frequency drift, \emph{i.e.} to the second-order evolution of the time-error function \cite[Section 4.5.4]{ituDefinitionsTerminologySynchronization1996}: the aging of the clock and the external effects, the latter being dominated by the effect of the evolution of the temperature on the clock frequency. Hence, when the temperature depends on the time, which is the case for any real-life single observation such as the one presented in Figure~7 of~\cite{gengExploitingNaturalNetwork2018a}, then the term $t\cdot y_{0,i}(T(t))$ has indeed a second-order (or even higher-order) term that depends on $\frac{dT}{dt}$. However, for a network model, it is generally impossible to predict the evolution of the temperature, \emph{i.e.} the function $T(t)$. Upper-bounding the frequency offset $y_{0,i}(T(t))$ by $\rho_{1,i}=\max_{\lbrace T\in\mathcal{T}\rbrace}y_{\max,i}(T)$ ensures that the model remains conservative for any possible evolution of the temperature. This is also consistent both with \emph{a)} the \ac{ITU} specifications \cite[Section 11.2.1]{ituTimingRequirementsSlave2004} which increase the linear frequency offset by a factor, $a_2$ in \cite[Section 11.2.1]{ituTimingRequirementsSlave2004}, that accounts for the temperature changes, and keep the aging as the only remaining second-order contributor, and \emph{b)} with the \ac{IEEE} specifications, which define a frequency offset bound independent from the temperature~\cite[Annex B.1.1]{ieeeIEEEStandardLocal2011}. As a consequence, the second-order evolution of the time-error function in our model only depends on the aging of the clock, the term $D_it^2/2$. It depends on the constant aging coefficient $D_i$ of clock $\h_i$.
	
	In industrial requirements, such as \ac{TSN}~\cite{ieeeIEEEStandardLocal2011}, this term is often neglected. To verify the assumption that the aging is negligible, we compare the linear coefficient due to the frequency offset, $\rho_{1,i}$, with the linear coefficient due to the aging, $D_i\frac{t+s}{2}$. The following numerical application verify that we can neglect $D_i$.
		
	\ant Call $L$ the order of magnitude of the lifetime of a network. Then the aging coefficient over $L$ is bounded by $D_i\cdot L$. If we take for $\rho_{1,i}$ the value specified in \cite[Annex B.1.1]{ieeeIEEEStandardLocal2011} and for $D_i$ (not specified in \ac{TSN}) the largest value specified in \cite[Table 24]{ituTimingRequirementsSlave2004}, we obtain $L=D_i/\rho_{1,i}\approx10^{-4}/10^{-14}=10^{10}s$. For the aging to be noticeable, compared to the acceptable frequency offset, the network shall be in operation for more than 300years.
	
	The last term of Equation~(\ref{eq:error-model-ev}) is made of noise and is further detailed in \cite[Section 8]{ituTimingRequirementsSlave2004}. It has two components. The former is the timing jitter \cite[Section 4.1.12]{ituDefinitionsTerminologySynchronization1996}, a high-frequency signal. It is usually constrained by a peak-to-peak jitter bound $\eta_i$ \cite[Annex B.1.3.1]{ieeeIEEEStandardLocal2011}.
	
	\textbf{Remark:} Due to its stochastic nature, the probability for a clock to present a jitter higher than the specified jitter bound, or to have a frequency offset higher than the frequency offset bound cannot be 0 \cite[Note of Section 8.3]{ituTimingRequirementsSlave2004}. However, providing bounds with a zero probability of excursion would not be possible, neither would it lead to deterministic delay bounds that are at the core of working groups such as \ac{TSN} or Detnet. We here assume that we can find bounds on the jitter and on the frequency offset with a fulfillment probability high enough so as to keep the probability of an excursion over the lifetime of a network negligible. Then, any clock that does not behave as per the specified bounds is considered as faulty, and the probability of a faulty clock, together with its probable consequences, can be studied using the best practices of the safety-analysis domain~\cite{ecssECSSQST3002CFailureModes2009,ecsssECSSQST4002CHazardAnalysis2008}. This assumption is also consistent with industrial requirements such as \ac{TSN} \cite[Annex B.1]{ieeeIEEEStandardLocal2011} that typically provides bounds on the jitter and on the frequency offset without specifying the fulfillment probability - hence considering any excursion as a failure.
	
	\ant In the \ac{TSN} requirements~\cite[Annex B.1.3.1]{ieeeIEEEStandardLocal2011}, the jitter of any clock $\h_i$ shall not exceed 2ns peak-to-peak, that is $\eta_i=\text{2ns}$ for any $\h_i$ in the network.
	
	The last noise component, the wander~\cite[Section 4.1.15]{ituDefinitionsTerminologySynchronization1996}, is a low-frequency noise signal. As opposed to the jitter or the frequency offset, it is usually constrained using the \ac{TDEV}, a statistical metric. The \ac{TDEV} of clock $\h_i$, $\text{dev}_i(t-s)$ is an upper-bound on the deviation of the time-error function over an observation period $t-s$. With the same argumentation as before, we assume that we can find $m_i\in\mathbb{R}$ such that the probability of the time-error function to present a wander over $t-s$ higher than $m_i\cdot\text{dev}_i(t-s)$ is negligible over the lifetime of the network, and that such situation can be considered as a failure. Note that the order of magnitude of $m_i$ would typically be around 10 for a normal distribution, as such multiples of the deviation already achieve a very high fulfillment probability.
	
	In the majority of technical requirements, such as for \ac{TSN}~\cite[Annex B.1]{ieeeIEEEStandardLocal2011}, the \ac{TDEV} is in the form $\text{dev}_i(t-s)=(t-s)c_i$, with $c_i$ a constant. In some cases, it can even be sub-linear~\cite[Section 8.1]{ituTimingRequirementsSlave2004} or negligible~\cite[Section 8.]{ituTimingRequirementsSlave2004}. To remain conservative, we consider the linear form and we define $\rho_{2,i}=m_i\cdot c_i$. Hence, the wander of Equation~(\ref{eq:error-model-ev}) is upper-bounded by $\rho_{2,i}\cdot(t-s)$.
	
	We define the stability of clock $\h_i$ as $\rho_i = 1+\rho_{1,i}+\rho_{2,i}$. As a consequence, the linear coefficient of Equation~(\ref{eq:error-model-ev}) is bounded by $\rho_i-1$. It is worth noting that in general, one of $\rho_{1,i}$,$\rho_{2,i}$ is negligible with respect to the other. For example, when the non-synchronized clock $\h_i$ uses the phase-locking mechanism~\cite[Section 4.4.4]{ituDefinitionsTerminologySynchronization1996} (also called syntonization or frequency synchronization \cite{recommendation20068261}) with a near-perfect clock representing $\hze$, then $\rho_{1,i}$ is null~\cite[Appendix I.3]{ituDefinitionsTerminologySynchronization1996}. When phase-locking mechanisms are not used, for instance in \ac{TSN}, then $\rho_{1,i}$ is usually much higher than $c_i$, and with our previous remark on the value of $m_i$, $\rho_{1,i}$ would typically remain much higher than $\rho_{2,i}$.
	
	\ant In the \ac{TSN} requirements, for any clock $\h_i$ in the network, we shall have: $\rho_{1,i}=100$ppm~\cite[Annex N.1.1]{ieeeIEEEStandardLocal2011} and $c_i=5\cdot10^{-9}$~\cite[Table B.1]{ieeeIEEEStandardLocal2011}. Hence, even with a margin $m_i$ of hundred times the deviation, $\rho_{2,i}$ remains much smaller than $\rho_{1,i}$ and we obtain that any clock $\h_i$ has the stability $\rho_i=1+1\cdot 10^{-4}$.

	From the above considerations, Equation~(\ref{eq:error-model-ev}) has (1) a jitter, high-frequency term, constrained by $\eta_i$, (2) a linear term, bounded by $\rho_i-1$, and (3) no higher-level terms. We can upper-bound: 
	\begin{equation}\label{eq:bounds-single-clock}
		\forall t\ge s, h_i(t) - h_i(s) \le (t-s)\rho_i + \eta_i
	\end{equation}
	We now lower-bound the evolution by changing $h_i$ for $h_i^{-1}$ and we obtain:
	\begin{equation}\label{eq:bounds-single-clock-bis}
	\forall t\ge s, (t-s-\eta_i) \frac{1}{\rho_i} \le h_i(t) - h_i(s) \le (t-s)\rho_i + \eta_i
	\end{equation}
	
	Let us now consider a pair of clocks $(\h_i,\h_g)$. We are interested in bounding the evolution of the relative time between $\h_g$ and $\h_i$, that is $d_{g\rightarrow i}(t)$ for $t$, the measure of a time instant with $\h_g$. We obtain, for $t \ge s$:
	\begin{align*}
		d_{g\rightarrow i}(t)-d_{g\rightarrow i}(s) &= h_i(h_g^{-1}(t)) - h_i(h_g^{-1}(s))\\
		&\le (h_g^{-1}(t) - h_g^{-1}(s))\rho_i + \eta_i\\
		&\le (t-s)\rho_i\rho_g + \eta_g \rho_i + \eta_i
	\end{align*}
	
	For a given network, we define the \emph{clock stability bound} of the network as $\rho=\max_{\lbrace \h_i,\h_g\rbrace}(\rho_i\rho_g)$ and the \emph{timing-jitter bound} of the network as $\eta=\max_{\lbrace \h_i,\h_g\rbrace}(\eta_g \rho_i + \eta_i)$. Then any pair of clocks $(\h_i,\h_g)$ in the network verify $d_{g\rightarrow i}(t)-d_{g\rightarrow i}(s) \le \rho (t-s) + \eta$.	The lower bound is obtained by symmetry, by flipping the $\h_g$ and $\h_i$ roles.
	
	We have finally obtained the following model: for any pair $(\h_g,\h_i)$, $\forall t \ge s$,
	\begin{equation}\label{eq:constr-async}
		\frac{1}{\rho} (t-s - \eta) \le d(t) - d(s) \le \rho (t-s) + \eta
	\end{equation}
	where $d = d_{g\rightarrow i}$. Note that $\rho$ and $\eta$ do not depend on $\h_g,\h_i$.
	
	\begin{figure}\centering\begin{minipage}{0.48\linewidth}\centering
		\resizebox{\linewidth}{!}{	\begin{tikzpicture}[samples=200]
		
	\pgfplotsset{ticks=none}
	
	\begin{axis}[xlabel=$t \text{ observed with }\mathcal{H}_{g}$, ylabel=$d(t) \text{ observed with }\mathcal{H}_{i}$,
	xmin=-3.5,xmax=20,ymin=-2,ymax=25, axis lines=center] %, width=6.2cm, height=5cm]

	\addplot[no marks, domain=0:20, dotted] {x+2} node[pos=0.6, sloped, above] {Slope 1};
	%\addplot[no marks, domain=0:20, dashed] {x}; % node[pos=0.8, sloped, above] {$d(t)=t$};
	\addplot[no marks, domain=3:20, red] {1.5*x+3.5} node[pos=0.3, sloped, above] {Slope $\rho$};
	\addplot[no marks, dotted, domain=3:6, red] {0.5*x+2};
	\draw[red] (axis cs:3,5) -- (axis cs:3,8);
	\addplot[no marks, domain=6:20, red] {0.5*x+2} node[pos=0.6, sloped, below] {Slope $1/\rho$};
	\addplot[no marks, domain=3:6, red] {0+5};
	\addplot[no marks] table[y expr=\thisrowno{1}+2] {asyncdat.dat};
	
	\draw[line width=0.3pt] (axis cs:3,0) -- (axis cs:3,5) node[pos=0, below] {$s$};
	\draw[line width=0.3pt] (axis cs:0,5) -- (axis cs:3,5) node[pos=0, left] {$d(s)$};

	%\draw[dotted] (axis cs:3,1.5) -- (axis cs:11,1.5);
	%\draw[dotted] (axis cs:3,4.5) -- (axis cs:11,4.5);
	
	\draw[dotted] (axis cs:3,8) -- (axis cs:1.5,8);
	\draw[dotted, <->, red] (axis cs:1.5,5) -- (axis cs:1.5,8) node[midway, left] {$\eta$};
	
	\draw[dotted] (axis cs:6,5) -- (axis cs:6,3);
	\draw[dotted, <->, red] (axis cs:3,3) -- (axis cs:6,3) node[midway, below] {$\eta$};

	%\addplot[no marks, domain=3:10] {0.7*x+0.9};
	%\draw[rounded corners=0.1cm] (axis cs:10,5.9) -- (axis cs:14,7.5) -- (axis cs:16,9.5) -- (axis cs:18,10.8) node[pos=0.5] (target) {}-- (axis cs:20,12.0);
	
	%\node[anchor=south west] at ([xshift=-0.1cm] axis cs:0,15) (posstraj) {\makecell[l]{Possible trajectory \\ of $d_i(t)$}};
	%\draw[->] ([yshift=-0.2cm, xshift=-0.1cm]posstraj.center) -- (target.center);
	\end{axis}
	\end{tikzpicture}}
		\subcaption{\label{fig:asynchroExemple}In the non-synchronized time-model.}
	\end{minipage}\hspace{0.02\linewidth}\begin{minipage}{0.48\linewidth}\centering
		\resizebox{\linewidth}{!}{\begin{tikzpicture}[samples=200]

\pgfplotsset{ticks=none}

\begin{axis}[xlabel=$t$ observed with $\h_g$, ylabel=$d(t)$ observed with $\h_i$,
xmin=-3.5,xmax=20,ymin=-2,ymax=25, axis lines=center] %,  width=6.2cm, height=5cm]

\addplot[no marks, domain=0:20, dashed] {x} node [pos=0.65,sloped, above]{$d(t)=t$};

%\addplot[no marks, dotted, domain=3:6, red] {0.5*x+0};

\draw[red] (axis cs:4,3) -- (axis cs:4,6);
\addplot[no marks, domain=4:7, red] {0+3};
\addplot[no marks, domain=7:9, red] {0.5*x-0.5}; % node[pos=0.7, sloped, below] {Slope $1/\rho$};
\addplot[no marks, domain=9:20, red] {x-5};
\addplot[no marks, domain=4:10, red] {1.5*x-0} node[pos=0.3, sloped, above] {Slope $\rho$};
\addplot[no marks, domain=10:20, red] {x+5};

\addplot[no marks] table {syncdat.dat};

\draw[line width=0.3pt] (axis cs:4,0) -- (axis cs:4,3) node[pos=0, below] {$s$};
\draw[line width=0.3pt] (axis cs:0,3) -- (axis cs:4,3) node[pos=0, left] {${d}(s)$};

\draw[<->, dotted, red] (axis cs:15,10) -- (axis cs:15,15) node[midway, right] {$\Delta$}; 

\draw[dotted] (axis cs:4,6) -- (axis cs:2.5,6);
\draw[dotted, <->, red] (axis cs:2.5,3) -- (axis cs:2.5,6) node[midway, left] {$\eta$};

%\draw[dotted, <->] (axis cs:11,1.5) -- (axis cs:11,4.5) node[midway, right] {$2\eta$};
%\draw[dotted] (axis cs:3,1.5) -- (axis cs:11,1.5);
%\draw[dotted] (axis cs:3,4.5) -- (axis cs:11,4.5);

%\addplot[no marks, domain=3:10] {0.7*x+0.9};
%\draw[rounded corners=0.1cm] (axis cs:10,5.9) -- (axis cs:14,7.5) -- (axis cs:16,9.5) -- (axis cs:18,10.8) node[pos=0.5] (target) {}-- (axis cs:20,12.0);

%\node[anchor=south west] at ([xshift=-0.1cm] axis cs:0,15) (posstraj) {\makecell[l]{Possible trajectory \\ of $d_i(t)$}};
%\draw[->] ([yshift=-0.2cm, xshift=-0.1cm]posstraj.center) -- (target.center);
\end{axis}
\end{tikzpicture}} %
		\subcaption{\label{fig:synchroExemple}In the synchronized time-model.} %
	\end{minipage}
		\caption{\label{fig:enveloppe-all}Envelope of $d(t)$ (red) and example of a possible evolution of $d(t)$ (black).}
	\end{figure}
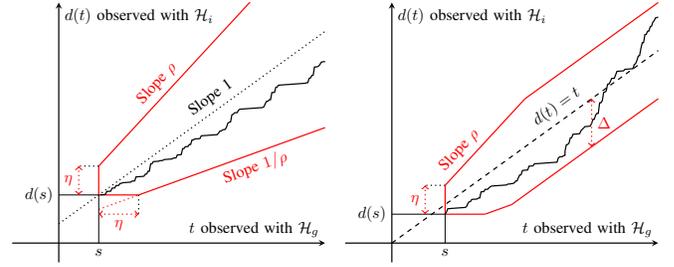	
	
	Figure~\ref{fig:asynchroExemple} presents, for a given known starting point $(s,d(s))$, the possible evolution space of $d(t)$ in the non-synchronized model as well as a possible trajectory. We note that the time-error function $d(t) - t$ can be unbounded in this model. We also note the symmetry of the non-synchronized envelope: if $d$ meets the stability conditions (\ref{eq:constr-async}), then $d^{-1}$ also does.
	
	\ant For a \ac{TSN} network, any clock $\h_i$ satisfies $\rho_i= 1+1\cdot 10^{-4}$ and $\eta_i=$ 2ns. Hence, for a \ac{TSN} network, the clock-independent parameters are $\rho=1+2\cdot 10^{-4}$ and $\eta=$ 4ns. Note that the above values are minimum requirements for clocks to be used as local clocks in \ac{TSN}. Of course, if better bounds are known from the manufacturer specifications, the values of $\rho$ and $\eta$ can be updated accordingly.

\subsection{The Synchronized Time-Model}
	In a synchronized network, clocks meeting the above stability requirements are synchronized with each other. The synchronization can be performed using for example \ac{PTP} \cite{ieeeIEEEStandardPrecision2008}, generalized \ac{PTP} defined in \ac{TSN}~\cite{ieeeIEEEStandardLocal2011}, \ac{NTP} \cite{millsNetworkTimeProtocol2010a}, WhiteRabbit \cite{moreiraWhiteRabbitSubnanosecond2009} or a \ac{GNSS}~\cite{powersGPSGalileoUTC2004}. When synchronization is used, the time-error function for any pair such as $(\h_g,\h_i)$ is bounded by the precision of the protocol.
	
	For a given network, we define $\Delta\ge0$ the \emph{time-error bound} of the network and we assume that for any pair $(\h_g,\h_i)$, $d$ meets the constraints of Equation (\ref{eq:constr-async}), plus
	\begin{equation}\label{eq:const-sync}
	 	\forall t, |d(t) - t| \le \Delta
	\end{equation}
	where $d=d_{g\rightarrow i}$. Note that $\Delta$ does not depend on $\h_g,\h_i$.

	\rema Clocks of a synchronized network have synchronization servos that control their frequency. This control can be used to perform some form of syntonization based on the synchronization messages as described in~\cite[Section 12.1]{ieeeIEEEStandardPrecision2008}. Another example, is NTP's second leap: for example, Google's public NTP services smear the second leap by increasing the clock frequency by 11.6ppm over a duration of 24h~\cite{googleLeapSmearPublic}. In both cases, the frequency offset used by the servo to control the clock might exceed the non-synchronized stability requirements of Equation~(\ref{eq:constr-async}). This is handled in our model by defining, for any clock $\h_i$, another linear term in the form $(t-s)\rho_{3,i}$ in Equation~(\ref{eq:error-model-ev}), where $\rho_{3,i}$ upper-bounds the properties of the servo of clock $\h_i$. Then, $\rho_i$ is redefined as $\rho_i=1+\rho_{1,i}+\rho_{2,i}+\rho_{3,i}$ and the network-wide parameter $\rho$ is redefined accordingly.
	
	\ant For tightly-synchronized networks, we take $\Delta=$1$\mu$s from \cite[Normative Annex B.3]{ieeeIEEEStandardLocal2011}. For loosely-synchronized networks, we select the precision of \ac{NTP}. In \cite[Figure 27]{rfc5905}, \ac{NTP} defines a ``step threshold'' of 125ms, and a survey of the \ac{NTP} network performed in \cite{murtaQRPp14CharacterizingQuality2006} notes that this value is hardly exceeded if the clock is synchronized~\cite[Section IV.B.1]{murtaQRPp14CharacterizingQuality2006}. Hence, we take $\Delta=$125ms.
	
 	Figure~\ref{fig:synchroExemple} presents, for a given known starting point $(s,d(s))$, the possible evolution space of $d(t)$ in the synchronized model as well as a possible trajectory. Note that the $\Delta$ envelope is not centered on the starting point but on the $d(t)=t$ function. Here again, we keep the symmetry noted previously on the non-synchronized constraints. % 
 %
% !TeX spellcheck = en_US
% !TeX rootfile = article.tex
\subsection{Network Model}\label{sec:model:network}
	%\todo{Begin Holly-Check}
	We model the network as a set of ``network elements'' and ``regulators''. Each flow $f$ has a non-cyclic path made of a series of connected network elements $(S_k)_{k=1\ldots n+1}$, with $k$ the index of the network element in the path of the considered flow $f$. Each flow $f$ is assumed to be processed by a regulator after each network element in its path, except the last one (Figure~\ref{fig:async-hop-model}). We call ``$k$th hop'' of this flow the sequence $S_k-\text{Reg}_k$. Each regulator can either be a \ac{PFR} or an \ac{IR}. For each flow $f$ and each index $k=1\ldots n$, we note Reg$_k$ the regulator that processes flow $f$ after the network element $S_k$. On Reg$_k$, we assume that we can configure a shaping curve $\sigma_k$ with a burst $b_{\text{Reg}_k}$ and a rate $r_{\text{Reg}_k}$ for this flow. Practical implementations support only limited accuracy, and we note $\mathcal{Q}_k(b)$ [resp. $\mathcal{R}_k(r)$] the lowest value that is configurable by this regulator and that is higher than $b$ [resp $r$]. By convention, Reg$_0$ denotes the source of the flow.
	\begin{figure}\centering %
		\resizebox{\linewidth}{!}{% !TeX spellcheck = en_US
% !TeX rootfile = article.tex
\begin{tikzpicture}
	\tikzstyle{conf} = [draw, line width=0.4pt]
	\tikzstyle{ref} = [anchor=north, xshift=0.1cm, yshift=0.1cm]
	\tikzstyle{tt} = [line width=0.4pt]
		
	\node at (0,0) (s) {};
	\node[draw, anchor=west] at ([xshift=1cm] s.east) (sys) {$S_k$};
	\node[draw, anchor=west] at ([xshift=1cm] sys.east) (pfr) {Reg$_k$};
	\node[anchor=north] at (pfr.south east) {$\mathcal{H}_{\text{Reg}_k}$};
	\node[anchor=north] at (sys.south east) {$\mathcal{H}_{S_k}$};
	
	\node[draw, dashed, text width=1cm] at (-1.5,0) (prevreg) {Reg$_{k-1}$};
	
	\draw[dashed, <-] (prevreg.west) -- ++(-0.5cm,0);
	\draw[dashed,->] (prevreg.east) -- (0,0);
	\draw[-, dashed] let \p1=(s.east) in (s.east) -- (0,\y1) node[pos=1,anchor=center] (target1) {};
	\draw[dashed] (0,-0.2cm) -- (0,0.2cm);
	
	\node[anchor=north] at ([yshift=0.1cm, xshift=0.2cm]prevreg.south east) {$\mathcal{H}_{\text{Reg}_{k-1}}$};

	\draw[->] (s.east) -- (sys.west) node[pos=0, above] {$f$} node[pos=0.5, anchor=center] (targetIn) {};
	\draw[->] (sys.east) -- (pfr.west);
	\draw[->] (pfr.east) -- ++(2cm,0) node[pos=0.5,anchor=center] (targetOut) {};

	\node[conf] at ([yshift=1cm]pfr) (pfrConf) {$\sigma_k$};
	\draw[->, tt] (pfrConf.south) -- (pfr.north);
	
	\draw[<->] ([yshift=-0.8cm] sys.south west) -- ([yshift=-0.8cm] sys.south east) node[right, pos=1] {$D_k$};
	\draw[<->] ([yshift=-1cm] sys.south west) -- ([yshift=-1cm] pfr.south east) node[right, pos=1] {$D_k'$};
	
	\node at ([yshift=1cm] targetIn.center) (valIn) {$\alpha_{k-1}$};
	\draw[->, tt] (valIn.south) -- (targetIn.center);
	
	\node at ([yshift=1cm] targetOut.center) (valOut) {$\alpha_{k}$};
	\draw[->, tt] (valOut.south) -- (targetOut.center);
	
\end{tikzpicture}} %
		\caption{\label{fig:async-hop-model} Network calculus model for a regulated flow. At each hop, the flow goes through a network element $S_k$ and is then regulated by the regulator Reg$_k$. Each device of the model has its own internal clock, noted at the bottom right.} %
	\end{figure}
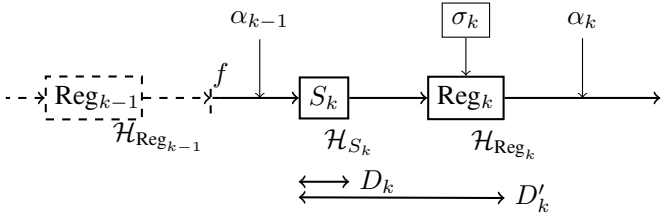 %
	When a set of flows share the same network element $S$, they do not need to share the same regulator just after $S$. However, when a set of flows share a same \ac{IR}, we assume that they also share the same upstream network element and that this network element is \ac{FIFO} for the considered set of flows.
		
	For each network element $S$, we assume that if we know the arrival curve of each flow $g$ going through $S$, then for any flow $f$ going through $S$ we can obtain a delay bound $D_S^f$ for the flow through the network element.
	
	Each network element and each regulator (including, for each flow $f$, its source Reg$_0$) is called a device. We assume that each device uses a local clock noted at the bottom right of the device, as in Figure~\ref{fig:async-hop-model}. In particular, any regulator releases packets according to the regulation rules that depend on its type \cite{leboudecTheoryTrafficRegulators2018,ayedHierarchicalTrafficShaping2014a} (see also Section~\ref{sec:prob}), but the computations for its operation are performed using its local clock. When we need to detail the clock used to observe a delay, an arrival curve or a service curve, we put it in superscript.
	
	Each flow $f$ exits it source Reg$_0$ with a leaky-bucket arrival curve of rate $r_0=r_{\text{Reg}_0}$ and burst $b_{0}=b_{\text{Reg}_0}$, when observed with $\h_{\text{Reg}_0}$.	For each flow $f$, we note $\alpha_k$ an arrival curve of the flow at the output of its $k$-th hop and $D_k$ [resp $D_k'$] an upper bound on the delay for the flow through S$_k$ [resp through the combination made of S$_k$ and Reg$_k$]. As the flow $f$ is not processed by any regulator after the last network element in its path $S_{n+1}$, for each flow $f$ we note by convention $D_{n+1}'=D_{n+1}$.
	%\todo{End Holly-Check}
	
	\rema This general model also applies in the specific case where several devices share the same clock, for instance if they are located within a same network node. We can indeed write that the clocks of the considered devices compose a sub-network, with parameters $\rho'=\rho$, $\eta'=\eta$, and that they are synchronized together as per the constraints of Equation~(\ref{eq:const-sync}) with a sub-network synchronization bound $\Delta'$ very small or null. However, the editors of IEEE P802.1Qcr (\acl{ATS}) mention in the ongoing draft that even two devices sharing the same oscillator of a same node can have different notions of time~\cite [Section 8.6.11.2]{ieeeDraftStandardLocal2019a}. They also mention that this difference needs to be accounted for in their ad-hoc delay analysis in~\cite[Informative Annex V.6]{ieeeDraftStandardLocal2019a}. In such case, we take for the sub-network parameters $\rho'=\rho$, $\eta'=\eta$ and $\Delta'=\max(|\text{ClockOffsetMax}|,|\text{ClockOffsetMin}|)$, where \text{ClockOffsetMax} and \text{ClockOffsetMin} are defined in \cite[Section 8.6.11.2]{ieeeDraftStandardLocal2019a}. %
%\input{hop-model} 
% !TeX spellcheck = en_US
% !TeX rootfile = article.tex

\section{Network Calculus Toolbox for Networks With Nonideal Clocks}\label{sec:toolbox}

The three-bound theorem of \cite{leboudecNetworkCalculusTheory2001} is valid whenever all the notions used in the theorem are expressed using the same clock: the arrival curves, the service curves, and the delays. Therefore, we propose a toolbox that can be used to change the clock used to observe one of the above notions; this results in an extension of network calculus.

In the entire section, we consider a device $j$ and a flow $f$ entering this device (see Figure~\ref{fig:ijSit}). Whenever $f$ and $j$ are unambiguous, the dependency of the defined functions and notions on $j,f$ will be omitted.
We also consider two clocks $\h_i$ and $\h_g$ and we denote with $d$ the function $d_{g\rightarrow i}$.

\subsection{Results on Delays}
\label{sec:toolbox:delay}
	We are interested in a bound on the measure, with the \ac{TAI}, of the delay that flow $f$ suffers through device $j$. If a delay bound is known with a different clock $\h_i$, the following proposition enables us to retrieve a bound as seen with any clock $\h_g$ (especially with the \ac{TAI}).
	
	\begin{proposition}[Changing the clock for a delay bound]\label{prop:toolbox:delay}
		If $D^{\h_i}$ is an upper-bound on a delay measured with clock $\h_i$, then $(d\oslash d)(D^{\h_i})$ is an upper-bound on the delay measured with $\h_g$, where $\oslash$ denotes the min-plus de-convolution\footnote{$f\oslash g:t\mapsto \sup_{\tau} \left[ f(t+\tau)-g(\tau) \right]$, see \cite{leboudecNetworkCalculusTheory2001}} and $d$ is the function $d_{g\rightarrow i}$.
	\end{proposition}
	The proof is in Appendix~\ref{proof:toolbox:delay}.	When a delay is observed with the \ac{TAI}, we call it a \ac{TAI} delay.

	\rema A measure of an \ac{ETE} delay is often a subtraction of two other measures: the time value given by a clock on the recipient's side when the packet arrives and the time value given by a clock located on the source side when the packet departs. If the clocks used to timestamps these two events are not perfectly synchronized with the \ac{TAI}, the measured delay may violate a delay bound computed in the \ac{TAI} using the work presented in this paper. Nevertheless, we consider the \ac{TAI} delay as being the ``true delay'' experienced by the flow, and any real-world measurement should take into account the uncertainties of the measurement clocks as per the best-practices of the time metrology domain \cite{ituDefinitionsTerminologySynchronization1996,ieeeIEEEStandardDefinitions2009}.
	
	In the \emph{Non-Synchronized Time Model}, we have, for any $t,\tau$, $(d \oslash d) (\tau) \le \rho \tau + \eta$ and thus $D^{\h_g}\leq \rho D^{\h_i}+\eta$. In the \emph{Synchronized Time Model}, we obtain $(d \oslash d) (\tau) \le \min(\rho \tau + \eta,\tau+2\Delta)$ and thus $D^{\h_g}\leq \min\left(\rho D^{\h_i}+\eta, D^{\h_i}+2\Delta\right)$.

	\ant For both time models, for values of the delay bound $D^{\h_i}$ ranging from $1\mu$sec to $200$msec, the relative increase on the bound ranges from $0.4\%$ to $0.02\%$. Practically, the effect of clock nonidealities on the \emph{definition} of delay bounds in time-sensitive networks can thus be ignored.

\subsection{Results on Arrival Curves}

	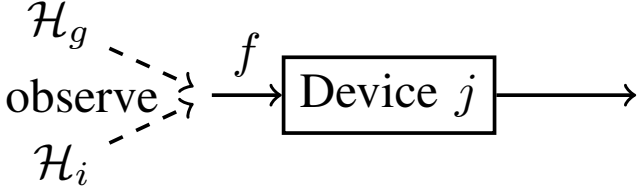
\begin{figure}\centering
		\resizebox{\linewidth}{!}{\begin{tikzpicture}
	\node[draw] at (0,0) (nj) {Device $j$};
	\node[anchor=east] at (-2,-0.5) (ni) {$\h_i$};
	\node[anchor=east] at (-2,0.5) (ng) {$\h_g$};
	\draw[->] ([xshift=-0.5cm] nj.west) -- (nj.west) node[pos=0.5, above] {$f$} node[pos=0, anchor=center] (f) {};
	\draw[->] (nj.east) -- ++(1cm,0);
	\draw[->, dashed] (ni) -- (f);
	\draw[->, dashed] (ng) -- (f);
	\node at (-2.2,0) {observe};
	
\end{tikzpicture}}
		\caption{\label{fig:ijSit}Clocks $\h_g$ and $\h_i$ observe flow $f$ entering device $j$}
	\end{figure}

	From \cite{leboudecNetworkCalculusTheory2001}, the cumulative function $R(t)$ of the flow entering the device $j$ is defined, when measured with the \acs{TAI}, as the number of bits entering the device between the instant measured as $0$ in the \ac{TAI} and the instant measured as $t$ in the \ac{TAI}. When we use clock $\h_i$ instead of \acs{TAI}, we obtain a different cumulative function. We call $R^{\h_i}(t)$ the cumulative function obtained when counting the number of bits entering a device between the instant measured as $0$ using $\h_i$ and the instant measured as $t$ using $\h_i$ (Figure~\ref{fig:ijSit}).  %from $\h_i$, as the number of bits of the flow entering the device between the instant measured as $0$ using $\h_i$ and the instant measured as $t$ using $\h_i$.
%	
	%However, we now consider clock $\h_i$ in  that needs to count the bits of the flow entering device $j$. We define the cumulative function $R^{\h_i}(t)$, viewed from $\h_i$, as the number of bits of the flow entering the device between the instant measured as $0$ using $\h_i$ and the instant measured as $t$ using $\h_i$.
%	
	The following proposition gives the relation between cumulative functions obtained with different clocks:
	\begin{proposition}[Changing the clock of a cumulative function]\label{prop:cumu} %
		For any clock pair $(\h_g,\h_i)$:
		$$\forall t, R^{\mathcal{H}_{g}}(t) = R^{\mathcal{H}_i}(d_{g\rightarrow i}(t))$$
	\end{proposition} %
	The proof of the proposition is given in Appendix~\ref{proof:toolbox:cumu}. We now define the concept of an arrival curve observed with a clock. %
	\begin{definition} %
		We say that a wide-sense increasing function $\alpha$ is an arrival curve for the flow entering the device $j$ when observed with clock $\mathcal{H}_i$ if $\forall t, \tau, R^{\mathcal{H}_i}(t+\tau) - R^{\mathcal{H}_i}(t) \le \alpha(\tau)$. We note such a function $\alpha^{\mathcal{H}_i}$.
	\end{definition}

	\rema In particular, when $\mathcal{H}_i$ is the \ac{TAI}, we retrieve the definition of an arrival curve, as defined in \cite[Chapter 1]{leboudecNetworkCalculusTheory2001}.
	
	\begin{proposition}[Changing the clock of an arrival curve]\label{prop:ac}
		If $\alpha^{\mathcal{H}_i}$ is an arrival curve for the flow entering the device and being observed with $\mathcal{H}_i$, then an arrival curve for the flow observed with $\mathcal{H}_{g}$ is $\alpha^{\mathcal{H}_{g}}:t\mapsto \alpha^{\mathcal{H}_i} \left((d\oslash d)(t)\right)$, where $\oslash$ is the min-plus de-convolution and $d=d_{g\rightarrow i}$.
	\end{proposition}

	The proof of the proposition is available in Appendix~\ref{proof:toolbox:ac}. We can now find an arrival curve for a flow as observed from any clock as long as we know one that is observed with one clock. We now apply the proposition to our two clock models.

	\paragraph*{Non-Synchronized Time Model} %
		Here, for any $t,\tau$, $(d \oslash d) (\tau) \le (\rho \tau + \eta)$, thus, if $\alpha^{\mathcal{H}_i}$ is an arrival curve observed with $\mathcal{H}_i$, then $\alpha^{\mathcal{H}_{g}}:t\mapsto \alpha^{\mathcal{H}_i}(\rho t+\eta)$ is an arrival curve observed with $\mathcal{H}_{g}$. %
				
 		\textbf{Example: }\emph{Application to leaky-bucket arrival curves.} Assume $\alpha^{\mathcal{H}_i}$ is a leaky-bucket arrival curve with burst $b_i$ and rate $r_i$, then $\alpha^{\mathcal{H}_{g}}$ is also a leaky-bucket arrival curve with burst $b_{g}=b_i+r_i \eta \ge b_i$ and a rate $r_{g}=\rho r_i \ge r_i$. We can also write $\alpha^{\mathcal{H}_{g}} = \gamma_{\rho r_i,b_i+r_i\eta}$.
 		
 		\ant This represents 0.02\% of a rate increase and a burst increase that is below 1 bit for most flows (below 250Mbits/s) as $\eta=$4ns.
 	\paragraph*{Synchronized Time Model}
	 	Here, we additionally obtain $ \forall \tau, (d \oslash d)(\tau) \le \tau + 2 \Delta$. Consequently, if $\alpha^{\mathcal{H}_i}$ is an arrival curve observed with $\mathcal{H}_i$, then $\alpha^{\mathcal{H}_{g}}:t\mapsto \alpha^{\mathcal{H}_i}(\min(t+2\Delta, \rho t + \eta))$ is an arrival curve observed with $\mathcal{H}_{g}$. As arrival curves are wide-sense increasing, we also have $\alpha^{\mathcal{H}_{g}}:t\mapsto \min\left[\alpha^{\mathcal{H}_i}(t+2\Delta),\alpha^{\mathcal{H}_i}(\rho t + \eta)\right]$.

		\textbf{Example: }\emph{Application to leaky-bucket arrival curves.}\label{ex:lb-async}
			Assume $\alpha^{\mathcal{H}_i}$ is a leaky-bucket arrival curve of burst $b_i$ and rate $r_i$.
			Then the function $t\mapsto \alpha^{\mathcal{H}_i}(t+2\Delta)$ is a leaky-bucket arrival curve of rate $r_i$ and burst $b_i+2r_i\Delta$. And $t\mapsto \alpha^{\mathcal{H}_i}(\rho t + \eta)$ is a leaky-bucket arrival curve of rate $r=\rho r_i$ and burst $b_i+r_i\eta$. Hence, $\alpha^{\mathcal{H}_{g}}= \gamma_{r_i,b_i+2r_i\Delta} \wedge \gamma_{\rho r_i,b_i+r_i\eta}$, the minimum of these two arrival curves (Figure~\ref{fig:synchroShape}). According to the numerical application, we expect to have $2r_i\Delta$ much higher than $r_i\eta$. In Figure~\ref{fig:synchroShape}, we respect this order but not the scale.

		\ant In a synchronized network, we have a second additional arrival curve with an unchanged rate but with an increased burst. For a flow of 500kbits/s, this represents a burst increase of 125kbits for a loosely-synchronized network and 1bit for a tightly-synchronized network. For loosely-synchronized networks, due to the high burst increase, the other part of the arrival curve, $\gamma_{\rho r_i,b_i+r_i\eta}$ needs to be used in order to obtain tight bounds.
			
		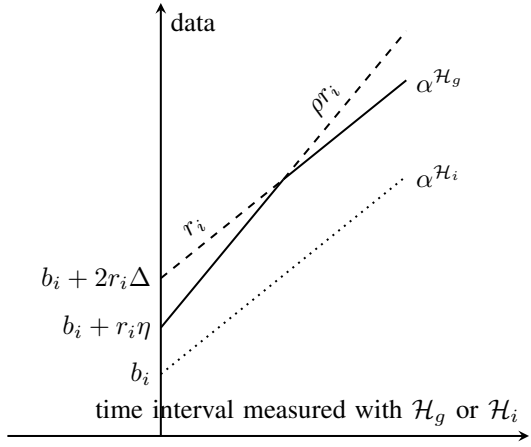
\begin{figure}\centering
			\resizebox{0.8\linewidth}{!}{	\begin{tikzpicture} %
	
	\pgfplotsset{ticks=none}
	\begin{axis}[xlabel=time interval measured with $\mathcal{H}_{g}$ or $\mathcal{H}_i$, ylabel=data,
	xmin=-5,xmax=12,ymin=0,ymax=7, axis lines=center] %, height=5cm, width=8cm]

	\addplot[no marks, domain=0:8, dotted] {0.4*x+1} node[pos=1, right] {$\alpha^{\mathcal{H}_i}$} node[pos=0,left] {$b_i$};
	
	\addplot[no marks, domain=0:4] {0.6*x+1.75} node[pos=0, left] {$b_i+r_i\eta$};
	\addplot[no marks, domain=0:4, dashed] {0.4*x+2.55} node[pos=0, left] {$b_i +2r_i\Delta$} node[pos=0.5, sloped, anchor=south east] {$r_i$};
	
	\addplot[no marks, domain=4:8, dashed] {0.6*x+1.75} node[pos=0.6, sloped, anchor=south east] {$\rho r_i$};
	\addplot[no marks, domain=4:8] {0.4*x+2.55}  node[pos=1, right] {$\alpha^{\mathcal{H}_g}$};

	\end{axis}
	\end{tikzpicture}}
			\caption{\label{fig:synchroShape}$\mathcal{H}_i$ is synchronized with $\mathcal{H}_{g}$. A flow $f$ has the leaky-bucket arrival curve $\alpha^{\mathcal{H}_i}$ when it is viewed from $\mathcal{H}_i$. Then, $f$ has the arrival curve $\alpha^{\mathcal{H}_{g}}$ when it is viewed from $\mathcal{H}_{g}$.}
		\end{figure}
	
		\begin{table}%
			\caption{\label{tab:ac-results} Relations between an arrival curve as observed with $\h_i$ and an arrival curve as observed with $\h_g$.}%
			\resizebox{\linewidth}{!}{% !TeX spellcheck = en_US
% !TeX rootfile = article.tex
	\begin{tabular}{l|l|l}
		& \multicolumn{2}{c}{Arrival curve} \\
		& General & Leaky-Bucket \\
		\hline
		in $\mathcal{H}_{i}$ & $\alpha^{\mathcal{H}_i}(t)$ & $\gamma_{r,b}$ \\ 
		in $\mathcal{H}_{g}$, general & $\alpha^{\mathcal{H}_i}((d_i \oslash d_i)(t))$ & --  \\
		in $\mathcal{H}_{g}$, non-sync & $\alpha^{\mathcal{H}_i}(\rho t + \eta)$ & $\gamma_{r \rho,b+r\eta}$  \\
		in $\mathcal{H}_{g}$, sync & $\alpha^{\mathcal{H}_i}(\min[\rho t + \eta,t+2\Delta])$ & $\gamma_{r \rho,b+r\eta} \wedge \gamma_{r,b+2r\Delta} $ \\
		
	\end{tabular}}
		\end{table}
	
		Table~\ref{tab:ac-results} regroups the results of this subsection. Note that the table can be used for any pair of clocks $(\h_g, \h_i)$, each being either a clock in the network or the \ac{TAI}.
	
\subsection{Results on Service Curves}
	As for the arrival curve concept, we define the service curve concept relative to a clock:
	\begin{definition} %
		We say that a wide-sense increasing function $\beta$ is a service curve of the device $j$ for the flow when observed with $\mathcal{H}_i$ if $\beta(0)=0$ and $\forall t \ge 0, R^{*\mathcal{H}_i}(t) \ge \inf_{0\le s\le t}R^{\mathcal{H}_i}(s) + \beta(t-s)$. We note such a function $\beta^{\mathcal{H}_i}$. %
	\end{definition} %
	\rema In particular, when $\mathcal{H}_i$ is the \ac{TAI}, we retrieve the definition of a service curve, as defined in \cite[Chapter 1]{leboudecNetworkCalculusTheory2001}.
	\begin{proposition}[Changing the clock for a service curve]\label{prop:sc} %
		If $\beta^{\mathcal{H}_i}$ is a service curve observed with $\mathcal{H}_i$, then a service curve observed with $\mathcal{H}_{g}$ is $\beta^{\mathcal{H}_{g}}:t\mapsto \beta^{\mathcal{H}_i}((d \overline{\oslash} d)(t))$, where $d=d_{g\rightarrow i}$ and $d\overline{\oslash}d(t)=\inf_{u\ge0} [d(t+u) - d(u)]$ (max-plus de-convolution~\cite[Section 3.2.1]{leboudecNetworkCalculusTheory2001}).
	\end{proposition} %

	The proof is in Appendix~\ref{proof:toolbox:sc}. Again, we now apply the proposition to the two time models:
	\paragraph*{Non-Synchronized Time Model}
		Here, $\forall t, (d \overline{\oslash} d)(t) \ge \frac{1}{\rho}\left|t-\eta\right|^+$, thus %where $|x|^+ = \max(0,x)$. And
if $\beta^{\mathcal{H}_i}$ is a service curve viewed from $\mathcal{H}_i$, then $\beta^{\mathcal{H}_{g}}:t\mapsto \beta^{\mathcal{H}_i}(\frac{1}{\rho}|t-\eta|^+)$ is a service curve viewed from $\mathcal{H}_{g}$.
		
	\textbf{Example: }\emph{Application to rate-latency service curves.} Assume $\beta^{\mathcal{H}_i}$ is a rate-latency service curve with rate $R_i$ and latency $T_i$, then $\beta^{\mathcal{H}_{g}}$ is a rate-latency service curve with rate $R_g=R_i/\rho\le R_i$ and latency $T_g=\eta+\rho T_i\ge T_i$.

	\ant The service observed with another clock has a guaranteed rate reduced by 0.02\% and a slightly increased latency. For a latency of $10\mu$s, this represents 6ns of added latency.
	
	\textbf{Example: }\emph{Application to \ac{PFR} service curves.}	Recall from Section~\ref{sec:prob} that, when observed with its internal clock, a \ac{PFR} offers to the regulated flow a service curve equal to its shaping curve. If we observe a token-bucket filter with rate $r_i$ and burst $b_i$ with a different clock, $\h_g$, we obtain the service curve $\beta^{\mathcal{H}_{g}} = \delta_{\eta} \otimes \gamma_{r_i/\rho,b_i}$. This corresponds to the leaky-bucket service curve with a burst $b_i$ and a service rate $r_g=r_i/\rho\le r_i$, delayed by delay $\eta$ (Figure~\ref{fig:lb-sc-async}).
	
	\paragraph*{Synchronized Time Model} If $\beta^{\mathcal{H}_i}$ is a service curve viewed from $\mathcal{H}_i$, then $\beta^{\mathcal{H}_{g}}:t\mapsto\beta^{\mathcal{H}_i}(\max[\frac{1}{\rho}(t-\eta),t-2\Delta,0])$ is a service curve viewed from $\mathcal{H}_{g}$.
		\begin{figure*}
		\begin{minipage}{0.30\linewidth}\centering
			\resizebox{\linewidth}{!}{% !TeX spellcheck = en_US
% !TeX rootfile = article.tex
%\begin{tikzpicture}
	\begin{tikzpicture}
	
	\pgfplotsset{ticks=none}
	\begin{axis}[xlabel=time interval measured with $\mathcal{H}_{g}$ or $\mathcal{H}_i$, ylabel=data,
	xmin=-1,xmax=11,ymin=-0.6,ymax=5, axis lines=center, height=7cm, width=8cm]

	\addplot[no marks, domain=0:8, dotted] {0.4*x+1} node[pos=1, right] {$\beta^{\mathcal{H}_i}$} node[pos=0,left] {$b_i$} node[pos=0.5,above,sloped] {$r_i$};
	
	\addplot[no marks, domain=0:1.6, dashed] {1};
	\addplot[no marks, domain=1.6:8] {0.3*x+0.5} node[pos=0.8,sloped,below] {$r_i/\rho$} node[pos=1,right] {$\beta^{\mathcal{H}_{g}}$};
	
	\draw[-] (axis cs:1.6,1) -- (axis cs:1.6,0) node[pos=1,below] {$\eta$};

	\end{axis}
%\end{tikzpicture}
\end{tikzpicture}}
			\subcaption{\label{fig:lb-sc-async}$\mathcal{H}_i$, $\mathcal{H}_{g}$ are not synchronized. $\beta^{\h_i}$ is the service curve of a \acf{PFR}.}
			\end{minipage}\hspace{0.04\linewidth}\begin{minipage}{0.30\linewidth}\centering
				\resizebox{\linewidth}{!}{% !TeX spellcheck = en_US
% !TeX rootfile = article.tex
%\begin{tikzpicture}
	\begin{tikzpicture}
	
	\pgfplotsset{ticks=none}
	\begin{axis}[xlabel=time interval measured with $\mathcal{H}_{g}$ or $\mathcal{H}_i$, ylabel=data,
	xmin=-1,xmax=12,ymin=-0.6,ymax=5, axis lines=center,  height=7cm, width=8cm]

	\addplot[no marks, domain=0:8, dotted] {0.4*x+1} node[pos=1, right] {$\beta^{\mathcal{H}_i}$} node[pos=0,left] {$b_i$} node[pos=0.5,above,sloped] {$r_i$};
	
	\addplot[no marks, dashed, domain=0:1.6] {1};
	\addplot[no marks, dashed, domain=1.6:3] {1};
	\draw[dashed] (axis cs:3,1) -- (axis cs:3,0) node[pos=1, below] {$2\Delta$};

	\addplot[no marks, domain=1.6:4.33] {0.2*x+0.68};
	\addplot[no marks, domain=4.33:8, dashed] {0.2*x+0.68} node[pos=0.8,sloped,below] {$r_i/\rho$};
	\addplot[no marks, domain=3:4.33, dashed] {0.4*x+1.8-2};
	\addplot[no marks, domain=4.33:8] {0.4*x+1.8-2} node[pos=1,right] {$\beta^{\mathcal{H}_{g}}$} node[pos=0.5,sloped,above] {$r_i$};
	
	\draw (axis cs:1.6,1) -- (axis cs:1.6,0) node[pos=1,below] {$\eta$};

	\end{axis}
%\end{tikzpicture}
\end{tikzpicture}}
				\subcaption{\label{fig:lb-sc-sync}$\mathcal{H}_i$, $\mathcal{H}_{g}$ are synchronized. $\beta^{\h_i}$ is the service curve of a \acf{PFR}.}
			\end{minipage}\hspace{0.04\linewidth}\begin{minipage}{0.30\linewidth}\centering
				\resizebox{\linewidth}{!}{% !TeX spellcheck = en_US
% !TeX rootfile = article.tex

	\begin{tikzpicture}
	
	\pgfplotsset{ticks=none}
	\begin{axis}[xlabel=$\mathcal{H}_{g}$ or $\mathcal{H}_i$, ylabel=data,
	xmin=-1,xmax=12,ymin=-3.4,ymax=8, axis lines=center, height=7cm, width=8cm]

	\addplot[no marks, domain=2:6, dotted] {1.5*x-3} node[pos=0,below] {$T_i$} node[pos=0.8,sloped,above] {$R_i$} node[pos=1,above] {$\beta^{\mathcal{H}_i}$};
	
	\addplot[no marks, domain=3:8] {0.75*x-2.25} node[pos=0, anchor=center] (y1){};
	\draw[dashed] (y1.center) -- ++(0,-0.5cm) node[pos=1, below, draw, dashed] {$\rho T_i+\eta$};
	
	\addplot[no marks, domain=8:10, dashed] {0.75*x-2.25}node[pos=0.7, below, sloped] {$R_i/\rho$};
	
	\addplot[no marks, domain=5.5:8, dashed] {1.5*x-8.25}node[pos=0.5,below,sloped] {$R_i$} node[pos=0,below] {$T_i+2\Delta$};
	
	\addplot[no marks, domain=8:9.5] {1.5*x-8.25} node[pos=1,above] {$\beta^{\mathcal{H}_{g}}$};

	\end{axis}
\end{tikzpicture}} %
				\subcaption{\label{fig:rl-sc-sync} $\mathcal{H}_i$, $\mathcal{H}_{g}$ are synchronized. $\beta^{\mathcal{H}_i}$ is a rate-latency service curve.} %
			\end{minipage}
			\caption{\label{fig:all-sync} %Different relation between service curves depending on the synchronization mode and on the service type.
		If the system offers service curve $\beta_{\h_i}$ when observed with $\h_i$, it offers service curve $\beta_{\h_g}$ when observed with $\h_g$.}
		\end{figure*}
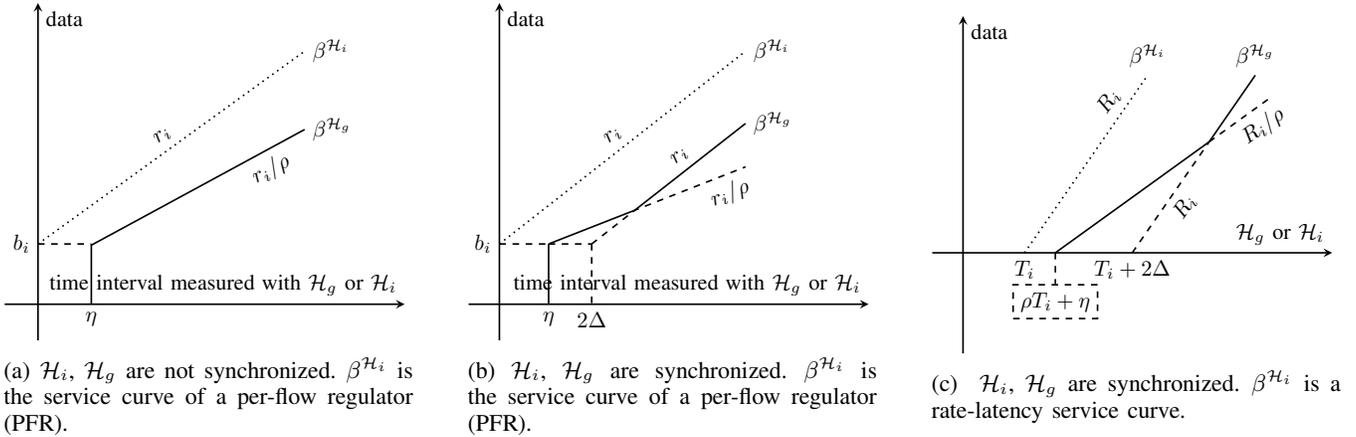
	
		\textbf{Example: }\emph{Application to rate-latency service curves.}
			Assume $\beta^{\mathcal{H}_i}$ is a rate-latency service curve with rate $R_i$ and latency $T_i$, then a service curve as observed with $\mathcal{H}_{g}$ is $\forall t,
			\beta^{\mathcal{H}_{g}}(t) = \lambda_{R_i/\rho, \eta+\rho T_i} \vee \lambda_{R_i, T_i+2\Delta}$ where $\vee$ denotes the maximum. I.e., $\beta^{\mathcal{H}_{g}}$ is the maximum between the rate-latency service curve of rate $R_i/\rho$, latency $\eta+\rho T_i$ and the rate-latency service curve of rate $R_i$ and latency $T_i+2\Delta$. Note that if $\rho-1$ is small and with our previous remark in Section~\ref{ex:lb-async} about $\eta$ and $\Delta$, we expect the latency $\eta+\rho T_i$ to be less than the latency $T_i+2\Delta$. Hence, the shape of $\beta^{\mathcal{H}_{g}}$ is given in Figure~\ref{fig:rl-sc-sync}.
		
		\ant In a synchronized network, compared to the non-synchronized time model, we also obtain a service with an unchanged rate but an increased latency: $2\mu$s for a tightly-synchronized network and 300ms for a loosely-synchronized network.
	
		\textbf{Example: }\emph{Application to \ac{PFR} service curves.}
			Assume $\beta^{\mathcal{H}_i}$ is a service curve offered by a token-bucket filter observed with its internal clock, with burst $b_i$ and rate $r_i$, then a service curve of this \ac{PFR} as observed with $\mathcal{H}_{g}$ is $\forall t$, $\beta^{\mathcal{H}_{g}}(t) = b_i + r_i \max[\frac{1}{\rho}(t-\eta),t-2\Delta,0]$. Ie,
			$\beta^{\mathcal{H}_{g}} = (\gamma_{r_i/\rho,b_i} \otimes \delta_{\eta}) \vee (\gamma_{r_i,b_i}\otimes\delta_{2\Delta})$ (Figure~\ref{fig:lb-sc-sync}).	
	\begin{table}%
		\caption{\label{tab:sc-results} Relations between a service curve of a system observed with $\h_i$ and a service curve of the same system observed with $\h_g$.}%
		\resizebox{\linewidth}{!}{% !TeX spellcheck = en_US
% !TeX rootfile = article.tex
	\begin{tabular}{l|l|l|l}
		& \multicolumn{3}{c}{Service curve} \\
		& General & Rate-Latency & Leaky-Bucket  \\
		\hline
		in $\mathcal{H}_{i}$ & $\beta^{\mathcal{H}_i}(t)$ & $\lambda_{R,T}$ & $\gamma_{r,b}$ \\ 
		in $\mathcal{H}_{g}$, general & $ \beta^{\mathcal{H}_i}((d_i \overline{\oslash} d_i)(t))$ & -- & -- \\
		in $\mathcal{H}_{g}$, non-sync & $\beta^{\mathcal{H}_i}(1/\rho\cdot|t - \eta|^+)$ & $\lambda_{R/\rho,\rho T+\eta}$ & $\delta_{\eta} \otimes \gamma_{r/\rho,b}$ \\
		\multirow{2}{*}{in $\mathcal{H}_{g}$, sync} &  \multirow{2}{*}{$\beta^{\mathcal{H}_i}(\max[1/\rho\cdot (t-\eta),t-2\Delta,0])$ }&  $\lambda_{R/\rho,\rho T+\eta} $ & $(\delta_{\eta} \otimes \gamma_{r/\rho,b})$\\
		& & $\vee \lambda_{R,T+2\Delta}$  & $\vee (\delta_{2\Delta} \otimes \gamma_{r,b})$  \\
		
	\end{tabular}}
	\end{table}
	Table~\ref{tab:sc-results} regroups the results of this subsection. Note that the table can be used for any pair of clocks $(\h_g, \h_i)$, each being either a clock in the network or the \ac{TAI}.

\subsection{Results on Regulators}\label{sec:toolbox:reg}
%	\todo{revoir en fonction du problem statement}
	
	%A \ac{PFR} is a device, defined in \cite[Definition 1.7.6]{leboudecNetworkCalculusTheory2001}, that releases packets as soon as doing so does not violate its configured shaping curve. An implementation of such a device will ensure that it meets this definition when observed with its own internal clock, but we can observe that the definition does not hold, in the general case, when the device is observed with an external clock.

	Recall that the operation of a \ac{PFR} requires measuring times, therefore, a system that is a \ac{PFR} when observed with its own clock might no longer be a \ac{PFR} when observed with a different clock. Indeed, assume, for instance, that the \ac{PFR}'s clock runs faster than expected when observed with an external clock. Then the external observer will see that the device violates the configured shaping curve and cannot be a \ac{PFR}. Conversely, if the clock of the device is two slow from an external point of view, then the external observer will conclude that the device is not a \ac{PFR} because it does not release the packets as soon as possible. Worse, the device clock may oscillate between being too fast or being too slow.
	
	Similar observations can be done with the \ac{IR}. Hence, the property of being a \ac{PFR} or an \ac{IR} is only valid when the regulator is observed with its internal clock. Consequently, none of the properties of the  \ac{IR}~\cite{leboudecTheoryTrafficRegulators2018} or the \ac{PFR}~\cite[Sections 1.5 and 1.7]{leboudecNetworkCalculusTheory2001} are expected to hold when observed with an external clock, in general.
	
	One of them states that, at the ouput of the regulator, a flow has the shaping curve as arrival curve. Observed with an external clock, this property does not hold but an arrival curve as observed with this external clock can be retrieved using Table~\ref{tab:ac-results}.
	
	In the following sections, we focus on the  shaping-for-free property (Section~\ref{sec:prob}), as it is underlies all delay computations in deterministic networks with regulators. %~\cite[Thm 1.5.2]{leboudecNetworkCalculusTheory2001}, \cite[Thm 4]{leboudecTheoryTrafficRegulators2018}.
%Recall that it expresses that a regulator, configured with the arrival curve of a flow at the input of a \ac{FIFO} system, and placed after that same system, does not increase the \acf{ETE} delay bound.
	Proposition~\ref{prop:insta-async} proves that in non-synchronized networks, the shaping-for-free property holds neither for a \ac{PFR} nor for an \ac{IR}.
	For synchronized networks, Section~\ref{sec:sync:unadapted-pfr} computes a lower and an upper bound on the worst-case penalty incurred by a \ac{PFR}. Finally, Proposition~\ref{prop:insta-sync-ir} proves that the \ac{IR} cannot provide any delay bound.

%\input{gs-prop.tex}	

% !TeX spellcheck = en_US
% !TeX rootfile = article.tex

% !TeX spellcheck = en_US
% !TeX rootfile = article.tex
\section{%Management and analysis of n
Non-synchronized networks with regulators}
\label{sec:manag-async}

	We now combine the set of results of the previous section with other network calculus results to analyze non-synchronized networks containing regulators. We focus on a network with flows that are leaky-bucket constrained at their sources.

	\subsection{Instability of Non-Adapted Regulators in Non-Synchronized Networks}
	\label{sec:manag-async:unstable} %	
		%\todo{drop-eligible}The following Proposition proves that regulators may yield unbounded latencies in non-synchronized networks.		%

		\begin{proposition}[Instability of non-adapted regulators in non-synchronized networks]\label{prop:insta-async} %
			Consider a non-synchronized network with $\rho>1$. Consider a network element $S$ that is FIFO per flow and guarantees to a flow $f$ a delay $\leq D$ if the flow satisfies an arrival curve $\gamma_{r,b}$ (both measured in \ac{TAI}). After processing by $S$, the flow is submitted to a non-adapted regulator (\ac{PFR} or \ac{IR}) (Figure~\ref{fig:prob}) with shaping curve $\gamma_{r,b}$. There exists adversarial source clocks within our time model and adversarial traffic generation satisfying the $\gamma_{r,b}$ arrival curve such that the delay of the flow through the regulator, measured using any clock, is unbounded. %
		\end{proposition} %
		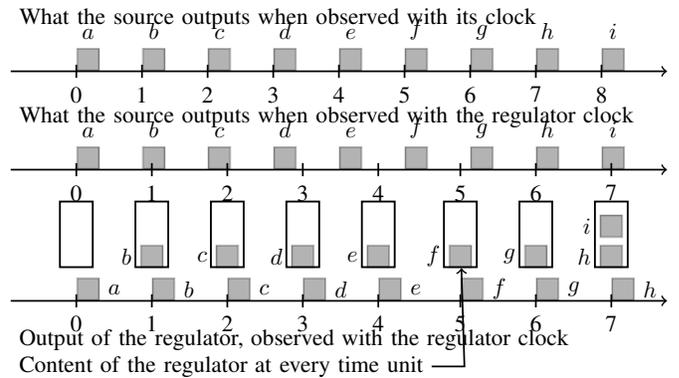
\begin{figure} %
			\resizebox{\linewidth}{!}{% !TeX spellcheck = en_US
% !TeX root=article.tex

\begin{tikzpicture}
	\tikzstyle{f} = [draw, regular polygon,regular polygon sides=4, anchor=south west, fill=black, opacity=0.3]
	
	\node[anchor=west] at (-1,0.8) {What the source outputs when observed with its clock};
	\draw[->] (-1,0) -- (9,0);
	\draw[-] (0.,-0.1) -- (0,0.1) node[pos=0, below] {0};
	\draw[-] (1,-0.1) -- (1,0.1) node[pos=0, below] {1};
	\draw[-] (2,-0.1) -- (2,0.1) node[pos=0, below] {2};
	\draw[-] (3,-0.1) -- (3,0.1) node[pos=0, below] {3};
	\draw[-] (4,-0.1) -- (4,0.1) node[pos=0, below] {4};
	\draw[-] (5,-0.1) -- (5,0.1) node[pos=0, below] {5};
	\draw[-] (6,-0.1) -- (6,0.1) node[pos=0, below] {6};
	\draw[-] (7,-0.1) -- (7,0.1) node[pos=0, below] {7};
	\draw[-] (8,-0.1) -- (8,0.1) node[pos=0, below] {8};
	\node[f, label=above:$a$] at (0,0) {};
	\node[f, label=above:$b$] at (1,0) {};
	\node[f, label=above:$c$] at (2,0) {};
	\node[f, label=above:$d$] at (3,0) {};
	\node[f, label=above:$e$] at (4,0) {};
	\node[f, label=above:$f$] at (5,0) {};
	\node[f, label=above:$g$] at (6,0) {};
	\node[f, label=above:$h$] at (7,0) {};
	\node[f, label=above:$i$] at (8,0) {};

	\draw[->] (-1,-1.5) -- (9,-1.5);
	\node[anchor=west] at (-1,-0.7) {What the source outputs when observed with the regulator clock};
	\draw[-] (0,-1.6) -- (0,-1.4) node[pos=0, below] {0};
	\draw[-] (1.15,-1.6) -- (1.15,-1.4) node[pos=0, below] {1};
	\draw[-] (2.30,-1.6) -- (2.30,-1.4) node[pos=0, below] {2};
	\draw[-] (3.45,-1.6) -- (3.45,-1.4) node[pos=0, below] {3};
	\draw[-] (4.60,-1.6) -- (4.60,-1.4) node[pos=0, below] {4};
	\draw[-] (5.85,-1.6) -- (5.85,-1.4) node[pos=0, below] {5};
	\draw[-] (7,-1.6) -- (7,-1.4) node[pos=0, below] {6};
	\draw[-] (8.15,-1.6) -- (8.15,-1.4) node[pos=0, below] {7};
	\node[f, label=above:$a$] at (0,-1.5) {};
	\node[f, label=above:$b$] at (1,-1.5) {};
	\node[f, label=above:$c$] at (2,-1.5) {};
	\node[f, label=above:$d$] at (3,-1.5) {};
	\node[f, label=above:$e$] at (4,-1.5) {};
	\node[f, label=above:$f$] at (5,-1.5) {};
	\node[f, label=above:$g$] at (6,-1.5) {};
	\node[f, label=above:$h$] at (7,-1.5) {};
	\node[f, label=above:$i$] at (8,-1.5) {};

	\node (rect) at (0,-3) (ir1) [anchor=south, draw,thick,minimum width=0.5cm,minimum height=1cm] {};
	\node (rect) at (1.15,-3) (ir2) [anchor=south, draw,thick,minimum width=0.5cm,minimum height=1cm] {};
	\node (rect) at (2.30,-3) (ir3) [anchor=south, draw,thick,minimum width=0.5cm,minimum height=1cm] {};
	\node (rect) at (3.45,-3) (ir4) [anchor=south, draw,thick,minimum width=0.5cm,minimum height=1cm] {};
	\node (rect) at (4.60,-3) (ir5) [anchor=south, draw,thick,minimum width=0.5cm,minimum height=1cm] {};
	\node (rect) at (5.85,-3) (ir6) [anchor=south, draw,thick,minimum width=0.5cm,minimum height=1cm] {};
	\node (rect) at (7,-3) (ir7) [anchor=south, draw,thick,minimum width=0.5cm,minimum height=1cm] {};
	\node (rect) at (8.15,-3) (ir8) [anchor=south, draw,thick,minimum width=0.5cm,minimum height=1cm] {};
	
	\draw[->] (-1,-3.5) -- (9,-3.5);
	\node[anchor=west] at (-1,-4.1) {Output of the regulator, observed with the regulator clock};
	\draw[-] (0,-3.6) -- (0,-3.4) node[pos=0, below] {0};
	\draw[-] (1.15,-3.6) -- (1.15,-3.4) node[pos=0, below] {1};
	\draw[-] (2.30,-3.6) -- (2.30,-3.4) node[pos=0, below] {2};
	\draw[-] (3.45,-3.6) -- (3.45,-3.4) node[pos=0, below] {3};
	\draw[-] (4.60,-3.6) -- (4.60,-3.4) node[pos=0, below] {4};
	\draw[-] (5.85,-3.6) -- (5.85,-3.4) node[pos=0, below] {5};
	\draw[-] (7,-3.6) -- (7,-3.4) node[pos=0, below] {6};
	\draw[-] (8.15,-3.6) -- (8.15,-3.4) node[pos=0, below] {7};
	
	\node[f, label=right:$a$] at (0,-3.5) {};
	\node[f, label=left:$b$, anchor=south] at (ir2.south) {};
	\node[f, label=right:$b$] at (1.15,-3.5) {};
	\node[f, label=left:$c$, anchor=south] at (ir3.south) {};
	\node[f, label=right:$c$] at (2.30,-3.5) {};
	\node[f, label=left:$d$, anchor=south] at (ir4.south) {};
	\node[f, label=right:$d$] at (3.45,-3.5) {};
	\node[f, label=left:$e$, anchor=south] at (ir5.south) {};
	\node[f, label=right:$e$] at (4.60,-3.5) {};
	\node[f, label=left:$f$, anchor=south] at (ir6.south) {};
	\node[f, label=right:$f$] at (5.85,-3.5) {};
	\node[f, label=left:$g$, anchor=south] at (ir7.south) {};
	\node[f, label=right:$g$] at (7,-3.5) {};
	\node[f, label=left:$h$, anchor=south] at (ir8.south) (h) {};
	\node[f, label=left:$i$, anchor=south] at ([yshift=0.1cm]h.north){};
	\node[f, label=right:$h$] at (8.15,-3.5) {};

	\node[anchor=west] at (-1,-4.5) (ll) {Content of the regulator at every time unit};
	\draw[->] (ll.east) -- ++(0.5cm,0) -- (ir6);
	
\end{tikzpicture} } %
			\caption{\label{fig:async:timeline-local} Adversarial case triggering the instability for non-synchronized non-adapted regulators.} %
		\end{figure} %
		The proof is in Appendix~\ref{proof:insta-async}. Recall that if the clocks would be ideal, the worst-case delay of the flow would not be increased by the regulator and the total delay would thus be $\leq D$. In contrast, with nonideal clocks, the total delay is unbounded, thus ``shaping-for-free'' does not hold.
		
		As illustration, consider for simplicity a \ac{PFR} or an \ac{IR} regulating only one flow: in this situation, both devices present the same behavior. We consider the source of a leaky-bucket constrained flow with all packets of same size, burst $b=1$~packet and rate $r=1$ packet per time unit. In the first line of Figure~\ref{fig:async:timeline-local}, we show the traffic trace at the output of the source, observed with its clock. The source is greedy, its sends at its maximum rate of one packet per time unit.
		
		Now assume that the \ac{FIFO} system $S$ has zero delay, whatever the clock used to observe it. Then the first timeline in Figure~\ref{fig:async:timeline-local} is also the traffic trace at the input of the regulator, but observed with the source clock.
		Assume now that the clock of the regulator is running slower than the one of the source. Here we exaggerate the phenomenon and we take $d_{\text{source} \rightarrow \text{regulator}}(t)~=~s_0~t$ with $s_0=6/7$. The regulator is non-adapted. It enforces the arrival curve of the source (one packet of burst, one packet per unit of time) by using its internal clock.
	
		Seen with the regulator clock, the incoming traffic trace is given with the second timeline in Figure~\ref{fig:async:timeline-local}. Packet $b$ appears to arrive too soon (at time
 $0$) and needs to be stored in the regulator. It is released when the regulator clock reaches time $1$. When measured with the regulator clock, packet $b$ suffers a delay of $1/7$ time unit, and $1/6$ time unit when observed with the source clock.
		
		We note that each packet has an increasing delay. At local time $7$, the regulator has now two packets in its buffer, and it will increase to three packets at time $14$, and so forth. We observe that the delays and the buffer occupation is linearly increasing. Hence, we have an unstable system with no delay bound. In practice, buffer sizes are finite and instead of unbounded delay, we will see unexpected packet losses (this contradicts the purpose of time-sensitive networks, which assume zero loss due to buffer overflow).%
		
		\ant Looking at the example, we observe that the delays observed with the source clock increase linearly with time and with $1/s_0-1$, which can be as large as $\rho-1$. In a \ac{TSN} context, this represents approximately $200\mu $s of increased worst-case delay per second of operation.

		The example presented above is inspired by a remark made in the TSN ATS draft \cite[Annex V.8]{ieeeDraftStandardLocal2019a}: ``\textit{If the upstream device [\ldots] runs faster than nominal and [the] downstream Bridge [\ldots] runs slower than nominal, the backlog as well as the per hop delay in the downstream Bridge could grow under peak conditions}''. In our model, the ``upstream device'' represents the upstream regulator, or the source, the ``downstream Bridge'' represents the next regulator, and ``peak conditions'' refer to a greedy source.
	
	\subsection{The Rate and Burst Cascade for Non-Synchronized Networks with Regulators}
	
		This last quoted remark highlights how a first solution for the instability problem can be formulated: one can make sure that whatever the clock conditions (but within the constraints of Equation~(\ref{eq:constr-async})), the downstream device will always have an output rate higher than the input. This requires increasing slightly the nominal rate of the regulator; and because this increase is performed at every hop, it generates a rate cascade that was first described in \cite[Annex V.8]{ieeeDraftStandardLocal2019a}. In this section, we refine the method into a rate an burst cascade, with the following differences:
		
		$\bullet$ We consider the jitter tolerance $\eta$ of the clocks, unlike~\cite{ieeeDraftStandardLocal2019a}.
		
		$\bullet$ The discussion in~\cite[Annex V.7]{ieeeDraftStandardLocal2019a}, that is related to the difference between the theoretical \ac{IR} behavior and its concrete implementation in \ac{ATS} is not considered here.
		
		$\bullet$ We take into account the finite resolution of the configurable rates and bursts on the regulators, whereas~\cite{ieeeDraftStandardLocal2019a} assumes that any rate or burst is configurable.
		
		Last, we use the network calculus toolbox in Section~\ref{sec:toolbox} to prove that the rate burst cascade ensures the stability of the network, and we also compute per-hop delay bounds.
		
		%$\bullet$ \todo{drop-eligible item} In our time model, we assume that the constraints of Equation~\ref{eq:constr-async} hold for any pair of clocks, not only when comparing to the true time. This prevents the double-penalty caused by the systematic reference to the true time and provides hence better bounds.

		The rate and burst cascade works as follows. For each flow $f$, we use the notation and the reference configuration in Section~\ref{sec:model:network} and Figure~\ref{fig:async-hop-model}.
		
		\textbf{Step~1:} For each flow $f$, and each hop $k=1\ldots n$ in its path, configure $\text{Reg}_k$ with $r_{\text{Reg}_{k}} = \mathcal{R}_k(\rho r_{\text{Reg}_{k-1}})$ and $b_{\text{Reg}_{k}} = \mathcal{Q}_k({b_{\text{Reg}_{k-1}} }+\eta r_{\text{Reg}_{k-1}})$. Recall that $\mathcal{R}_k(r)$ [resp. $\mathcal{Q}_k(b)$] denote the configurable rate [resp. burst] higher than $r$ [resp. $b$] for this regulator. Recall also that $\rho$ and $\eta$ are network-wide parameters that depend neither on the considered clock nor on $k$.
		
		\textbf{Step~2:} For each flow $f$, and each hop $k=1\ldots n$ in its path, the configured shaping curve $\sigma_{k-1}$ is an arrival curve at the output of the regulator Reg$_{k-1}$, when observed with the clock of the regulator (Section~\ref{sec:toolbox:reg}). Using Table~\ref{tab:ac-results} with $\h_g=\hze$ and $\h_i=\mathcal{H}_{\text{Reg}_{k-1}}$, it follows that flow $f$ has a leaky-bucket arrival curve $\alpha_k^{\hze}$ of rate $\rho r_{\text{Reg}_{k-1}}$ and burst $b_{\text{Reg}_{k-1}} + \eta r_{\text{Reg}_{k-1}} $ at the input of S$_k$, when observed with $\hze$.
		
		For each network element $S$ in the network and each flow $f'$ crossing $S$, the arrival curve of $f'$ at the input of $S$, observed with $\hze$ is given by the above $\alpha_k^{\hze}$, with $k$ the index of $S$ in the path of flow $f'$. According to the assumptions of Section~\ref{sec:model:network}, we can compute a \ac{TAI} delay bound $D_S^{f,\hze}$ of any flow $f$ that goes through $S$.
		
		\textbf{Step~3:} For each flow $f$, and each hop $k=1\ldots n$ in its path, compute the TAI delay bound, $D_k'^{\hze}$, of the flow through the sequence $S_k-\textmd{Reg}_k$, using the next proposition: %
		\begin{proposition}\label{lemma:cascade:whole-delay} %
			If regulators are configured as in Step~1, for each flow $f$ that goes through network element $S$,
a bound on the \ac{TAI} delay of the flow
through the concatenation of $S$ and the next regulator is
$D'^{f,\hze}= \rho^2 D^{f,\hze} + \eta (1+\rho)$ where $D^{f,\hze}$ is a bound on the \ac{TAI} delay of flow $f$ through $S$, computed in Step~2.
		\end{proposition} %
		The proof is in Appendix~\ref{proof:cascade:whole-delay}. Intuitively, with the clock of the regulator, the arrival curve of the flow at the input of $S$ is $\leq$ the shaping curve of the regulator, due to the inflation of rate and burst; hence shaping-for-free holds with this clock.
	
		\ant Assume that the regulators' configurations have infinite precision ($\mathcal{R},\mathcal{Q}$ are the identity function).	With these settings, the rate and burst cascade method increases the rate [resp the burst] of the flow by approx. 0.02\% [resp less than one bit for most flows] at each hop.
		%Note that when the regulator is an \ac{IR}, several flows coming from several sources, each with an independent clock, are merged together as in Figure~\ref{fig:ir-prob} and their interaction must be taken into account. Still, with the notations of Figure~\ref{fig:ir-prob}, we can note that for any source $j$ and any flow $p$, the cascade method ensures  \todo{Mystère?} $\alpha_{f_{j,p}}^{\mathcal{H}_{\text{IR}}} \le \sigma_{f_{j,p}}^{\mathcal{H}_{\text{IR}}}$, and we can apply the shaping-for-free property of the \ac{IR} when observed with $\mathcal{H}_{\text{IR}}$.
%		\todo{ou retirer phrase ci-dessus ?s}
		
		The rate and burst cascade has the drawback that the configuration of a regulator depends on its position on the flow path. This puts complexity on the control plane, specifically, for computing, distributing and managing the configuration.
%		
%		Furthermore, it leads to a pessimistic reservation of rates. Indeed, every hop is conservatively configured so as to be stable even with the worst possible source and clock. Observed with clock $\mathcal{H}_{\text{Reg}_0}$, the rate configured at hop $k$ is amplified by a factor of $\rho^{k-1}$.
%
% and source even if the previous regulator is greedy (it sends packet exactly at its committed rate) and even if the clock of the previous regulator is too fast relatively to the next one, then the considered hop remains stable. The configured rate of the previous regulator is greater than the rate of the flow at its source: observed with clock $\mathcal{H}_{\text{Reg}_0}$, the ratio is at least $\rho^{k-1}$, with $k$ the index of the regulator. However the previous regulator cannot be greedy for eternity, because the long-term rate of the flow cannot be more than $r_{0}$ (the source rate) when observed with $\h_{\text{Reg}_0}$, i.e. there is overprovision of rates. \todo{Is this true ? Do we reduced the reserved rate ? Explain the imnpact on rate reservation.}
%
This motivates us to propose next an alternative method which, however, works only with \acp{PFR}.
		
	\subsection{The ADAM method for Non-Synchronized Networks with Per-Flow Regulators}
\label{sec:async:ADAM}
	
		The goal of the \acf{ADAM} is, for any given flow, to have the same parameters at all regulators along the flow path. Thus,when applying \ac{ADAM}, we require that the rounding functions $\mathcal{R}_k,\mathcal{Q}_k$ (which capture the accuracy with which regulator parameters are actually implemented) are the same at all network nodes (and we consequently drop the index $k$ for these functions). We also require that all the regulators in the network be \acp{PFR}.

		The main idea of \ac{ADAM} is to establish that each flow $f$ has an arrival curve, expressed in \ac{TAI}, of the form
	$\alpha_k=\alpha_{1}\wedge\alpha_{2,k}$ at the output of its $k$-th hop where $\alpha_{1}$, $\alpha_{2,k}$ are leaky-bucket arrival curves and the former is independent of the hop index.
%
%is to consider the arrival curves $\lbrace\alpha_k\rbrace_k$ as being each the minimum of two Leaky-Bucket arrival curves: $\forall k$,
%
%$\alpha_k=\alpha_{1,k}\otimes\alpha_{2,k}$ and to use two different methods in propagating each of the members $\alpha_{1,k}$, $\alpha_{2,k}$.
%Consider a flow of interest, as inThe method uses a parameter, the rate margin $W$, that must be chosen such that $W\ge \rho^2$ and $\mathcal{R}(Wr_0)=Wr_0$. It is as follows:
		
		\textbf{Step~1}: For each flow $f$, find a rate margin $W$ such that $W\ge \rho^2$ and $Wr_0$ can be exactly implemented, i.e $\mathcal{R}(Wr_0)=Wr_0$. Configure the shaping curves at all regulators along the path of the flow with rate $r_{\text{Reg}_{k}} = Wr_{0}$ and burst $b_{\text{Reg}_{k}} = b_{0}$. Since $\rho$ is a network-wide parameter that does not depend on a clock, all regulators except the source have the same configuration, independent of the hop index $k=1...n$. Here $r_0,b_0$ are the rate and burst at the source (which depend on $f$, though the dependency on $f$ is not shown for simplicity of notation).
		
		\textbf{Step~2}: Any flow $f$ that goes through a network element $S$ is output by the regulator at its previous hop. Thus, using the same justification as step 2 of the rate and burst cascade, it has the arrival curve $\alpha_1=\gamma_{\rho W r_{0},b_{0}+\eta Wr_{0} }$ (i.e. leaky bucket with rate $\rho W r_{0}$ and burst $b_{0}+\eta Wr_{0}$), when observed with $\hze$.
 		
		Then, again using the same justification as step 2 of the rate and burst cascade, compute a \ac{TAI} delay bound $D_S^{f,\hze}$ through any network element $S$ and for any flow $f$ that goes through it.	
					
		\textbf{Step~3}: For each flow $f$ and each hop $k=1\ldots n$ in its path, compute a \ac{TAI} delay bound, $D_k'^{\hze}$, of the flow through the sequence $S_k-\textmd{Reg}_k$, using Algorithm~\ref{algo:adam}. \begin{algorithm}[t]
	\caption{\label{algo:adam} Computing the TAI delay for a flow through its hop $m$ using the \ac{ADAM} method}
	\begin{algorithmic}[1]
		\Require $\lbrace D_k^{\hze} \rbrace_k$, the set of all the \ac{TAI} delay bounds for the flow through the systems S$_k$ in its path (from Step~2).
		\Require $m$, the index at which to compute the \ac{TAI} hop delay.
		\Require $r_0$, $b_0$, the rate and burst of the flow at the source, observed with the source's clock. $W$, the rate margin.
		\Require $\eta$, $\rho$ the network-wide parameters of the time-model.
		\Function{ComputeDelayHopM}{$\lbrace D_k^{\hze} \rbrace_k$,$m$,$(r_0,b_0)$,$W$}
		
			\State $r_2 \gets \rho r_{0}$\label{algo:adam:b20a}
			\State $b_{2,0}\gets b_{0}+\eta r_0$ \label{algo:adam:b20b}
			\For{$k=1\ldots m$}
				\State {$D_k'^{\hze} \gets D_k^{\hze} + \eta (1+\rho) + \frac{b_{2,k-1}-b_{0}-\eta Wr_0}{\rho r_0} \frac{\rho^2-1}{W - 1}$}\label{algo:adam:delay}
				\State {$b_{2,k} \gets b_{2,k-1} + \rho r_{0} \cdot D_k'^{\hze}$}\Comment{See Proposition~\ref{prop:adam}}\label{algo:adam:a2}
			\EndFor
		\State \Return $D_m'^{\hze}$
		\EndFunction
	\end{algorithmic}
\end{algorithm}

		%The burst defined by recurrence in Line~\ref{algo:adam:a2} describes also, combined with the rate of Line~\ref{algo:adam:b20a}, a leaky-bucket arrival curve for any $k\ge 1$, that we note $\alpha_{2,k}$.

%
%
%If $\alpha_{2,k-1}$ is an arrival curve at the input of hop $k$ for the flow observed with $\hze$, then Proposition~\ref{lemma:adam:whole-delay} proves that the value computed in Line~\ref{algo:adam:delay} is a \ac{TAI} delay bound for the whole hop, and that $\alpha_{2,k}$ defined with the burst of Line~\ref{algo:adam:a2} is an arrival curve of the flow observed with $\hze$ at the input of the next hop.
		
\begin{proposition}[Correctness of Algorithm~\ref{algo:adam}]\label{prop:adam} For a flow $f$ that has $n$ hops and for $m=1\ldots n$: %
\begin{enumerate} %
  \item Let $\alpha_{2,0}$ be the leaky-bucket function with rate $r_2$ (Line~\ref{algo:adam:b20a}) and burst $b_{2,0}$ (Line~\ref{algo:adam:b20b}); it is an arrival curve for the flow at its source when observed with $\hze$.
  \item Let $\alpha_{2,k}$ be the leaky-bucket function with rate $r_2$ (Line~\ref{algo:adam:b20a}) and burst $b_{2,k}$ (Line~\ref{algo:adam:a2}). It is an arrival curve of the flow observed with $\hze$ at the output of $\text{Reg}_{k}$.
  \item $D_k'^{\hze}$ is a \ac{TAI} delay bound for the flow through the concatenation $S_k-\text{Reg}_k$, for $k=1\ldots m$. %
\end{enumerate} %
\end{proposition} %
%\todo{Begin Holly-Check}
		The proof is in Appendix~\ref{proof:adam:whole-delay}. Observe that, unlike with the rate and burst cascade method, here the regulators do increase (slightly) the delay bound, specifically, shaping-for-free does not hold. The proof captures this increase by using a service-curve characterization of \acp{PFR}, together with the arrival curve $\alpha_{1}\wedge\alpha_{2,k-1}$
for the flow at the input of $S_k$ observed with $\hze$. Then it applies the Network Calculus results using these curves observed with $\hze$. As we do not know any service-curve characterization for \acp{IR}, we are not able to extend this method to \acp{IR}.

Also note that the delay bound in Step~2 is computed using the arrival curve $\alpha_1$ and not the full arrival curve known by the method. This is because Step~2 is performed at every node in the network, before knowing the results of Step~3 that is performed per-flow. Using the result of Step~3 in Step~2 is possible in feed-forward networks and might in some cases lead to slightly smaller delay bounds. However, one of the major applications of regulators is in non-feedforward networks. Therefore, we do not explore such possible optimizations in this paper.

%\todo{End Holly-Check}

		\ant %\todo{or new subsection 'comparison between...'}
			We consider a network with one unique flow going through $n$ nodes, each with a \ac{FIFO} system and a regulator. Figure~\ref{fig:comparaison} compares, for one example and for both \acs{ADAM} and the rate and burst cascade, the increase of the \acl{ETE} \ac{TAI} delay bound with respect to the theoretical situation with ideal clocks. The delay-bound increase is larger with \acs{ADAM} but in both cases it is very small (less than 4\% for paths smaller than 10 hops). %\todo{Configured rates}

% !TeX spellcheck = en_US
% !TeX rootfile = article.tex
\section{%Management and analysis of s
Synchronized networks with regulators}
\label{sec:manag-sync}		
In synchronized networks, we expect that unbounded delays due to non-adapted regulators cannot occur, as clock rates cannot diverge for long periods of time. We now examine to which extent this holds. We study separately
\acp{PFR} and \acp{IR}.
\begin{figure} %
	\resizebox{\linewidth}{!}{\input{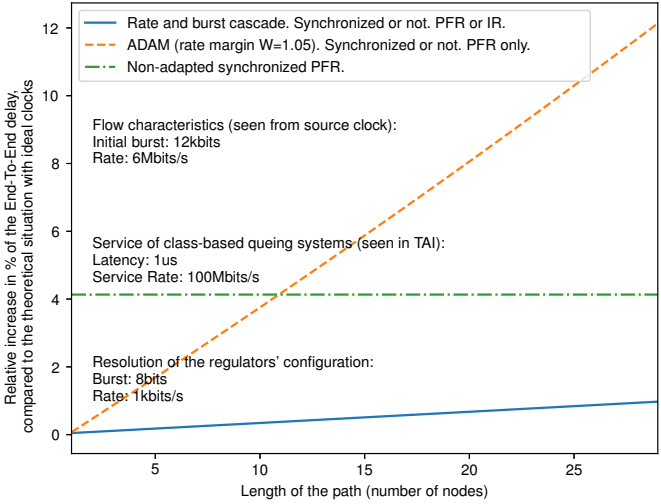}} %
	\caption{\label{fig:comparaison} Comparison of the relative increase of the \ac{TAI} \acl{ETE} delay bound of a flow obtained with several methods, with respect to the theoretical situation with ideal clocks, and as a function of the path length.} %
\end{figure}

	\subsection{Delay Penalty of Non-Adapted \acp{PFR} in Synchronized Networks}\label{sec:sync:unadapted-pfr}
	
	Even when synchronized, a non-adapted \ac{PFR} increases the worst-case delay, i.e. does not ``shape for free'':
	\begin{proposition}[Lower bound on the worst-case delay penalty of synchronized non-adapted \aclp{PFR}]\label{prop:sync:pfr-lower}
		For any leaky-bucket arrival curve $\gamma_{r,b}$, there exists a network configuration such that (1) one flow has the arrival curve $\gamma_{r,b}$ at the source, (2) the flow goes through one network element followed by one non-adapted \ac{PFR} (hence with shaping curve $\gamma_{r,b}$), (3) the clocks of the source and the \ac{PFR} are synchronized with time-error bound $\Delta$, and (4) the \ac{PFR} increases the worst-case delay of the flow by at least $\Delta$.
%
%such that the worst-delay increase due to the \ac{PFR} is $\Delta$
%adversarial clocks for $\h_{\text{Reg}_{0}}$ (the source) and $\h_{\text{PFR}_{1}}$ (the first \ac{PFR}), and adversarial traffic generation such that the source outputs the flow with the arrival curve $\gamma_{r,b}$, when observed with $\h_{\text{Reg}_{0}}$ and the delay penalty of the first non-adapted synchronized \ac{PFR} (hence with shaping curve $\gamma_{r,b}$) is $\Delta$ (i.e. the time error bound).
	\end{proposition} %
	The proof is in Appendix~\ref{proof:sync:pfr-lower}. We give here a less-formal example.
	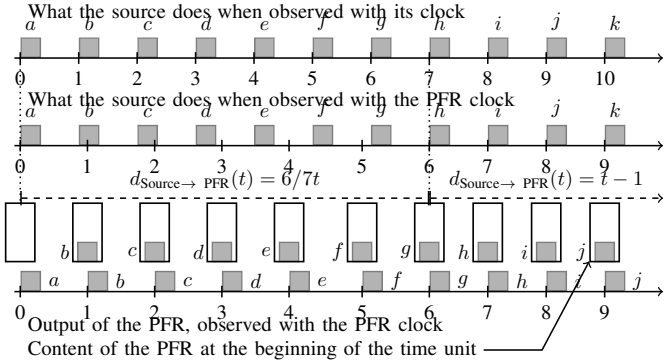
\begin{figure} %
		\resizebox{\linewidth}{!}{% !TeX spellcheck = en_US
% !TeX root=article.tex

\begin{tikzpicture}
	\tikzstyle{f} = [draw, regular polygon,regular polygon sides=4, anchor=south west, fill=black, opacity=0.3]
	
	\node[anchor=west] at (0,0.8) {What the source does when observed with its clock};
	\draw[->] (-0.1,0) -- (11,0);
	\draw[-] (0.,-0.1) -- (0,0.1) node[pos=0, below] {0};
	\draw[-] (1,-0.1) -- (1,0.1) node[pos=0, below] {1};
	\draw[-] (2,-0.1) -- (2,0.1) node[pos=0, below] {2};
	\draw[-] (3,-0.1) -- (3,0.1) node[pos=0, below] {3};
	\draw[-] (4,-0.1) -- (4,0.1) node[pos=0, below] {4};
	\draw[-] (5,-0.1) -- (5,0.1) node[pos=0, below] {5};
	\draw[-] (6,-0.1) -- (6,0.1) node[pos=0, below] {6};
	\draw[-] (7,-0.1) -- (7,0.1) node[pos=0, below] {7};
	\draw[-] (8,-0.1) -- (8,0.1) node[pos=0, below] {8};
	\draw[-] (9,-0.1) -- (9,0.1) node[pos=0, below] {9};
	\draw[-] (10,-0.1) -- (10,0.1) node[pos=0, below] {10};
	%\draw[-] (11,-0.1) -- (11,0.1) node[pos=0, below] {11};
	%\draw[-] (12,-0.1) -- (12,0.1) node[pos=0, below] {12};
	\node[f, label=above:$a$] at (0,0) {}; 
	\node[f, label=above:$b$] at (1,0) {};
	\node[f, label=above:$c$] at (2,0) {};
	\node[f, label=above:$d$] at (3,0) {};
	\node[f, label=above:$e$] at (4,0) {};
	\node[f, label=above:$f$] at (5,0) {};
	\node[f, label=above:$g$] at (6,0) {};
	\node[f, label=above:$h$] at (7,0) {};
	\node[f, label=above:$i$] at (8,0) {};
	\node[f, label=above:$j$] at (9,0) {};
	\node[f, label=above:$k$] at (10,0) {};

	\draw[->] (-0.1,-1.5) -- (11,-1.5);
	\node[anchor=west] at (0,-0.7) {What the source does when observed with the PFR clock};
	\draw[-] (0,-1.6) -- (0,-1.4) node[pos=0, below] {0};
	\draw[-] (1.15,-1.6) -- (1.15,-1.4) node[pos=0, below] {1};
	\draw[-] (2.30,-1.6) -- (2.30,-1.4) node[pos=0, below] {2};
	\draw[-] (3.45,-1.6) -- (3.45,-1.4) node[pos=0, below] {3};
	\draw[-] (4.60,-1.6) -- (4.60,-1.4) node[pos=0, below] {4};
	\draw[-] (5.85,-1.6) -- (5.85,-1.4) node[pos=0, below] {5};
	\draw[-] (7,-1.6) -- (7,-1.4) node[pos=0, below] {6};
	\draw[-] (8,-1.6) -- (8,-1.4) node[pos=0, below] {7};
	\draw[-] (9,-1.6) -- (9,-1.4) node[pos=0, below] {8};
	\draw[-] (10,-1.6) -- (10,-1.4) node[pos=0, below] {9};
	%\draw[-] (11,-1.6) -- (11,-1.4) node[pos=0, below] {10};
	%\draw[-] (12,-1.6) -- (12,-1.4) node[pos=0, below] {11};
	\node[f, label=above:$a$] at (0,-1.5) {};
	\node[f, label=above:$b$] at (1,-1.5) {}; 
	\node[f, label=above:$c$] at (2,-1.5) {}; 
	\node[f, label=above:$d$] at (3,-1.5) {}; 
	\node[f, label=above:$e$] at (4,-1.5) {}; 
	\node[f, label=above:$f$] at (5,-1.5) {}; 
	\node[f, label=above:$g$] at (6,-1.5) {}; 
	\node[f, label=above:$h$] at (7,-1.5) {};  
	\node[f, label=above:$i$] at (8,-1.5) {};
	\node[f, label=above:$j$] at (9,-1.5) {};
	\node[f, label=above:$k$] at (10,-1.5) {};
	
	\draw[dotted] (0,0) -- (0,-2.5);
	\draw[dotted] (7,0) -- (7,-2.5);
	\draw[|-|, dashed] (0,-2.4) -- (7,-2.4) node[pos=0.5,above] {$d_{\text{Source}\rightarrow \text{PFR}}(t) = 6/7 t$};
	\draw[|->, dashed] (7,-2.4) -- (11,-2.4) node[pos=0.5,above] {$d_{\text{Source}\rightarrow \text{PFR}}(t) = t-1$};

	\node (rect) at (0,-3.5) (ir1) [anchor=south, draw,thick,minimum width=0.5cm,minimum height=1cm] {};
	\node (rect) at (1.15,-3.5) (ir2) [anchor=south, draw,thick,minimum width=0.5cm,minimum height=1cm] {};
	\node (rect) at (2.30,-3.5) (ir3) [anchor=south, draw,thick,minimum width=0.5cm,minimum height=1cm] {};
	\node (rect) at (3.45,-3.5) (ir4) [anchor=south, draw,thick,minimum width=0.5cm,minimum height=1cm] {};
	\node (rect) at (4.60,-3.5) (ir5) [anchor=south, draw,thick,minimum width=0.5cm,minimum height=1cm] {};
	\node (rect) at (5.85,-3.5) (ir6) [anchor=south, draw,thick,minimum width=0.5cm,minimum height=1cm] {};
	\node (rect) at (7,-3.5) (ir7) [anchor=south, draw,thick,minimum width=0.5cm,minimum height=1cm] {};
	\node (rect) at (8,-3.5) (ir8) [anchor=south, draw,thick,minimum width=0.5cm,minimum height=1cm] {};
	\node (rect) at (9,-3.5) (ir9) [anchor=south, draw,thick,minimum width=0.5cm,minimum height=1cm] {};
	\node (rect) at (10,-3.5) (ir10) [anchor=south, draw,thick,minimum width=0.5cm,minimum height=1cm] {};
	
	\draw[->] (-0.1,-4) -- (11,-4);
	\node[anchor=west] at (0,-4.6) {Output of the PFR, observed with the PFR clock};
	\draw[-] (0,-4.1) -- (0,-3.9) node[pos=0, below] {0};
	\draw[-] (1.15,-4.1) -- (1.15,-3.9) node[pos=0, below] {1};
	\draw[-] (2.30,-4.1) -- (2.30,-3.9) node[pos=0, below] {2};
	\draw[-] (3.45,-4.1) -- (3.45,-3.9) node[pos=0, below] {3};
	\draw[-] (4.60,-4.1) -- (4.60,-3.9) node[pos=0, below] {4};
	\draw[-] (5.85,-4.1) -- (5.85,-3.9) node[pos=0, below] {5};
	\draw[-] (7,-4.1) -- (7,-3.9) node[pos=0, below] {6};
	\draw[-] (8,-4.1) -- (8,-3.9) node[pos=0, below] {7};
	\draw[-] (9,-4.1) -- (9,-3.9) node[pos=0, below] {8};
	\draw[-] (10,-4.1) -- (10,-3.9) node[pos=0, below] {9};
	
	\node[f, label=right:$a$] at (0,-4) {};
	\node[f, label=left:$b$, anchor=south] at (ir2.south) {};
	\node[f, label=right:$b$] at (1.15,-4) {}; 
	\node[f, label=left:$c$, anchor=south] at (ir3.south) {};
	\node[f, label=right:$c$] at (2.30,-4) {}; 
	\node[f, label=left:$d$, anchor=south] at (ir4.south) {};
	\node[f, label=right:$d$] at (3.45,-4) {}; 
	\node[f, label=left:$e$, anchor=south] at (ir5.south) {};
	\node[f, label=right:$e$] at (4.60,-4) {}; 
	\node[f, label=left:$f$, anchor=south] at (ir6.south) {};
	\node[f, label=right:$f$] at (5.85,-4) {}; 
	\node[f, label=left:$g$, anchor=south] at (ir7.south) {};
	\node[f, label=right:$g$] at (7,-4) {}; 
	\node[f, label=left:$h$, anchor=south] at (ir8.south) (h) {};
	\node[f, label=right:$h$] at (8,-4) {};  
	\node[f, label=left:$i$, anchor=south] at (ir9.south) (h) {};
	\node[f, label=right:$i$] at (9,-4) {}; 
	\node[f, label=left:$j$, anchor=south] at (ir10.south) (h) {};
	\node[f, label=right:$j$] at (10,-4) {}; 
	\node[anchor=west] at (0,-5) (ll) {Content of the PFR at the beginning of the time unit};
	\draw[->] (ll.east) -- ++(1cm,0) -- (ir10);
\end{tikzpicture}} %
		\caption{\label{fig:sync:timeline-pfr} Adversarial case achieving the lower penalty bound for synchronized non-adapted \acp{PFR}.} %
	\end{figure} %
	We take the same example as for Section~\ref{sec:manag-async:unstable} and Figure~\ref{fig:async:timeline-local}, with the source clock being the TAI, the source being greedy, and the system S$_k$ having no delay. But this time (Figure~\ref{fig:sync:timeline-pfr}), the regulator clock is too slow by $s_1 = 6/7$ between Time Units 0 and 7 (measured with the TAI), and then it regains the same speed, with a time function $d_{\text{Source}\rightarrow \text{PFR}}(t) = t - \Delta$. We exaggerate $\Delta$ to equal one unit of time. Then we can note that $d_{\text{Source}\rightarrow \text{PFR}}(t)$ meets Equations (\ref{eq:constr-async}) and (\ref{eq:const-sync}) with $\rho\ge s_1$.	We observe in Figure~\ref{fig:sync:timeline-pfr} that the TAI delay of each packet is increasing for packets from $a$ to $h$, and packets $h$ to $j$ have a delay of one unit, that is $\Delta$.
	
	Thanks to the synchronization, though, the delay penalty of the \ac{PFR} can be controlled:
	
	\begin{proposition}[Upper bound on the delay penalty of synchronized non-adapted \aclp{PFR}]\label{prop:sync:pfr-stable} Assume a flow is regulated at the source with rate $r_0$ and burst $b_0$ and goes through a sequence of $n$ concatenations of network elements and regulators. The regulators are nonadapted (i.e. have shaping curve $\gamma_{r_0,b_0}$), and the network is synchronized with time-error bound $\Delta$. Let $D_k^{\hze}$ be an upper bound on the \ac{TAI} delay at $S_k$, the $k$th network element on the path of the flow. A bound on the \ac{TAI} delay of the flow
through the concatenation of $S_k$ and the next regulator is $D_k'^{\hze} = D_k^{\hze} + 4\Delta$.
	\end{proposition} %	
	The proof is in Appendix~\ref{proof:sync:pfr-stable}. Note that a TAI delay bound $D_k^{\hze}$
 can be obtained by using the fact that the shaping curve is a valid arrival curve at the output of the previous regulator. Hence we can use the last line of Table~\ref{tab:ac-results}; when observed with $\hze$, the flow enters $S_k$ with a double arrival-curve constraint: a leaky-bucket arrival curve of rate $\rho r_0$ and burst $b_0+r_0\eta$, and a leaky-bucket arrival curve of rate $r_0$ and burst $b_0 + 2 r_0 \Delta$.
	
	It follows from Propositions~\ref{prop:sync:pfr-lower} and \ref{prop:sync:pfr-stable} that the worst-case penalty on the \ac{TAI} per-hop delay of non-adapted \acp{PFR} in synchronized networks is between $\Delta$ and $4\Delta$, i.e., is of the order of magnitude of the synchronization precision. For tightly-synchronized networks and \acp{PFR}, the current practice of ignoring clock nonidealities is thus perfectly acceptable. However, in loosely-synchronized networks, the value of $\Delta$ ($\sim 125$ms) is larger than the required delay bound for flows with stringent delay requirements. The two solutions (rate and burst cascade, and \ac{ADAM}) that apply to non-synchronized networks also apply here and can be used.
	
	\ant In Figure~\ref{fig:comparaison}, we compare, with the same conditions as in Section~\ref{sec:async:ADAM}, the delay bound with non-adapted tightly-synchronized \acp{PFR}, obtained with Proposition~\ref{prop:sync:pfr-stable} with the methods of Section~\ref{sec:manag-async}. We observe that they perform similarly, each being within a few percents of the ideal-clocks case.

% !TeX spellcheck = en_ZA
% !TeX root=article.tex

\subsection{Instability of Non-Adapted \acp{IR} in Synchronized Networks}\label{sec:ir-s-0}

	For the interleaved regulator, however, the conclusions are very different. Indeed, the \ac{IR} may yield unbounded latencies, even with tightly-synchronized networks, as shown in the following proposition.

	\begin{proposition}[Instability of non-adapted synchronized interleaved regulators]\label{prop:insta-sync-ir}
		Consider an interleaved regulator as in Figure \ref{fig:ir-prob}, with $n$ upstream systems. Each upstream system $j$ outputs $P_j$ flows $\lbrace f_{j,p}\rbrace_{p=1}^{P_j}$, each with a known arrival curve $\alpha_{f_{j,p}}^{\mathcal{H}_j}$ when observed with clock $\h_j$. Assume the interleaved regulator is non-adapted: $\forall(j,p), \sigma_{f_{j,p}}^{\mathcal{H}_{\text{IR}}} =  \alpha_{f_{j,p}}^{\mathcal{H}_{j}}$. Finally, assume that the clocks $\mathcal{H}_{\text{IR}}$,$\mathcal{H}_{\text{FIFO}}$ and each of the $\mathcal{H}_j$ are synchronized (as per  Equations~(\ref{eq:constr-async}) and (\ref{eq:const-sync})).
		Then, for any parameters $n, \rho,\eta,\Delta$ with $n\ge3$, $\rho>1$, $\eta\ge0$ and $\Delta>0$, there exists a \ac{FIFO} system, adversarial clocks for $\mathcal{H}_{\text{IR}}$,$\mathcal{H}_{\text{FIFO}}$ and $\lbrace\mathcal{H}_j\rbrace_j$, and adversarial traffic generation of the upstream systems, such that the flows crossing the IR have unbounded latency within the IR, when observed with any clock of the network.
	\end{proposition}
%\todo{Begin Holly-Check}
	The proof is in Appendix~\ref{proof:insta-sync-ir}. We give now an informal example of a missed deadline. Consider Figure~\ref{fig:contract} that gives the traffic profile that is expected to enter the previous FIFO system. It corresponds to two flows, each one with a leaky-bucket arrival curve with a burst of one packet and a rate of 1/2 packet per unit of time. Assume that the \ac{IR} is configured with this contract.
	Assume that System 1 generates Flow $1$ and System 2 generates Flow $2$. Each System believes, according to its clock, that it generates in Figure~\ref{fig:reality}, its flow's packets as per the expected behaviour of Figure~\ref{fig:contract}. But the clock of System 2 suddenly increases its speed relative to the \ac{IR} clock between Units 1 and 3 (Figure~\ref{fig:reality}). When seen from the clock of the \ac{IR} (here assumed to share the same clock as system 1), this sudden speeding of Clock 2 increases the burst of Flow $2$. Even worse, it allows Packet $2b$ to squeeze in before Packet $1b$, when $2b$ was supposed to come after, as per Figure~\ref{fig:contract}.
	
	Assume a simulation of the \ac{FIFO} system that provides the output given in Figure~\ref{fig:outfifo-reality}. The worst-case delay through the \ac{FIFO} is five units of \ac{TAI} time, reached by Packet $1a$. For Packet $1b$ to meet this deadline for the entire hop (FIFO+IR), it must be released no later than Time Unit 7. But recall that the \ac{IR} looks only at the head-of-line packet and all the following packets, even from other flows, need to wait for this head packet to be released before being processed by the \ac{IR}. $1b$, blocked by $2b$, is released no earlier than Time Unit 8. Thus, $1b$ has missed its deadline.
	
	\ant In the proof, the delay divergence increases at a rate $\sqrt{\rho}-1$ for any number of previous sources $n\ge3$, and any synchronization precision $\Delta>0$. This corresponds to $100\mu$s of increased worst-case delay per second of network operation for both tightly- or loosely-synchronized networks.
	
\begin{figure}\centering %
\resizebox{0.65\linewidth}{!}{\begin{tikzpicture}
	\tikzstyle{mybox} = [draw, minimum height=1cm]
	
	\node[mybox] at (0,0) (fifo) {FIFO system};
	\node[anchor=north] at (fifo.south east) {$\mathcal{H}_{\text{FIFO}}$};

	\node[mybox, anchor=west] at ([xshift=0.7cm] fifo.east) (pfs) {IR};
	\node[anchor=north] at (pfs.south east) {$\mathcal{H}_{\text{IR}}$};
	
	\node[draw, dashed, line width=0.5pt, anchor=south] (irconf) at ([yshift=0.5cm] pfs.north) {$\lbrace  \sigma_{f_{j,p}}^{\mathcal{H}_{\textbf{IR}}} = \alpha_{f_{j,p}}^{\mathcal{H}_j}\rbrace_{j,p}$};
	\node[anchor=north] at (irconf.south east) {$\mathcal{H}_{\text{IR}}$};
	\draw[->, dashed] (irconf) -- (pfs);

	\draw[->] ([yshift=1.2cm]fifo.west)++(-1cm,0) -- ([yshift=0.2cm]fifo.west) node[pos=0, left, draw, line width=0.5pt] (s1) {$\lbrace  \alpha_{f_{1,p}}^{\mathcal{H}_1}\rbrace_p$}; 
	\node[anchor=north] at (s1.south east) {$\mathcal{H}_1$};
	%n
	\draw[->] ([yshift=-0.2cm]fifo.west)++(-1cm,0) -- ([yshift=-0.2cm]fifo.west) node[pos=0, left, draw, line width=0.5pt] (s2) {$\lbrace  \alpha_{f_{n,p}}^{\mathcal{H}_n}\rbrace_p$}; 
	\node[anchor=north] at (s2.south east) {$\mathcal{H}_n$};
	\node[anchor=east, align=right] at (-2.5,0.5) {\ldots};
	
	\draw[->] (fifo) -- (pfs);
	\draw[->] (pfs.east) -- ++(1cm,0);
\end{tikzpicture}} %
\caption{\label{fig:ir-prob} Adversarial situation with several sources.}
\end{figure}
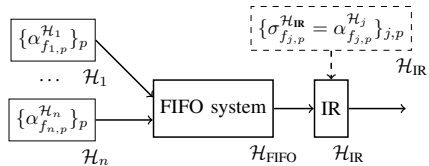
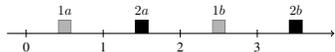
\begin{figure}\centering
\resizebox{0.5\linewidth}{!}{% !TeX spellcheck = en_US
% !TeX root=article.tex

\begin{tikzpicture}
	\tikzstyle{f} = [draw, regular polygon,regular polygon sides=4, anchor=south, fill=black, opacity=0.3]
	\tikzstyle{ff} = [draw, regular polygon,regular polygon sides=4, anchor=south, fill=black, opacity=1]
	\draw[->] (-0.5,0) -- (8,0);
	\draw[-] (0.,-0.1) -- (0,0.1) node[pos=0, below] {0};
	\draw[-] (2,-0.1) -- (2,0.1) node[pos=0, below] {1};
	\draw[-] (4,-0.1) -- (4,0.1) node[pos=0, below] {2};
	\draw[-] (6,-0.1) -- (6,0.1) node[pos=0, below] {3};
	\node[f, label=above:$1a$] at (1,0) {}; 
	\node[ff, label=above:$2a$] at (3,0) {};
	\node[f, label=above:$1b$] at (5,0) {};
	\node[ff, label=above:$2b$] at (7,0) {};
\end{tikzpicture}}
\caption{\label{fig:contract}Expected input of the previous \ac{FIFO} system.}
\end{figure}
\begin{figure}\centering
\begin{minipage}{0.48\linewidth}\centering
	\resizebox{0.8\linewidth}{!}{% !TeX spellcheck = en_US
% !TeX root=article.tex

\begin{tikzpicture}
	\tikzstyle{f} = [draw, regular polygon,regular polygon sides=4, anchor=south,xshift=0.1cm, fill=black, opacity=0.3]
	\tikzstyle{ff} = [draw, regular polygon,regular polygon sides=4, anchor=south,xshift=0.1cm, fill=black, opacity=1]
	\draw[->] (-1,0) -- (8,0);
	\node[anchor=west] at (-1,0.8) {Output system 1};
	\node[anchor=center] at (7,0) {\makecell[r]{$\h_{1}$\\$\h_{\text{IR}}$}};
	\draw[-] (0,-0.1) -- (0,0.1) node[pos=0, below] {0};
	\draw[-] (2,-0.1) -- (2,0.1) node[pos=0, below] {1};
	\draw[-] (4,-0.1) -- (4,0.1) node[pos=0, below] {2};
	\draw[-] (6,-0.1) -- (6,0.1) node[pos=0, below] {3};
	\node[f, label=above:$1a$] at (1,0) {};
	%\node[ff, label=above:2a] at (3,0) {};
	\node[f, label=above:$1b$] at (5,0) {};
	%\node[ff, label=above:2b] at (7,0) {};
	
	\draw[->] (-1,-1.5) -- (8,-1.5);
	\node[anchor=west] at (-1,-0.7) {Output system 2};
	\node[anchor=south] at (7,-1.5) {$\h_{2}$};
	\draw[-] (0,-1.6) -- (0,-1.4) node[pos=0, below] {0};
	\draw[-] (1,-1.6) -- (1,-1.4) node[pos=0, below] {1};
	\draw[-] (3.3,-1.6) -- (3.3,-1.4) node[pos=0, below] {2};
	\draw[-] (4.1,-1.6) -- (4.1,-1.4) node[pos=0, below] {3};
	%\node[f, label=above:1a] at (1,0) {};
	\node[ff, label=above:$2a$] at (2.5,-1.5) {};
	%\node[f, label=above:1b] at (5,0) {};
	\node[ff, label=above:$2b$] at (4.3,-1.5) {};

	\draw[->] (-1,-3) -- (8,-3);
	\node[anchor=west] at (-1,-2.2) {Output systems 1 and 2};
	\node[anchor=south] at (7,-3) {$\h_{\text{IR}}$};
	\draw[-] (0,-3.1) -- (0,-2.9) node[pos=0, below] {0};
	\draw[-] (2,-3.1) -- (2,-2.9) node[pos=0, below] {1};
	\draw[-] (4,-3.1) -- (4,-2.9) node[pos=0, below] {2};
	\draw[-] (6,-3.1) -- (6,-2.9) node[pos=0, below] {3};
	\node[f, label=above:$1a$] at (1,-3) {};
	\node[ff, label=above:$2a$] at (2.5,-3) {};
	\node[f, label=above:$1b$] at (5.2,-3) {};
	\node[ff, label=above:$2b$] at (4.3,-3) {};
\end{tikzpicture}}
	\subcaption{\label{fig:reality}Output of each source.}
\end{minipage}\hspace{0.01\linewidth}
\begin{minipage}{0.48\linewidth}\centering
	\resizebox{0.8\linewidth}{!}{% !TeX spellcheck = en_US
% !TeX root=article.tex

\begin{tikzpicture}
	\tikzstyle{f} = [draw, regular polygon,regular polygon sides=4, anchor=south west,xshift=0.1cm, fill=black, opacity=0.3]
	\tikzstyle{ff} = [draw, regular polygon,regular polygon sides=4, anchor=south west,xshift=0.1cm, fill=black, opacity=1]
	
	\draw[->] (-1,0) -- (8,0);
	\draw[-] (0,-0.1) -- (0,0.1);
	\draw[-] (2,-0.1) -- (2,0.1);
	\draw[-] (4,-0.1) -- (4,0.1);
	\draw[-] (6,-0.1) -- (6,0.1);
	\node[f, label=above:$1a$] at (0,0) {}; 
	\node[ff, label=above:$2a$] at (2,0) {};
	\node[draw, regular polygon, regular polygon sides=4, anchor=south east, fill=black, opacity=0.3, label=above:$1b$] at (3.8,0) {};
	\node[draw, regular polygon, regular polygon sides=4, anchor=south, fill=black, opacity=1, label=above:$1b$] at (3,0) {};
	
	\node[anchor=south] at (7,0) {$\h_{\text{IR}}$};
	
	\draw[->] (-1,-2) -- (8,-2);
	\draw[-] (0,-2.1) -- (0,-1.9) node[pos=0, below] {5};
	\draw[-] (2,-2.1) -- (2,-1.9) node[pos=0, below] {6};
	\draw[-] (4,-2.1) -- (4,-1.9) node[pos=0, below] {7};
	\draw[-] (6,-2.1) -- (6,-1.9) node[pos=0, below] {8};
	\node[f, label=above:$1a$] at (0,-2) {}; 
	\node[ff, label=above:$2a$] at (2,-2) {};
	\node[f, label=above:$1b$] at (6.8,-2) {};
	\node[ff, label=above:$2b$] at (6,-2) {};
	
	\node (rect) at (0,-1.2) (ir1) [anchor=south, draw,thick,minimum width=0.5cm,minimum height=1cm] {};
	\node (rect) at (2,-1.2) (ir) [anchor=south, draw,thick,minimum width=0.5cm,minimum height=1cm] {};
	\node (rect) at (4,-1.2) [anchor=south, draw,thick,minimum width=0.5cm,minimum height=1cm] {};
	\node (rect) at (6,-1.2) [anchor=south, draw,thick,minimum width=0.5cm,minimum height=1cm] {};
	\draw[<-] (3.3,-1.2) -- (3.3, -0.2) node[pos=0.5,sloped,above] {FIFO};
	\node[f, anchor=south, xshift=-0.1cm, label=right:$1b$] at (4,-0.7) {};
	\node[ff, anchor=south, xshift=-0.1cm, label=right:$2b$] at (4,-1.2) (hol) {};
	\node[ff, anchor=south, xshift=-0.1cm, label=right:$2b$] at (6,-1.2) {};
	\node[f, anchor=south, xshift=-0.1cm, label=right:$1b$] at (6,-0.7) {};
	
	%\node at (6.5,0.4) (holt) {Head-of-line packet};
	%\draw[->] (holt.west) -- ++(-0.2cm,0) -- ([xshift=0.7cm]hol.east) -- ([xshift=0.5cm]hol.east);
	
	%\node at (4,1) (irt) {\makecell[l]{Content of the \ac{IR} at\\the beginning of the time unit}};
	%\draw[->] (irt.west) -- ++(-0.2cm,0) -- (ir1.north east);
\end{tikzpicture}}
	\subcaption{\label{fig:outfifo-reality}Input (above) and output (below) of the \ac{IR}, with $\h_{\text{IR}}$.}
\end{minipage}
\caption{\label{fig:case-ir} Example of a missed deadline caused by nonideal clocks. Packet $1b$ misses its deadline (5 units of time).}
\end{figure}
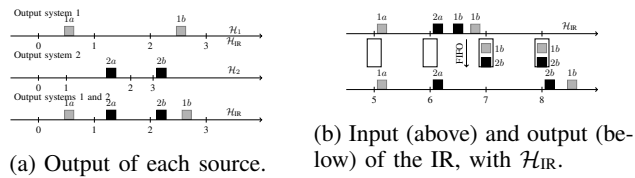

%\input{old-management}

% !TeX spellcheck = en_US
\section{conclusion}
\label{sec:conclu}
We have developed a theory for adding clock nonidealities to network calculus. We have applied the theory to time-sensitive networks with regulators and have obtained the following conclusions.

In loosely-synchronized and in non-synchronized networks, regulators are affected by clock issues; this leads to large or unbounded delays if not properly addressed. We have proposed two solutions for \acp{PFR}: the rate and burst cascade, and ADAM. The former imposes that the regulator parameters depend on the position of the regulator on the flow path, which complexifies the control plane. For interleaved regulators, only the rate and burst cascade applies.

In tightly-synchronized networks, per-flow regulators are affected, but the delay penalty is of the order of the time-error bound and can be neglected. In contrast, interleaved regulators are affected well beyond the time-error bound, and this can lead to unbounded delays if the issue is not properly addressed. Therefore, a solution such as the method of rate and burst cascade should be applied with interleaved regulators, even in tightly-synchronized networks.%\todo{End Holly-Check} 

In this paper, we proposed a network calculus toolbox for networks with nonideal clocks, we then focused on applying the toolbox on \acf{ATS}. Computing the effect of nonideal clocks on other mechanisms and non-work conservatives scheduler such as \ac{TAS}, \ac{CBS}, the damper, constitute a first interesting axis of future work.

We also proposed methods to manage the regulator parameters of \ac{ATS} within a network, focusing on delay bounds obtained with Network Calculus. The evaluation of these methods under realistic network and clock conditions, using simulations constitute another axis of future work.
\section*{Acknowledgements}
	This work was supported by Huawei Technologies Co., Ltd. in the framework of the project \emph{Large Scale Deterministic Network}. The authors thank Bryan Liubingyang for his contribution to the idea of this paper and Ahlem Mifdaoui for many fruitful discussions on an early version.

% The next two lines define the bibliography style to be used, and
% the bibliography file.
%\bibliographystyle{ACM-Reference-Format}
\bibliography{bibli,bibli-zotero}

\begin{thebibliography}{10}

\bibitem{iecIECIEEE608022019}
IEC and IEEE, ``{{IEC}}/{{IEEE}} 60802 - {{Time}}-{{Sensitive Networking
  Profile}} for {{Industrial Automation}},'' vol.~IEC/IEEE 60802 (D1.1), 2019.
\newblock
  \url{http://www.ieee802.org/1/files/private/60802-drafts/d1/60802-d1-1.pdf}.

\bibitem{AFDX}
A.~Committee {\em et~al.}, ``{Aircraft Data Network Part 7, Avionics Full
  Duplex Switched Ethernet (AFDX) Network, ARINC Specification 664},'' {\em
  Annapolis, Maryland: Aeronautical Radio}, 2002.

\bibitem{TTE}
H.~Kopetz, A.~Ademaj, P.~Grillinger, and K.~Steinhammer, ``{The Time-Triggered
  Ethernet (TTE) Design},'' in {\em Object-Oriented Real-Time Distributed
  Computing, 2005. ISORC 2005. Eighth IEEE International Symposium on},
  pp.~22--33, IEEE, 2005.

\bibitem{ecssSpaceWireLinksNodes2008}
ECSS, ``{{SpaceWire}} \textendash{} {{Links}}, nodes, routers and networks (31
  {{July}} 2008) | {{European Cooperation}} for {{Space Standardization}},''
  {\em ECSS-E-ST-50-12C}, July 2008.
\newblock
  \url{https://ecss.nl/standard/ecss-e-st-50-12c-spacewire-links-nodes-routers-and-networks/}.

\bibitem{ieeeDraftStandardLocal2019b}
IEEE, ``Draft {{Standard}} for {{Local}} and metropolitan area networks
  \textemdash{} {{Time}}-{{Sensitive Networking Profile}} for {{Automotive
  In}}-{{Vehicle Ethernet Communications}},'' {\em IEEE
  P802.1DG\texttrademark/D1.1}, vol.~In IEEE802.1 private repository. To obtain
  the access credentials, visit
  https://www.ietf.org/proceedings/52/slides/bridge-0/tsld003.htm or contact
  the IEEE802.1 chair., Oct. 2019.
\newblock
  \url{http://www.ieee802.org/1/files/private/dg-drafts/d1/802-1DG-d1-1.pdf}.

\bibitem{ieeeIEEEStandardLocal2018}
IEEE, ``{{IEEE Standard}} for {{Local}} and {{Metropolitan Area
  Networks}}\textendash{{Bridges}} and {{Bridged Networks}} \textendash{}
  {{Amendment}} 31: {{Stream Reservation Protocol}} ({{SRP}}) {{Enhancements}}
  and {{Performance Improvements}},'' {\em IEEE Std 802.1Qcc-2018 (Amendment to
  IEEE Std 802.1Q-2018 as amended by IEEE Std 802.1Qcp-2018)}, pp.~1--208, Oct.
  2018.

\bibitem{ieeeIEEEStandardLocal2017}
IEEE, ``{{IEEE Standard}} for {{Local}} and metropolitan area
  networks\textendash{{Frame Replication}} and {{Elimination}} for
  {{Reliability}},'' {\em IEEE Std 802.1CB-2017}, pp.~1--102, Oct. 2017.

\bibitem{wandeler2006performance}
E.~Wandeler, A.~Maxiaguine, and L.~Thiele, ``Performance analysis of greedy
  shapers in real-time systems,'' in {\em Proceedings of the Design Automation
  \& Test in Europe Conference}, vol.~1, pp.~6--pp, IEEE, 2006.

\bibitem{specht2016urgency}
J.~Specht and S.~Samii, ``Urgency-based scheduler for time-sensitive switched
  ethernet networks,'' in {\em Real-Time Systems (ECRTS), 2016 28th Euromicro
  Conference on}, pp.~75--85, IEEE, 2016.

\bibitem{mohammadpourLatencyBacklogBounds2018}
E.~Mohammadpour, E.~Stai, M.~Mohiuddin, and J.~Le~Boudec, ``Latency and
  {{Backlog Bounds}} in {{Time}}-{{Sensitive Networking}} with {{Credit Based
  Shapers}} and {{Asynchronous Traffic Shaping}},'' in {\em 2018 30th
  {{International Teletraffic Congress}} ({{ITC}} 30)}, vol.~02, pp.~1--6,
  Sept. 2018.
\newblock \url{http://doi.org/10.1109/ITC30.2018.10053}.

\bibitem{leboudecNetworkCalculusTheory2001}
J.-Y. Le~Boudec and P.~Thiran, {\em Network {{Calculus}}: {{A Theory}} of
  {{Deterministic Queuing Systems}} for the {{Internet}}}.
\newblock Lecture {{Notes}} in {{Computer Science}}, {{Lect}}.{{Notes
  Computer}}. {{Tutorial}}, {Berlin Heidelberg}: {Springer-Verlag}, 2001.
\newblock \url{https://www.springer.com/us/book/9783540421849}.

\bibitem{leboudecTheoryTrafficRegulators2018}
J.-Y. Le~Boudec, ``A {{Theory}} of {{Traffic Regulators}} for {{Deterministic
  Networks With Application}} to {{Interleaved Regulators}},'' {\em IEEE/ACM
  Transactions on Networking}, vol.~26, pp.~2721--2733, Dec. 2018.
\newblock \url{http://doi.org/10.1109/TNET.2018.2875191}.

\bibitem{wagner2001short}
K.~Wagner, ``Short evaluation of linux's token-bucket-filter (tbf) queuing
  discipline,'' {\em
  http://www.docum.org/stef.coene/qos/docs/other/tbf02\_kw.ps}, 2001.

\bibitem{ieeeIEEEStandardPrecision2008}
IEEE, ``{{IEEE Standard}} for a {{Precision Clock Synchronization Protocol}}
  for {{Networked Measurement}} and {{Control Systems}},'' {\em IEEE Std
  1588-2008 (Revision of IEEE Std 1588-2002)}, pp.~1--300, July 2008.

\bibitem{moreiraWhiteRabbitSubnanosecond2009}
P.~Moreira, J.~Serrano, T.~Wlostowski, P.~Loschmidt, and G.~Gaderer, ``White
  rabbit: {{Sub}}-nanosecond timing distribution over ethernet,'' in {\em
  Control and {{Communication}} 2009 {{International Symposium}} on {{Precision
  Clock Synchronization}} for {{Measurement}}}, pp.~1--5, Oct. 2009.

\bibitem{powersGPSGalileoUTC2004}
E.~Powers and J.~Hahn, ``{{GPS}} and {{Galileo UTC}} time distribution,'' in
  {\em 2004 18th {{European Frequency}} and {{Time Forum}} ({{EFTF}} 2004)},
  pp.~484--488, Apr. 2004.

\bibitem{murtaQRPp14CharacterizingQuality2006}
C.~D. Murta, P.~R. Torres~Jr., and P.~Mohapatra, ``{{QRPp1}}-4:
  {{Characterizing Quality}} of {{Time}} and {{Topology}} in a {{Time
  Synchronization Network}},'' in {\em {{IEEE Globecom}} 2006}, pp.~1--5, Nov.
  2006.

\bibitem{ieeeDraftStandardLocal2019a}
IEEE, ``Draft {{Standard}} for {{Local}} and {{Metropolitan Area
  Networks}}\textemdash{{Bridges}} and {{Bridged
  Networks}}\textemdash{{Amendment}}: {{Asynchronous Traffic Shaping}},'' {\em
  IEEE P802.1Qcr/D2.0}, vol.~In IEEE802.1 private repository. To obtain the
  access credentials, visit
  https://www.ietf.org/proceedings/52/slides/bridge-0/tsld003.htm or contact
  the IEEE802.1 chair., Dec. 2019.
\newblock
  \url{http://www.ieee802.org/1/files/private/cr-drafts/d2/802-1Qcr-d2-0.pdf}.

\bibitem{ituDefinitionsTerminologySynchronization1996}
ITU, ``Definitions and terminology for synchronization networks,'' {\em ITU
  G.810}, 1996.
\newblock \url{https://www.itu.int/rec/T-REC-G.810-199608-I/en}.

\bibitem{ituTimingRequirementsSlave2004}
ITU, ``Timing requirements of slave clocks suitable for use as node clocks in
  synchronization networks,'' {\em ITU G.812}, 2004.
\newblock \url{https://www.itu.int/rec/T-REC-G.812-200406-I/en}.

\bibitem{ieeeIEEEStandardDefinitions2009}
IEEE, ``{{IEEE Standard Definitions}} of {{Physical Quantities}} for
  {{Fundamental Frequency}} and {{Time Metrology}}\textemdash{{Random
  Instabilities}},'' {\em IEEE Std Std 1139-2008}, pp.~c1--35, Feb. 2009.

\bibitem{ieeeIEEEStandardLocal2011}
IEEE, ``{{IEEE Standard}} for {{Local}} and {{Metropolitan Area Networks}} -
  {{Timing}} and {{Synchronization}} for {{Time}}-{{Sensitive Applications}} in
  {{Bridged Local Area Networks}},'' {\em IEEE Std 802.1AS-2011}, pp.~1--292,
  Mar. 2011.

\bibitem{rfc5905}
J.~Martin, J.~Burbank, W.~Kasch, and P.~D.~L. Mills, ``{Network Time Protocol
  Version 4: Protocol and Algorithms Specification}.'' RFC 5905, June 2010.

\bibitem{wuClusterBasedConsensusTime2015}
J.~Wu, L.~Zhang, Y.~Bai, and Y.~Sun, ``Cluster-{{Based Consensus Time
  Synchronization}} for {{Wireless Sensor Networks}},'' {\em IEEE Sensors
  Journal}, vol.~15, pp.~1404--1413, Mar. 2015.
\newblock Conference Name: IEEE Sensors Journal.

\bibitem{liuJointTimeSynchronization2016}
J.~Liu, Z.~Wang, J.-H. Cui, S.~Zhou, and B.~Yang, ``A {{Joint Time
  Synchronization}} and {{Localization Design}} for {{Mobile Underwater Sensor
  Networks}},'' {\em IEEE Transactions on Mobile Computing}, vol.~15,
  pp.~530--543, Mar. 2016.
\newblock Conference Name: IEEE Transactions on Mobile Computing.

\bibitem{dierikxWhiteRabbitPrecision2016}
E.~F. Dierikx, A.~E. Wallin, T.~Fordell, J.~Myyry, P.~Koponen, M.~Merimaa,
  T.~J. Pinkert, J.~C.~J. Koelemeij, H.~Z. Peek, and R.~Smets, ``White {{Rabbit
  Precision Time Protocol}} on {{Long}}-{{Distance Fiber Links}},'' {\em IEEE
  Transactions on Ultrasonics, Ferroelectrics, and Frequency Control}, vol.~63,
  pp.~945--952, July 2016.

\bibitem{frerisFundamentalLimitsSynchronizing2011}
N.~M. Freris, S.~R. Graham, and P.~R. Kumar, ``Fundamental {{Limits}} on
  {{Synchronizing Clocks Over Networks}},'' {\em IEEE Transactions on Automatic
  Control}, vol.~56, pp.~1352--1364, June 2011.
\newblock Conference Name: IEEE Transactions on Automatic Control.

\bibitem{ridouxjulienPrinciplesRobustTiming2010}
J.~Ridoux and D.~Veitch, ``Principles of robust timing over the internet,''
  {\em Communications of the ACM}, May 2010.
\newblock \url{https://dl.acm.org/doi/abs/10.1145/1735223.1735241}.

\bibitem{gengExploitingNaturalNetwork2018a}
Y.~Geng, S.~Liu, Z.~Yin, A.~Naik, B.~Prabhakar, M.~Rosenblum, and A.~Vahdat,
  ``Exploiting a {{Natural Network Effect}} for {{Scalable}}, {{Fine}}-grained
  {{Clock Synchronization}},'' in {\em 15th
  \{\vphantom\}{{USENIX}}\vphantom\{\} {{Symposium}} on {{Networked Systems
  Design}} and {{Implementation}} (\{\vphantom\}{{NSDI}}\vphantom\{\} 18)},
  pp.~81--94, 2018.
\newblock \url{https://www.usenix.org/node/211256}.

\bibitem{veitchRobustSynchronizationSoftware2004}
D.~Veitch, S.~Babu, and A.~P{\`a}sztor, ``Robust synchronization of software
  clocks across the internet | {{Proceedings}} of the 4th {{ACM SIGCOMM}}
  conference on {{Internet}} measurement,'' in {\em {{IMC}} '04:
  {{Proceedings}} of the 4th {{ACM SIGCOMM}} Conference on {{Internet}}
  Measurement}, ({Taormina Sicily, Italy}), 2004.
\newblock \url{https://dl.acm.org/doi/abs/10.1145/1028788.1028817}.

\bibitem{zhaoComparisonTimeSensitive2018}
L.~Zhao, F.~He, E.~Li, and J.~Lu, ``Comparison of {{Time Sensitive Networking}}
  ({{TSN}}) and {{TTEthernet}},'' in {\em 2018 {{IEEE}}/{{AIAA}} 37th {{Digital
  Avionics Systems Conference}} ({{DASC}})}, pp.~1--7, Sept. 2018.

\bibitem{bouillard2018deterministic}
A.~Bouillard, M.~Boyer, and E.~Corronc, {\em Deterministic {{Network
  Calculus}}: {{From Theory}} to {{Practical Implementation}}}.
\newblock Networks and {{Telecommunications}}, {Wiley}, 2018.
\newblock \url{http://doi.org/10.1002/9781119440284}.

\bibitem{wandelerPerformanceAnalysisGreedy2006a}
E.~Wandeler, A.~Maxiaguine, and L.~Thiele, ``Performance analysis of greedy
  shapers in real-time systems,'' in {\em Proceedings of the {{Design
  Automation Test}} in {{Europe Conference}}}, vol.~1, pp.~6 pp.--, Mar. 2006.

\bibitem{IEEEISOIEC2019a}
{ISO-IEC-IEEE}, ``{{IEEE}}/{{ISO}}/{{IEC International Standard}} -
  {{Information}} technology - {{Telecommunications}} and information exchange
  between systems - {{Local}} and metropolitan area networks - {{Specific}}
  requirements - {{Part 1Q}}: {{Bridges}} and bridged networks - {{AMENDMENT}}
  7: {{Cyclic}} queuing and forwarding,'' {\em ISO/IEC/IEEE
  8802-1Q:2016/Amd.7:2019(E)}, pp.~1--34, Mar. 2019.
\newblock Conference Name: ISO/IEC/IEEE 8802-1Q:2016/Amd.7:2019(E).

\bibitem{IEEEStandardLocal2010}
IEEE, ``{{IEEE Standard}} for {{Local}} and metropolitan area
  networks\textendash{} {{Virtual Bridged Local Area Networks Amendment}} 12:
  {{Forwarding}} and {{Queuing Enhancements}} for {{Time}}-{{Sensitive
  Streams}},'' {\em IEEE Std 802.1Qav-2009 (Amendment to IEEE Std
  802.1Q-2005)}, pp.~1--72, Jan. 2010.
\newblock Conference Name: IEEE Std 802.1Qav-2009 (Amendment to IEEE Std
  802.1Q-2005).

\bibitem{iso-iec-ieeeISOIECIEEE2018}
{ISO-IEC-IEEE}, ``{{ISO}}/{{IEC}}/{{IEEE International Standard}} \textendash{}
  {{Information}} technology \textendash{} {{Telecommunications}} and
  information exchange between systems \textendash{} {{Local}} and metropolitan
  area networks \textendash{} {{Specific}} requirements \textendash{} {{Part
  1Q}}: {{Bridges}} and bridged networks {{AMENDMENT}} 3: {{Enhancements}} for
  scheduled traffic,'' {\em ISO/IEC/IEEE 8802-1Q:2016/Amd.3:2017(E)},
  pp.~1--62, Feb. 2018.
\newblock Conference Name: ISO/IEC/IEEE 8802-1Q:2016/Amd.3:2017(E).

\bibitem{navetUsingMachineLearningto2019}
N.~NAVET, T.~L. MAI, and J.~MIGGE, ``Using {{Machine Learningto SpeedUp}} the
  {{Design Space Exploration}} of {{Ethernet TSN}} networks,'' tech. rep., Jan.
  2019.
\newblock
  \url{https://orbilu.uni.lu/bitstream/10993/38604/1/feasibility-with-ml.pdf}.

\bibitem{nasrallahPerformanceComparisonIEEE2019}
A.~Nasrallah, A.~S. Thyagaturu, Z.~Alharbi, C.~Wang, X.~Shao, M.~Reisslein, and
  H.~Elbakoury, ``Performance {{Comparison}} of {{IEEE}} 802.1 {{TSN Time Aware
  Shaper}} ({{TAS}}) and {{Asynchronous Traffic Shaper}} ({{ATS}}),'' {\em IEEE
  Access}, vol.~7, pp.~44165--44181, 2019.
\newblock Conference Name: IEEE Access.

\bibitem{nsnamNs3NetworkSimulator2011}
{nsnam}, ``Ns-3 {{Network Simulator}}. {{Project}} homepage.''
  \url{https://www.nsnam.org/}, 2011.

\bibitem{maruyamaNS3BasedIEEE2015}
T.~Maruyama, T.~Yamada, S.~Yoshida, M.~Kido, and C.~Komatsu, ``{{NS}}-3 based
  {{IEEE}} 1588 synchronization simulator for multi-hop network,'' in {\em 2015
  {{IEEE International Symposium}} on {{Precision Clock Synchronization}} for
  {{Measurement}}, {{Control}}, and {{Communication}} ({{ISPCS}})},
  pp.~99--104, Oct. 2015.

\bibitem{bergerRelevanceAdversarialQueueing2014}
D.~S. Berger, M.~Karsten, and J.~Schmitt, ``On the relevance of adversarial
  queueing theory in practice,'' in {\em The 2014 {{ACM}} International
  Conference on {{Measurement}} and Modeling of Computer Systems},
  {{SIGMETRICS}} '14, ({Austin, Texas, USA}), pp.~343--354, {Association for
  Computing Machinery}, June 2014.
\newblock \url{https://doi.org/10.1145/2591971.2592006}.

\bibitem{bhattacharjeeInstabilityFIFOArbitrarily2005a}
R.~Bhattacharjee, A.~Goel, and Z.~Lotker, ``Instability of {{FIFO}} at
  {{Arbitrarily Low Rates}} in the {{Adversarial Queueing Model}},'' {\em SIAM
  Journal on Computing}, vol.~34, pp.~318--332, Jan. 2005.
\newblock \url{https://epubs.siam.org/doi/10.1137/S0097539703426805}.

\bibitem{andrewsInstabilityFIFOPermanent2009}
M.~Andrews, ``Instability of {{FIFO}} in the {{Permanent Sessions Model}} at
  {{Arbitrarily Small Network Loads}},'' {\em ACM Trans. Algorithms}, vol.~5,
  pp.~33:1--33:29, July 2009.
\newblock \url{http://doi.acm.org/10.1145/1541885.1541894}.

\bibitem{phanComposingFunctionalStateBased2007}
L.~T.~X. Phan, S.~Chakraborty, P.~S. Thiagarajan, and L.~Thiele, ``Composing
  {{Functional}} and {{State}}-{{Based Performance Models}} for {{Analyzing
  Heterogeneous Real}}-{{Time Systems}},'' in {\em 28th {{IEEE International
  Real}}-{{Time Systems Symposium}} ({{RTSS}} 2007)}, pp.~343--352, Dec. 2007.

\bibitem{daigmorteTraversalTimeWeakly2016}
H.~Daigmorte and M.~Boyer, ``Traversal time for weakly synchronized {{CAN}}
  bus,'' in {\em Proceedings of the 24th {{International Conference}} on
  {{Real}}-{{Time Networks}} and {{Systems}}}, {{RTNS}} '16, ({Brest, France}),
  pp.~35--44, {Association for Computing Machinery}, Oct. 2016.
\newblock \url{https://doi.org/10.1145/2997465.2997477}.

\bibitem{daigmorteEvaluationAdmissibleCAN2017}
H.~Daigmorte and M.~Boyer, ``Evaluation of admissible {{CAN}} bus load with
  weak synchronization mechanism,'' in {\em Proceedings of the 25th
  {{International Conference}} on {{Real}}-{{Time Networks}} and {{Systems}}},
  {{RTNS}} '17, ({Grenoble, France}), pp.~277--286, {Association for Computing
  Machinery}, Oct. 2017.
\newblock \url{https://doi.org/10.1145/3139258.3139261}.

\bibitem{daigmorteReducingCANLatencies2017}
H.~Daigmorte, M.~Boyer, and J.~Migge, ``Reducing {{CAN}} latencies by use of
  weak synchronization between stations,'' 2017.

\bibitem{cruzCalculusNetworkDelay1991}
R.~L. Cruz, ``A calculus for network delay. {{II}}. {{Network}} analysis,''
  {\em IEEE Transactions on Information Theory}, vol.~37, pp.~132--141, Jan.
  1991.
\newblock \url{http://doi.org/10.1109/18.61110}.

\bibitem{cruzCalculusNetworkDelay1991a}
R.~L. Cruz, ``A calculus for network delay. {{I}}. {{Network}} elements in
  isolation,'' {\em IEEE Transactions on Information Theory}, vol.~37,
  pp.~114--131, Jan. 1991.
\newblock \url{http://doi.org/10.1109/18.61109}.

\bibitem{changPerformanceGuaranteesCommunication2000}
C.-S. Chang, {\em Performance {{Guarantees}} in {{Communication Networks}}}.
\newblock Telecommunication {{Networks}} and {{Computer Systems}}, {London}:
  {Springer-Verlag}, 2000.
\newblock \url{https://www.springer.com/gp/book/9781852332266}.

\bibitem{norros1994storage}
I.~Norros, ``A storage model with self-similar input,'' {\em Queueing systems},
  vol.~16, no.~3-4, pp.~387--396, 1994.

\bibitem{boyer2012deficit}
M.~Boyer, G.~Stea, and W.~M. Sofack, ``Deficit round robin with network
  calculus,'' in {\em 6th International ICST Conference on Performance
  Evaluation Methodologies and Tools}, pp.~138--147, IEEE, 2012.

\bibitem{daigmorte2018modelling}
H.~Daigmorte, M.~Boyer, and L.~Zhao, ``Modelling in network calculus a tsn
  architecture mixing time-triggered, credit based shaper and best-effort
  queues,'' 2018.

\bibitem{mohammadpour2019improved}
E.~Mohammadpour, E.~Stai, and J.-Y. Le~Boudec, ``Improved delay bound for a
  service curve element with known transmission rate,'' {\em IEEE Networking
  Letters}, vol.~1, no.~4, pp.~156--159, 2019.

\bibitem{charnyDelayBoundsNetwork2000}
A.~Charny and J.-Y. Le~Boudec, ``Delay {{Bounds}} in a {{Network}} with
  {{Aggregate Scheduling}},'' in {\em Quality of {{Future Internet Services}}}
  (J.~Crowcroft, J.~Roberts, and M.~I. Smirnov, eds.), Lecture {{Notes}} in
  {{Computer Science}}, pp.~1--13, {Springer Berlin Heidelberg}, 2000.
\newblock \url{https://link.springer.com/chapter/10.1007/3-540-39939-9_1}.

\bibitem{le2002some}
J.-Y. Le~Boudec, ``Some properties of variable length packet shapers,'' {\em
  IEEE/ACM Transactions on Networking}, vol.~10, no.~3, pp.~329--337, 2002.

\bibitem{chang1998general}
C.-S. Chang and Y.~H. Lin, ``A general framework for deterministic service
  guarantees in telecommunication networks with variable length packets,'' in
  {\em 1998 Sixth International Workshop on Quality of Service (IWQoS'98)(Cat.
  No. 98EX136)}, pp.~49--58, IEEE, 1998.

\bibitem{thomasCyclicDependenciesRegulators2019}
L.~Thomas, J.-Y. Le~Boudec, and A.~Mifdaoui, ``On {{Cyclic Dependencies}} and
  {{Regulators}} in {{Time}}-{{Sensitive Networks}},'' in {\em 2019 {{IEEE
  Real}}-{{Time Systems Symposium}} ({{RTSS}})}, Dec. 2019.
\newblock \url{https://infoscience.epfl.ch/record/272599}.

\bibitem{ecssECSSQST3002CFailureModes2009}
ECSS, ``{{ECSS}}-{{Q}}-{{ST}}-30-{{02C}} \textendash{} {{Failure}} modes,
  effects (and criticality) analysis ({{FMEA}}/{{FMECA}}) \textendash{} (6
  {{March2009}}) | {{European Cooperation}} for {{Space Standardization}},''
  tech. rep., 2009.
\newblock
  \url{https://ecss.nl/standard/ecss-q-st-30-02c-failure-modes-effects-and-criticality-analysis-fmeafmeca/}.

\bibitem{ecsssECSSQST4002CHazardAnalysis2008}
ECSSS, ``{{ECSS}}-{{Q}}-{{ST}}-40-{{02C}} \textendash{} {{Hazard}} analysis (15
  {{November}} 2008) | {{European Cooperation}} for {{Space
  Standardization}},'' tech. rep., 2008.
\newblock \url{https://ecss.nl/standard/ecss-q-st-40-02c-hazard-analysis/}.

\bibitem{recommendation20068261}
I.~Recommendation, ``8261/y. 1361 timing and synchronization aspects in packet
  networks,'' {\em International Telecommun. Union}, 2006.

\bibitem{millsNetworkTimeProtocol2010a}
D.~Mills, J.~Martin, J.~Burbank, and W.~Kasch, ``Network {{Time Protocol
  Version}} 4: {{Protocol}} and {{Algorithms Specification}},'' Tech. Rep.
  RFC5905, {RFC Editor}, June 2010.
\newblock \url{https://www.rfc-editor.org/info/rfc5905}.

\bibitem{googleLeapSmearPublic}
Google, ``Leap {{Smear}} | {{Public NTP}}.''
  \url{https://developers.google.com/time/smear}.
\newblock Library Catalog: developers.google.com.

\bibitem{ayedHierarchicalTrafficShaping2014a}
H.~Ayed, A.~Mifdaoui, and C.~Fraboul, ``Hierarchical traffic shaping and frame
  packing to reduce bandwidth utilization in the {{AFDX}},'' in {\em
  Proceedings of the 9th {{IEEE International Symposium}} on {{Industrial
  Embedded Systems}} ({{SIES}} 2014)}, pp.~77--86, June 2014.

\end{thebibliography}

\appendix
% !TeX spellcheck = en_US
% !TeX rootfile = article.tex
\section{Appendix}
\label{sec:appendix}

\subsection{Proof of Proposition~\ref{prop:toolbox:delay}}
	\begin{proof}[Proof of Proposition~\ref{prop:toolbox:delay}]\label{proof:toolbox:delay}
		Call $A(t)$ the time, measured with $\h_i$, at which a packet that has entered device $j$ at $t$ leaves the device. By definition, for any time $t$ observed with $\h_i$, $A(t)-t\le D$. As $d$ is increasing, for any time $t$ measured in $\h_i$, $d(A(t)) - d(t) \le d(t+D) - d(t) \le \sup_{t'} (d(t'+D) -d(t')) = (d\oslash d) (D)$.
		
		Note that $d(A(t))$ also equals $A(d(t))$ because the packet that enters device $j$ at $t$ observed with $\h_i$ [resp $d(t)$ observed with $\h_g$] leaves device $j$ at $A(t)$ [resp $A(d(t))$]. Hence for any $t$ observed using $\h_i$, $A(d(t)) - d(t) \le (d\oslash d)(D)$, which gives the result.
	\end{proof}

\subsection{Proof of Proposition~\ref{prop:cumu}}
	\begin{proof}[Proof of Proposition~\ref{prop:cumu}]\label{proof:toolbox:cumu}
		Note that the same number of bits of the flow enter the device $j$ between the time instants measured as $t_1$ and $t_2$ using $\mathcal{H}_{i}$ and between the time instants measured as $d_{g\rightarrow i}(t_1)$ and $d_{g\rightarrow i}(t_2)$ using $\mathcal{H}_g$. We now split the proof into two situations.
		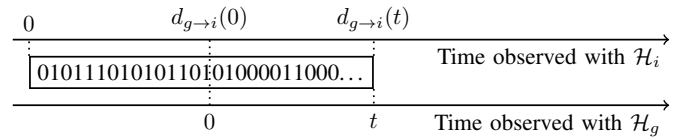
\begin{figure}[h]\centering%
			\resizebox{\linewidth}{!}{% !TeX spellcheck = en_US
% !TeX rootfile = article.tex
\begin{tikzpicture}
	\draw[->] (0,1) -- (10,1) node[pos=1,anchor=north east] {Time observed with $\mathcal{H}_i$};
	\draw[->] (0,0) -- (10,0) node[pos=1,anchor=north east] {Time observed with $\mathcal{H}_{g}$};
	\node[draw, anchor=east] at (5.5,0.5) (p) {01011101010110101000011000\ldots};
	
	\draw[-, dotted] (p.north west) -- ++(0,0.25cm) node [above] {$0$};;
	\draw[-, dotted] (p.north east) -- ++(0,0.25cm) node [above] {$d_{g\rightarrow i}(t)$};
	\draw[-, dotted] (p.south east) -- ++(0,-0.25cm) node [below] {$t$};
	
	\draw[-,dotted] (3,0) -- (3,1) node[pos=0,below] {$0$} node[pos=1,above] {$d_{g\rightarrow i}(0)$};

\end{tikzpicture}}%
			\caption{\label{fig:cumufCasA} Case where the initial offset of the local clock $\h_i$ is positive compared to the reference clock $\h_g$.}%
		\end{figure}
		
		\paragraph{Case $d(0) \ge 0$} If the origin of $\h_i$ is before the origin of $\h_g$, then the situation is as in Figure \ref{fig:cumufCasA}, where we represent face-to-face the two clocks. For any $t$, the number of bits observed using $\h_i$ between $0$ and $d(t)$ equals the number of bits observed between $0$ and $d(0)$, plus the number of bits observed between $d(0)$ and $d(t)$. That second term also equals the number of bits observed between $0$ and $t$ using $\h_g$ (Figure \ref{fig:cumufCasA}), i.e., $R^{\mathcal{H}_{g}}(t)$. Hence,
		$$\forall t, R^{\mathcal{H}_i}(d(t)) = R^{\mathcal{H}_i}(d(0)) + R^{\mathcal{H}_{g}}(t) $$
		Due to the definition of $T_{\text{start}}$, no bit could have been sent before the time measured as $0$ using $\mathcal{H}_{g}$. Consequently, $ R^{\mathcal{H}_i}(d(0))=0$, and we have the result.
		
		\paragraph{Case  $d(0) < 0$} If the origin of clock $\h_i$ is after the origin of clock $\h_g$, then by symmetry, we simply flip the pair and obtain the result from case \emph{a)}.
	\end{proof}

\subsection{Proof of Proposition~\ref{prop:ac}}
	\begin{proof}[Proof of Proposition~\ref{prop:ac}]\label{proof:toolbox:ac}
		For any $t,\tau$, we have
		\begin{equation*}\label{eq:proof1}
			\begin{aligned}
				&R^{\mathcal{H}_{g}}(t+\tau) - R^{\mathcal{H}_{g}}(t) \\
				&= R^{\mathcal{H}_i}(d(t+\tau)) - R^{\mathcal{H}_i}(d(t)) \\
				& \qquad \text{per Proposition~\ref{prop:cumu}} \\
				&\le \alpha^{\mathcal{H}_i}(d(t+\tau)-d(t))\\
				& \qquad \text{because } \alpha^{\mathcal{H}_i} \text{ is an arrival curve when observed with } \mathcal{H}_i \\
				&\le \sup_{t'\ge 0} \left[ \alpha^{\mathcal{H}_i}(d(t'+\tau) - d(t')) \right] \\
				&\le  \alpha^{\mathcal{H}_i}(d \oslash d (\tau))\\
				& \qquad \text{as } \alpha^{\mathcal{H}_i} \text{ is a wide-sense increasing function}\\
			\end{aligned}
		\end{equation*}
		Also, per Equation~(\ref{eq:constr-async}),  $(d\oslash d)(\tau) \le \rho\tau + \eta$, and $d\oslash d$ is always finite in our model.
	\end{proof}

\subsection{Proof of Proposition~\ref{prop:sc}}
	\begin{proof}[Proof of Proposition~\ref{prop:sc}]\label{proof:toolbox:sc}
		Fix any arbitrary $t\ge0$.
		Then
		\begin{equation*}
			\begin{aligned}
				R^{*\mathcal{H}_{g}}(t) &= R^{*\mathcal{H}_i}(d(t)) \\
				&\hspace{-0.8cm}\ge \inf_{0\le \sigma \le d(t)} \left[R^{\mathcal{H}_i}(\sigma) + \beta^{\mathcal{H}_i}(d(t)-\sigma) \right]
			\end{aligned}
		\end{equation*}
		As per Proposition~\ref{prop:cumu} and because $\beta^{\mathcal{H}_i}$ is a service curve observed with $\mathcal{H}_i$. As $d$ is strictly increasing, for any time instant $\sigma$ measured with $\mathcal{H}_i$, we note $\sigma = d(s)$ with $s$ the measure using $\mathcal{H}_{g}$. Hence,
		\begin{equation*}
			\begin{aligned}
				&R^{*\mathcal{H}_{g}}(t)  \\
				&\ge \inf_{d^{-1}(0)\le s \le t} \left[ R^{\mathcal{H}_i}(d(s)) + \beta^{\mathcal{H}_i}(d(t)-d(s)) \right]
			\end{aligned}
		\end{equation*}
		
		If $t< T_{\text{start}}$ (defined in Section~\ref{sec:time-model}), the result holds because no bit has been transmitted: $R^{*\mathcal{H}_{g}}(t)=0$ and $\forall s\le t, R^{\mathcal{H}_i}(d(s))=0$ and $\inf_{d^{-1}(0)\le s \le t} \beta^{\mathcal{H}_i}(d(t)-d(s))=0$, obtained for $s=t$.
		
		Let's now assume $t\ge T_{\text{start}}$ and call $A:s\mapsto R^{\mathcal{H}_i}(d(s)) + \beta^{\mathcal{H}_i}(d(t)-d(s))$. As $T_{\text{start}}\ge d^{-1}(0)$, we can split the domain of the $\inf_s A(s)$ into the two cases $s\le T_{\text{start}}$ and $s\ge T_{\text{start}}$, the result being the minimum of the two obtained $\inf$.
		
		For $s < T_{\text{start}}$, $\beta^{\mathcal{H}_i}(d(t)-d(s)) \ge \beta^{\mathcal{H}_i}(d(t)-d(T_{\text{start}}))$ because $d$ and $\beta^{\mathcal{H}_i}$ are both increasing functions. On the other hand, based on the assumption on $T_{\text{start}}$, $R^{\mathcal{H}_i}(d(s))=R^{\mathcal{H}_i}(d(T_{\text{start}}))$ as both quantities equal 0 bit.
		
		Consequently,	
		$$\inf_{d^{-1}(0)\le s \le T_{\text{start}}} A(s) \ge A(T_{\text{start}})$$
		and $\inf_s A(s)$ is obtained for $s\in[T_{\text{start}},t]$, i.e.,
		$$ R^{*\mathcal{H}_{g}}(t) \ge \inf_{T_{\text{start}}\le s \le t} A(s)$$
		By definition $T_{\text{start}}\ge 0$, so $[T_{\text{start}}, t] \subset [0, t]$, and hence,
		\begin{align*}
			R^{*\mathcal{H}_{g}}(t) &\ge \inf_{0\le s \le t} A(s)\\
			&\ge\inf_{0\le s \le t} \left[ R^{\mathcal{H}_i}(d(s)) + \beta^{\mathcal{H}_i}(d(t)-d(s)) \right]\\
			&\hspace{-1.5cm}\ge\inf_{0\le s \le t} \left[ \makecell[l]{R^{\mathcal{H}_i}(d(s)) + \\\inf_{0\le u \le s} \beta^{\mathcal{H}_i}(d(t-s+u)-d(u))} \right]\\
			&\hspace{-1.5cm}\ge\inf_{0\le s \le t} \left[ R^{\mathcal{H}_i}(d(s)) + \beta^{\mathcal{H}_i}((d \overline{\oslash} d)(t-s)) \right]
		\end{align*}
		Because $\beta^{\mathcal{H}_i}$ is wide-sense increasing. Now, as $R^{\mathcal{H}_i}(d(s)) = R^{\mathcal{H}_{g}}(s)$ per Proposition~\ref{prop:cumu}, we obtain
		$$R^{*\mathcal{H}_{g}}(t) \ge \inf_{0\le s \le t} \left[ R^{\mathcal{H}_{g}}(s) + \beta^{\mathcal{H}_i}((d \overline{\oslash} d)(t-s)) \right]$$
		which proves that $t\mapsto\beta^{\mathcal{H}_i}((d \overline{\oslash} d)(t-s))$ is a service curve observed with $\mathcal{H}_{g}$.
	\end{proof}

\subsection{Proof of Proposition~\ref{prop:insta-async}}
	\begin{proof}[Proof of Proposition~\ref{prop:insta-async}]\label{proof:insta-async}
		Consider a flow such as $f$ in Figure~\ref{fig:prob}, assume that the source of the flow is at the input to system $S$ and call $\mathcal{H}_{1}$ the clock used to define the arrival curve $\alpha^{\mathcal{H}_{1}}=\gamma_{r,b}$ at the source. Assume the regulator is a \ac{PFR} and let $\mathcal{H}_{\text{Reg}}\neq \mathcal{H}_{1} $ be its clock. Since the \ac{PFR} is non-adapted, it is configured with shaping curve $\sigma^{\mathcal{H}_{\text{Reg}}} = \alpha^{\mathcal{H}_{1}}$.
		
		Let the flow be greedy after $T_{\text{start}}$ and generate packets of size $\ell$ bits, with $\ell\leq b$.
		Its cumulative arrival function at the source, measured in $\mathcal{H}_{1}$, is
		\begin{equation}\label{eq-proof-async-0}
		 R^{\mathcal{H}_{1}}(t) = \left\lfloor \frac{\alpha^{\mathcal{H}_{1}}(|t-T_{\text{start}}|^+)}{\ell}\right\rfloor\ell
		\end{equation}
		where $\lfloor \cdot\rfloor$ is the floor function. The fact that this flow satisfies the arrival-curve constraint $\alpha^{\mathcal{H}_{1}}$ in $\mathcal{H}_{1}$ follows from \cite[Thm III.2]{le2002some}. Also, system $S$ provides a delay bound $D$ in \acs{TAI} thus, by Section \ref{sec:toolbox:delay}, a bound $D_1=\rho D+ \eta$ in $\mathcal{H}_{1}$. Let $R'^{\mathcal{H}_{1}}$ be the cumulative arrival function of flow $f$ at the input of the \ac{PFR}, observed with $\mathcal{H}_{1}$. Thus, for every $t$,
		\begin{equation}\label{eq-proof-async-1}
		 R^{\mathcal{H}_{1}}(t-D_1) \leq R'^{\mathcal{H}_{1}}(t)%\leq R^{\mathcal{H}_{1}}(t)
		\end{equation}
		Let $d_{1\rightarrow\text{Reg}}(t) = t/\rho$ be the relative time function between $\h_{1}$ and $\h_{\text{Reg}}$. The function meets the conditions of Equation~(\ref{eq:constr-async}).
		%		
		%
		%Then it exists $\mathcal{H}_{1}$ such that 1) $\mathcal{H}_{1} \ne \mathcal{H}_{\text{Reg}}$, 2) flow $f$ has, at the input of the FIFO system, an arrival curve $\alpha^{\mathcal{H}_{1}}$ when observed with $\h_{\text{valid}}$ and 3) the configured shaping curve of the regulator is $\sigma^{\mathcal{H}_{\text{Reg}}} = \alpha^{\mathcal{H}_{1}}$. Then the following example leads to unbounded delays.
		%		Take $\alpha^{\mathcal{H}_{1}}=\gamma_{r,b}$ a leaky-bucket arrival curve and consider that the source is greedy after $T_{\text{start}}$: $R^{\mathcal{H}_{1}}(t) = \alpha^{\mathcal{H}_{1}}(|t-T_{\text{start}}|^+)$. Also, assume that system S provides infinite service and zero delay. Consider $d_{\text{valid}\rightarrow\text{Reg}}(t) = t/\rho$ as an adversarial relative time function between $\h_{\text{valid}}$ and $\h_{\text{Reg}}$. The function meets the conditions of Equation~\ref{eq:constr-async}.
		%		
		Using Proposition~\ref{prop:cumu}, the cumulative arrival function of flow $f$ at the input of the \ac{PFR}, when observed with the \ac{PFR} clock $\h_{\text{Reg}}$, is given by
		\begin{equation}\label{eq-proof-async-1a}
		R'^{\mathcal{H}_{\text{Reg}}}(\tau) = R'^{\mathcal{H}_{1}}\left(d_{1\to\text{Reg}}^{-1}(\tau)\right)=R'^{\mathcal{H}_{1}}(\rho\tau)
		\end{equation}
		
		Network calculus results are valid as long as all the notions are in the same time reference. Recall from Section~\ref{sec:prob} that the PFR can be modelled as a fluid greedy shaper followed by a packetizer, and the latter can be ignored for delay computations. The output of the fluid greedy shaper, observed with $\h_{\text{Reg}}$, is thus
		$
		  R^{*\mathcal{H}_{\text{Reg}}}=R'^{\mathcal{H}_{\text{Reg}}}\otimes \gamma_{r,b}
		$. It follows that, for all $\tau$, $R^{*\mathcal{H}_{\text{Reg}}}(\tau) \le R^{\h_{\text{Reg}}}(\T) + \alpha^{\mathcal{H}^{1}}(\tau-\T) = r|\tau-\T|^++b$ by definition of $\T$.

		We now show that for any $e>0$, the worst-case delay through the $\ac{PFR}$, measured in $\h_{\text{Reg}}$, is $\geq e$. From the previous equation, it follows that, for $\tau\geq\T$,
		\begin{equation}\label{eq-proof-async-1b}
		  R^{*\mathcal{H}_{\text{Reg}}}(\tau+e)\leq r(\tau+e-\T)+b
		\end{equation}
		Combine Eqs~(\ref{eq-proof-async-0})--(\ref{eq-proof-async-1b}) and obtain
		\begin{equation}
		\begin{aligned}
		 R^{*\mathcal{H}_{\text{Reg}}}(\tau+e)-R'^{\mathcal{H}_{\text{Reg}}}(\tau)
		  \\
		  \leq r(\tau+e-\T)+b - R^{\h_{1}}(\rho\tau -D_1)
		  \\
		  =
		  r(\tau+e-\T)+b
		  -
		  \left\lfloor \frac{r(\rho \tau-D_1\T)+b}{\ell}\right\rfloor\ell\\
		   \leq  r(\tau+e-\T)+b - (r(\rho \tau-D_1-\T)+b)+\ell\\
		   = r(1-\rho)\tau+re+rD_1+\ell
		 \end{aligned}
		\end{equation}
		Thus $R^{*\mathcal{H}_{\text{Reg}}}(\tau+e)-R'^{\mathcal{H}_{\text{Reg}}}(\tau)<0$ whenever $\tau > \left(\frac{re+rD_1+\ell}{(\rho-1)r}\vee\T\right)$; it follows that the delay, measured with $\h_{\text{Reg}}$, for packets arrived at the \ac{PFR} after time $\left(\frac{re+rD_1+\ell}{(\rho-1)r}\vee\T\right)$ is larger than $e$. This holds for any arbitrary $e>0$, therefore the delay measured with $\h_{\text{Reg}}$ is unbounded.
		 %
		%
		%
		%
		%
		% Hence, observing with $\mathcal{H}_{\text{Reg}}$, we have that  $\alpha^{\mathcal{H}_{1}}$ is a maximum service curve for the regulator \cite[Thm 1.5.1]{leboudecNetworkCalculusTheory2001} and that $R^{*\mathcal{H}_{\text{Reg}}}(\tau) \le R^{\h_{\text{Reg}}}(\T) + \alpha^{\mathcal{H}^{\text{valid}}}(\tau-\T) = r|\tau-\T|^++b$ by definition of $\T$. As a consequence the instantaneous backlog of the \ac{PFR} is $R^{\mathcal{H}_{\text{Reg}}}(\tau) - R^{*\mathcal{H}_{\text{Reg}}}(\tau) \ge r|\tau-\T|^+(\rho-1)$ and is not bounded,
		
		 By Section~\ref{sec:toolbox:delay}, this also proves that the delay is not bounded when viewed from any clock of the network.
				
				For the \ac{IR}, the same adversarial example applies because an \ac{IR} processing only one flow has the same behavior as a \ac{PFR}~\cite{leboudecTheoryTrafficRegulators2018}.			
	\end{proof}

\subsection{Proof of Proposition~\ref{lemma:cascade:whole-delay}}
	\begin{proof}[Proof of Proposition~\ref{lemma:cascade:whole-delay}]\label{proof:cascade:whole-delay}
		Assume that $S$ is the $k$th hop for flow $f$ and note $S=S_k$ as in Figure~\ref{fig:async-hop-model}. When observed with $\h_{\text{Reg}_{k-1}}$, the flow has the arrival curve $\sigma_{k-1}^{\h_{\text{Reg}_{k-1}}}$ at the output of $\text{Reg}_{k-1}$. We now apply Table~\ref{tab:ac-results} with $\h_g=\h_{\text{Reg}_k}$ and $\h_i=\h_{\text{Reg}_{k-1}}$.
		We obtain that, when observed with the clock of the next regulator, $\h_{\text{Reg}_{k}}$, the flow leaves $\text{Reg}_{k-1}$ with a leaky-bucket arrival curve $\alpha_{k-1}^{\h_{\text{Reg}_{k}}}$ of rate $\rho r_{\text{Reg}_{k-1}}$ and burst ${\text{Reg}_{k-1}} + \eta r_{\text{Reg}_{k-1}}$.
		
		From the configuration of $\text{Reg}_{k-1}$ and $\text{Reg}_{k}$, we note that $\alpha_{k-1}^{\h_{\text{Reg}_{k}}} \le \sigma_k^{\h_{\text{Reg}_k}}$. Consequently, all the conditions for the shaping-for-free property are met when observed with clock $\mathcal{H}_{\text{Reg}_k}$. We apply the respective theorems for both the \ac{PFR}~\cite[Thm 1.5.2]{leboudecNetworkCalculusTheory2001} and the \ac{IR}~\cite[Thm 5]{leboudecTheoryTrafficRegulators2018} using this clock. An upper bound on the delay for the flow through the system S$_k$ as measured with $\mathcal{H}_{\text{Reg}_k}$ is $\rho D_k^{\hze}+\eta$ (Proposition~\ref{prop:toolbox:delay}). Applying the shaping-for-free property, this is also a valid delay bound for the flow trough the whole hop ($S_k$ followed by regulator), when measuring with $\mathcal{H}_{\text{Reg}_k}$. To obtain a delay bound back in the measurement clock $\hze$, we apply again Proposition~\ref{prop:toolbox:delay}, which gives the result.
	\end{proof}

\subsection{Proof of Proposition~\ref{prop:adam}}
	\label{proof:adam:whole-delay}
We first establish the following lemma.
\begin{lemma}\label{lemma:adam:whole-delay}
			Assume that $\alpha_{2,k-1}$ is an arrival curve for the flow at the input of the $S_k$, observed in $\hze$. Then
\begin{enumerate}
  \item $D_k'^{\hze} = D_k^{\hze} + \eta (1+\rho) + \frac{b_{2,k-1}-b_{0}-\eta Wr_0}{\rho r_0} \frac{\rho^2-1}{W - 1}$ is a \ac{TAI} delay bound for the flow through the concatenation of $S_k$ and $\text{Reg}_k$.
  \item $\gamma_{\rho r_0, b_{2,k-1} + \rho r_{0} \cdot D_k'^{\hze}}$ is an arrival curve for the flow observed in $\hze$ at the output of $\text{Reg}_k$.
\end{enumerate}
		\end{lemma} %
\begin{proof}

		\begin{figure}
			\resizebox{\linewidth}{!}{% !TeX spellcheck = en_US
% !TeX rootfile = article.tex
\begin{tikzpicture}[samples=200]

	\pgfplotsset{ticks=none}

	\begin{axis}[xlabel=time interval measured in $\hze$, ylabel=bits,xmin=-4,xmax=11,ymin=-1,ymax=11, axis lines=center,width=11cm,height=11cm]
	
	\draw (axis cs:3,0) -- (axis cs:3,3) node[pos=0,below] {$D_k^{\hze} + \eta$};
	\draw[dotted] (axis cs:0,3) -- (axis cs:3,3) node[pos=0,anchor=north west] {$b_{0}$};
	\addplot[domain=3:9] {x} node [pos=0.5,sloped, yshift=-0.25cm] {rate $Wr_{0}/\rho$} node[pos=1,anchor=north west] {$\beta_{\text{Hop}_k}^{\hze}$};
	
	\addplot[domain=0:1] {1.5*x+4} node[pos=0,anchor=center] (targb1) {};
	\node[draw, dashed,line width=0.5pt] at (axis cs:-1.1,2) (b1) {\makecell[c]{$b_{0}$\\$+\eta Wr_{0}$}};
	\draw[dashed, ->, line width=0.5pt] (b1.north) |- (targb1.center);
	\addplot[domain=1:4, dashed] {1.5*x+4} node [pos=0.4,sloped, yshift=0.3cm] {rate $r_{1} = \rho Wr_{0}$} node[pos=1,anchor=south west] {$\alpha_{1}^{\hze}$};

	\addplot[domain=0:1, dashed] {0.5*x+5}node[pos=0,left] {$b_{2,k-1}$};
	\addplot[domain=1:9] {0.5*x+5} node [pos=0.5,sloped,yshift=0.3cm] {rate $r_{2} = \rho r_{0}$} node[pos=1,anchor=south west]{$\alpha_{2,k-1}^{\hze}$};
	
	\draw[<->, line width=1pt] (axis cs:1,5.5) -- (axis cs:5.5,5.5) node[pos=0.5, draw, below] {$D_k'^{\hze}$};

{}	\end{axis}
\end{tikzpicture}}
			\caption{\label{fig:adam-delay} TAI delay bound computation for the flow through the whole hop ($S_k$ followed by regulator) in the \ac{ADAM} method. The knowledge of $\alpha_{2,k}$ is required to provide a bounded delay whereas the knowledge of $\alpha_{1}$ helps having a tight delay bound.}
		\end{figure}
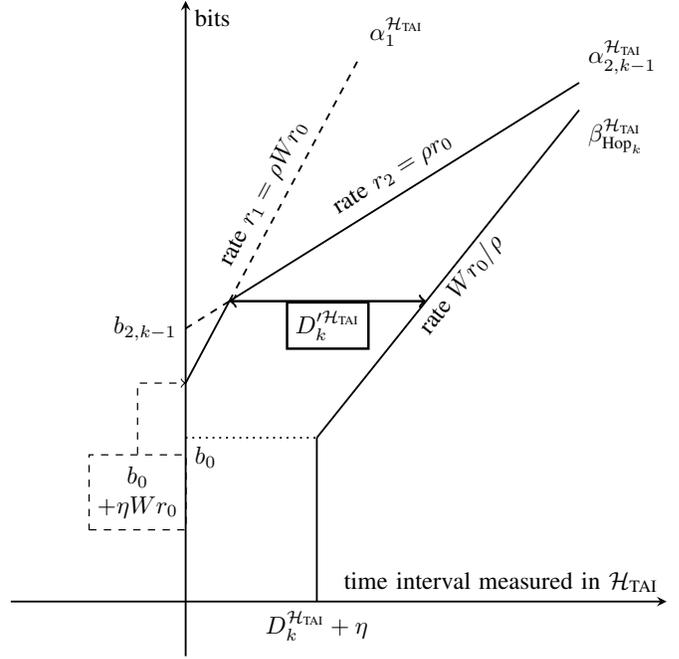

		(1) We use the property that the shaping curve of a \ac{PFR} is also a service curve, when observed with its own clock. Using the results of Table~\ref{tab:sc-results} with $\h_g=\hze$ and $\h_i = \h_{\text{Reg}_{k}}$, we obtain a service curve of the \ac{PFR} when observed with $\hze$:  $\beta_{\text{Reg}_k}^{\hze} = \delta_{\eta}\otimes\gamma_{Wr_0/\rho,b_0}$. Additionally, as $D_k^{\hze}$ is a delay bound through S$_k$, observed with $\hze$, $\delta_{D_k^{\hze}}$ is a service curve of S$_k$ when observed with $\hze$. Due to the concatenation property of service curves \cite[Thm 1.4.6]{leboudecNetworkCalculusTheory2001}, applied with $\hze$, the whole $k$-th hop offers, when observed with $\hze$, the service curve $\beta_{\text{Hop}_k}^{\hze} = \delta_{D_k^{\hze}}\otimes\delta_{\eta}\otimes\gamma_{Wr_0/\rho,b_0}$to the flow. Its shape is given in Figure~\ref{fig:adam-delay}.
		
		Conversely, the flow has, when observed with $\hze$ and at the input of the hop, both $\alpha_{1}$ (\ac{PFR} output arrival-curve property) and $\alpha_{2,k-1}$ (induction assumption) as arrival curves. Their shape are given in Figure~\ref{fig:adam-delay}.
		
		We apply the Network Calculus three-bound theorem \cite[Thm 1.4.2]{leboudecNetworkCalculusTheory2001} to obtain a delay bound as the maximal horizontal distance between  $\alpha_{1} \otimes \alpha_{2}$ and $\beta_{\text{Hop}_k}$, reached at the location marked on Figure~\ref{fig:adam-delay}. Note that knowing only $\alpha_{1}$ does not prove the existence of a maximal horizontal distance because $Wr_0/\rho < \rho Wr_0$ (for $\rho>1$). However, knowing also $\alpha_{2,k-1}$ proves it because $Wr_0/\rho \ge \rho r_0$ because $W\ge \rho^2$. Geometrical considerations gives the result in the proposition.
		
		(2) Use the output flow part of the three-bound theorem \cite[Thm 1.4.3]{leboudecNetworkCalculusTheory2001}, applied with $\hze$.
	\end{proof}

\begin{proof}[Proof of Proposition~\ref{prop:adam}]
(1)
  Applying Table~\ref{tab:ac-results} with $\h_g=\hze$ and $\h_i=\h_{\rr_0}$ proves that $\alpha_{2,0}$ is an arrival curve for the flow at its source when observed with $\hze$.

  (2) and (3): The combination of the two statements is shown by induction on $k$. The base step $k=0$ follows from the previous item. The induction step follows immediately from items~(2) and (3) of Lemma~\ref{lemma:adam:whole-delay}

\end{proof}

\subsection{Proof of Proposition~\ref{prop:sync:pfr-lower}}

	\begin{proof}[Proof of Proposition~\ref{prop:sync:pfr-lower}]\label{proof:sync:pfr-lower}
			
			Take the clock of the source to be exactly the \ac{TAI} and let $\mathcal{H}_{\text{Reg}}\neq \hze $ be the clock of the \ac{PFR}. Since the \ac{PFR} is non-adapted, it is configured with shaping curve $\sigma^{\mathcal{H}_{\text{Reg}}} = \alpha^{\hze}$.
			
			Let $x_1 = \T + \frac{\rho\Delta}{\rho - 1}$ and take, for the adversarial relative time function $d_{\text{\ac{TAI}}\rightarrow\text{PFR}}$, the following piecewise linear function
			\begin{equation*}
					d_{\text{TAI}\rightarrow\text{PFR}}(t) = \left\lbrace
					\begin{aligned}
					t &\qquad\text{ if } t\le \T\\
					\frac{1}{\rho} (t-\T) + \T &\qquad\text{ if } \T<t\le x_1\\
					t - \Delta &\qquad\text{ if } x_1<t\\
					\end{aligned}
					\right.	
			\end{equation*}
			The shape of $d_{\text{TAI}\rightarrow\text{PFR}}$ is given in Figure~\ref{fig:shape-dpfr}. It is continuous, strictly increasing, and meets the constraints of Equations~(\ref{eq:constr-async}) and~(\ref{eq:const-sync}).
			\begin{figure}\centering%
				\resizebox{\linewidth}{!}{% !TeX spellcheck = en_US
% !TeX root=ms.tex
\begin{tikzpicture}
\begin{axis}[xlabel=time observed with $\hze$, ylabel=time observed with $\h_{\text{Reg}}$, xmin=-4, ymin=-4, xmax=20, ymax=20, no marks, axis lines=center, ticks=none, width=10cm, height=10cm]
	\addplot[domain=0:18, dashed] {x};
	\addplot[domain=18:19, dotted] {x};
	\addplot[domain=0:18, dashed] {x-2};
	\addplot[domain=18:19, red, dotted] {x-2};
	\addplot[domain=0:20, dashed] {x+2};
	\draw[<->] (axis cs:17,17) -- (axis cs:17,15) node[pos=1, left] {$\Delta$};
	
	\draw[red] (axis cs:0,0) -- (axis cs:5,5);
	\draw[dotted] (axis cs:5,0) -- (axis cs:5,5) node[below,pos=0] {$\T$};
	\draw[dotted] (axis cs:0,5) -- (axis cs:5,5) node[left, pos=0] {$\T$};
	\draw[red] (axis cs:5,5) -- (axis cs:12,10);
	\draw[red] (axis cs:12,10) -- (axis cs:18,16);
	\draw[dotted] (axis cs:12,0) -- (axis cs:12,10) node[below, pos=0] {$x_1$};
	\draw[dotted] (axis cs:0,10) -- (axis cs:12,10) node[left, pos=0] {$x_1-\Delta$};

	\end{axis}
\end{tikzpicture}} %
				\caption{\label{fig:shape-dpfr} Shape of the time function $d_{\text{TAI}\rightarrow\text{PFR}}$. When the \ac{TAI} reaches $\T$, the \ac{PFR} clock $\h_{\text{Reg}}$ starts counting the time slower (with a relative factor $1/\rho$) until its time is late by $\Delta$ compared to the \ac{TAI}. After that point, it counts the time at the same speed but remain $\Delta$ seconds late.}
			\end{figure}
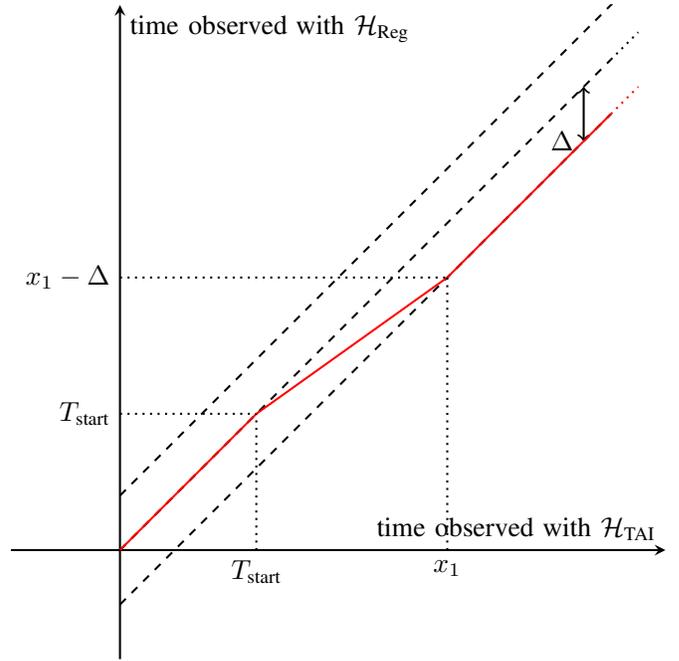
		
		As in Appendix~\ref{proof:insta-async}, let the source be greedy after $T_{\text{start}}$ and generate packets of size $\ell$ bits, with $b\ge\ell$.
		Its cumulative arrival function at the source, measured in $\mathcal{H}_{1}$, is
		\begin{equation}\label{eq-proof-sync-0}
		R^{\hze}(t) = \left\lfloor \frac{\alpha^{\hze}(|t-T_{\text{start}}|^+)}{\ell}\right\rfloor\ell
		\end{equation}
		where $\lfloor \cdot\rfloor$ is the floor function. The fact that this flow satisfies the arrival-curve constraint $\alpha^{\hze}$ in $\hze$ follows again from \cite[Thm III.2]{le2002some}.
		
		Consider now the first network element for the flow, $S_1$ in Figure~\ref{fig:async-hop-model}, as having no delay, for any observation clock.
		
		Let $R'$ be the cumulative arrival function in the \ac{PFR}, that is at the output of the network element. Then for any $t$, $R'^{\hze}(t) = R^{\hze}(t)$ and for any $\tau$, $R'^{\hr}(\tau)=R^{\hr}(\tau)$.
		
		With the same arguments as in Appendix~\ref{proof:insta-async}, we model the PFR as a fluid greedy shaper followed by a packetizer. The latter can be ignored for delay computations. The output of the fluid greedy shaper, observed with $\h_{\text{Reg}}$, is thus
		$
		R^{*\mathcal{H}_{\text{Reg}}}=R'^{\mathcal{H}_{\text{Reg}}}\otimes \gamma_{r,b}
		$. It follows that, for all $\tau$, $R^{*\mathcal{H}_{\text{Reg}}}(\tau) \le R^{\h_{\text{Reg}}}(\T) + \alpha^{\hze}(\tau-\T) = r|\tau-\T|^++b$ by definition of $\T$.
		
		Now let $x_2$ be the next \ac{TAI} time after $x_1$ at which the source finishes sending a packet. Then by definition of $R$ and $R'$
		
		\begin{equation}\label{eq:proof:lower:0}
		R'^{\hze}(x_2)=R^{\hze}(x_2) = \alpha^{\hze}(x_2 - \T)
		\end{equation}
		
		Now consider $x_2+\Delta$ and compute the cumulative output of the regulator at $x_2+\Delta$: $R^{*\hze}(x_2+\Delta) = R^{\hr}(d_{\text{TAI}\rightarrow\text{PFR}}(x_2+\Delta))=R^{\hr}(x_2)$ because $x_2 \ge x_1$ and by definition of $d_{\text{TAI}\rightarrow\text{PFR}}$.
		
		Yet $R^{\hr}(x_2) \le \alpha^{\hze}(x_2-\T)$ so
		\begin{equation}\label{eq:proof:lower:1}
			R^{*\hze}(x_2+\Delta) \le \alpha^{\hze}(x_2-\T)
		\end{equation}
		Combining Equations~\ref{eq:proof:lower:0} and \ref{eq:proof:lower:1} proves
		\begin{equation}\label{eq:proof:lower:2}
			R^{*\hze}(x_2+\Delta) - R^{\hze}(x_2) \le 0
		\end{equation}
		Equation~(\ref{eq:proof:lower:2}) proves that the delay of the packet output at $x_2$ (observed with $\hze$) from the source exits the greedy shaper at $x_2+\Delta$ (observed with $\hze$). It has hence suffered a delay of $\Delta$ measured with $\hze$, which is $\Delta$ more than the worst-case delay through the network element. The worst-case delay is hence lower-bounded by this reachable value.
	\end{proof}

\subsection{Proof of Proposition~\ref{prop:sync:pfr-stable}}
\begin{proof}[Proof of Proposition~\ref{prop:sync:pfr-stable}]\label{proof:sync:pfr-stable}
	We use the service-curve property of \acp{PFR}. The $k$th \ac{PFR} is non-adapted, its configuration is
	\begin{equation*}
		\left\lbrace
		\begin{aligned}
			r_{\text{Reg}_{k}} &= r_0\\
			b_{\text{Reg}_{k}} &= b_0
		\end{aligned}
		\right.
	\end{equation*}
	
	One one hand, $\gamma_{r_{\text{Reg}_{k}},b_{\text{Reg}_{k}}}$ is a service curve of each \ac{PFR} when observed with its clock $\h_{\rr_k}$.
	We use the synchronized part of Table~\ref{tab:sc-results} with $\h_g=\hze$ and $\h_i=\h_{\rr_k}$ and obtain that  $\beta_{\text{PFR}_k}^{\hze} =\left(\delta_{\eta}\otimes\gamma_{r_0/\rho,b_0}\right) \vee\left(\delta_{2\Delta}\otimes\gamma_{r_0,b_0}\right)$ is a service curve of the \ac{PFR} when observed with $\hze$.
	
	On the other hand, $\gamma_{r_{\text{Reg}_{k-1}},b_{\text{Reg}_{k-1}}}$ is an arrival curve of the flow at the input of the $k$-th hop, when observed with $\h_{\rr_{k-1}}$. We apply the synchronized part of Table~\ref{tab:ac-results} with $\h_g=\hze$ and $\h_i=\h_{\rr_{k-1}}$ and obtain that
	$\alpha_{k-1}^{\hze} = \gamma_{\rho r_0,b_0+r_0\eta} \wedge \gamma_{r_0,b_0+2r_0\Delta}$ is an arrival curve of the flow at the input of hop $k$ when observed with $\hze$. Its shape is given in Figure~\ref{fig:pfr-unadapted}.
	
	Assume now that the delay through S$_k$, $D_k^{\hze}$ is computed using the above $\alpha_{k-1}^{\hze}$ as an arrival curve in S$_k$ when observed with $\hze$. The whole hop ($S_k$ followed by regulator) offers, when observed with $\hze$, the service curve $\beta_{\text{Hop}_k}^{\hze} = \delta_{D_k^{\hze}} \otimes \beta_{\text{PFR}_k}^{\hze}$. Its shape is given in Figure \ref{fig:pfr-unadapted}.
	
	We then compute the maximal horizontal distance in the figure. Geometrical considerations give
	\begin{equation*}
		\left\lbrace
		\begin{aligned}
			y_A &= b_0 + \frac{r_0}{\rho-1} (2\Delta \rho - \eta)\\
			y_B &= b_0 + \frac{r_0}{\rho-1} (2\Delta - \eta)
		\end{aligned}
		\right.
	\end{equation*}
	Hence, $y_A\ge y_B$ and the maximum horizontal distance is reached at $A$. We obtain
	\begin{equation}\label{eq:pfr-unadapted:whole-delay}
	D_k'^{\hze} = D_k^{\hze} + 4\Delta
	\end{equation}
	
	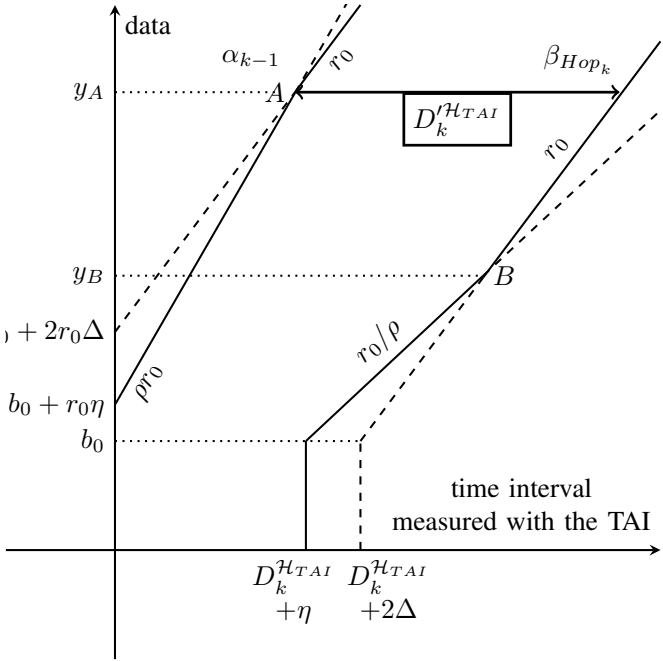
\begin{figure}
		\resizebox{\linewidth}{!}{% !TeX spellcheck = en_US
% !TeX rootfile = article.tex

	\begin{tikzpicture}
	
	\pgfplotsset{ticks=none}
	\begin{axis}[xlabel=\makecell{time interval\\measured with the TAI}, ylabel=data,
	xmin=-4,xmax=20,ymin=-3,ymax=15, axis lines=center, width=10cm, height=10cm]

	%\node[anchor=north] at (axis cs:5,0) (dk) {$D_k^{\mathcal{H}_{m}}$};
	
	\draw[-] (axis cs:7,0) -- (axis cs:7,3) node [pos=0, anchor=north east, xshift=0.5cm] (dpeta) {\makecell[c]{$D_k^{\h_{{TAI}}}$\\$+\eta$}};
	\draw[-, dashed] (axis cs:9,0) -- (axis cs:9,3) node [pos=0, anchor=north west, xshift=-0.3cm] (dpdelta) {\makecell[c]{$D_k^{\h_{{TAI}}}$\\$+2\Delta$}};
	
	\addplot[domain=7:13.5] {0.7*x-1.9} node[pos=0.5,above,sloped] {$r_0/\rho$};	
	\addplot[domain=13.5:20, dashed] {0.7*x-1.9} node[pos=0,right] {$B$};
	\draw[dotted] (axis cs:0,7.55) -- (axis cs:13.5,7.55) node[left,pos=0] {$y_B$};
	\addplot[domain=9:13.5, dashed] {x-6};	
	\addplot[domain=13.5:20] {x-6}  node[pos=0.5,above,sloped] {$r_0$};
	\node at (axis cs:17,13.5) { $\beta_{{Hop}_k}$};
	\node at (axis cs:5,13.5) { $\alpha_{k-1}$};
	
	\draw[dotted] (axis cs:0,3) -- (axis cs:9,3) node[pos=0, left] {$b_0$};
	
	\addplot[domain=0:6.6, dashed] {x+6} node[pos=0,left] {$b_0+2r_0\Delta$};
	\addplot[domain=6.6:9] {x+6}node[pos=0.5,sloped, below] {$r_0$} node[pos=0,left] (na) {$A$};
	\draw[dotted] (axis cs:0,12.6) -- (axis cs:6,12.6) node[left,pos=0] {$y_A$};
	\addplot[domain=0:6.6] {1.3*x+4} node[pos=0,left] {$b_0+r_0\eta$} node[pos=0.1,sloped,below] {$\rho r_0$};
	\addplot[domain=6.6:20, dashed] {1.3*x+4};
	
	\draw[<->, line width=1pt] (axis cs:6.6,12.6) -- (axis cs:18.5,12.6) node[pos=0.5, draw, below] {$D_k'^{\h_{{TAI}}}$};

	\end{axis}
\end{tikzpicture}}
		\caption{\label{fig:pfr-unadapted} Delay bound computation for the flow through the whole $k$th hop ($S_k$ followed by regulator) when the regulator is a \ac{PFR} and the network is synchronized.}
	\end{figure}
\end{proof}

\subsection{Proof of Proposition~\ref{prop:insta-sync-ir}}
	\begin{proof}[Proof of Proposition~\ref{prop:insta-sync-ir}]\label{proof:insta-sync-ir}
		Take $\eta \ge 0$, $\rho > 1$, $\Delta > 0$ and $n \ge 3$. The following example respects the constraints of the model and has unbounded flow latencies.
		
		We consider $n$ sources, each source inputs a single flow to the FIFO system. We choose a \ac{FIFO} system with infinite service and with $\mathcal{H}_{\text{FIFO}}=\mathcal{H}_{\text{IR}}$ such that the delay of the flows trough the \ac{FIFO} system, when observed in the clock $\mathcal{H}_{\text{IR}}$, is null. We further take $\mathcal{H}_{\text{IR}}=\hze$. If the TAI delay is not bounded, then it is also not bounded in any other clock that meets the stability requirements of Equation~(\ref{eq:constr-async}).
		\paragraph*{Adversarial Clocks}
		
		We consider a starting point $x_1\ge T_{\text{start}}$. We choose a slope $s_1$ such that $1< s_1 \le \min(1.5,\sqrt{\rho})$ and define $$I=\frac{\Delta s_1}{s_1-1}$$We also consider any $\epsilon > 0$ such that $\epsilon < I(1-\frac{1}{s_1})$. $\epsilon$ is well defined because $s_1>1$. We also define $\tau = nI/s_1 + n \epsilon$ and $x_j=x_1 + (j-1)I/s_1 + (j-1)\epsilon$ for $j=1\ldots n$.  Now, for every clock $j$, we choose as relative time function  $d_{\text{IR}\rightarrow j}$ the following piecewise linear function
		\begin{equation*}
			\resizebox{\linewidth}{!}{$%
				d_{\text{IR}\rightarrow j} = \left\lbrace
				\begin{aligned}
				t - \Delta/2 &\qquad\text{ if } t\le x_j\\
				s_1 (t-x_j) + x_j - \Delta/2 &\qquad\text{ if } x_j<t\le x_j + I/s_1\\
				\left.\makecell[r]{1/s_1 (t -  x_j + I/s_1)\\+ I + x_j - \Delta/2}\right\rbrace&\qquad\text{ if } x_j+\frac{I}{s_1} < t \le x_j + \frac{I}{s_1} + I\\
				t - \Delta/2 &\qquad\text{ if } x_j + \frac{I}{s_1} + I <t\le x_j + \tau\\
				\tau + d_j(t-\tau) &\qquad\text{ if }  x_j + \tau < t \\
				\end{aligned}
				\right.	
				$}%
		\end{equation*}
		The shape of function $d_{\text{IR}\rightarrow j}$ is available in Figure~\ref{fig:shape-dj}. The slope between $x_j$ and $x_j + I/s_1$ corresponds to the exaggerated compression of the time line in Figure~\ref{fig:reality} of the example in Section~\ref{sec:ir-s-0}. We obtain directly the following properties for any $j$
		\begin{itemize}
			\item $d_{\text{IR}\rightarrow j}$ is continuous and strictly increasing
			\item The time-error function $t\mapsto d_{\text{IR}\rightarrow j}(t) - t$ is periodic with period $\tau$ for $t\ge x_j$
		\end{itemize}
		
		Also, for any $j,j'$ $d_{j \rightarrow j'}(t) = d_{\text{IR}\rightarrow j'}(d^{-1}_{\text{IR}\rightarrow j}(t))$. As $s_1\le\sqrt{\rho}$ and $|d_{\text{IR}\rightarrow j}(t) - t|\le \Delta/2$, any pair of clocks $(\h_{j},\h_{j'})$ meets the constraints of Equations~(\ref{eq:constr-async}) and (\ref{eq:const-sync}), which shows that our adversarial clocks are within the synchronized time model proposed in Section~\ref{sec:time-model}.
		
		\begin{figure}\centering%
			\resizebox{\linewidth}{!}{% !TeX spellcheck = en_US
% !TeX root=ms.tex
\begin{tikzpicture}
\begin{axis}[xlabel=time observed with $\h_{\text{IR}}$, ylabel=time observed with $\h_j$, xmin=-5, ymin=-4, xmax=20, ymax=20, no marks, axis lines=center, ticks=none, width=10cm, height=10cm]
	\addplot[domain=0:18, dashed] {x};
	\addplot[domain=18:19, dotted] {x};
	\addplot[domain=0:18, dashed] {x-2};
	\addplot[domain=18:19, dotted] {x-2};
	\addplot[domain=0:20, dashed] {x+2};
	
	\draw[dotted] (axis cs:4,0) -- (axis cs:4,2) node[pos=0, below] {$x_j$};
	\draw[red] (axis cs:2,0) -- (axis cs:4,2);
	
	\draw[dotted] (axis cs:7,0) -- (axis cs:7,9) node[pos=0, below] {\makecell[c]{$x_j$\\$+\frac{I}{s_1}$}};
	\draw[red] (axis cs:4,2) -- (axis cs:7,9) node[pos=0.6, above, sloped] {$s_1$};
	
	\draw[dotted] (axis cs:15,0) -- (axis cs:15,13) node[pos=0,below] {\makecell[c]{$x_j$\\$+\frac{I}{s_1}$\\$+I$}};
	\draw[red] (axis cs:7,9) -- (axis cs:15,13) node[below,pos=0.5, sloped] {$1/s_1$};
	
	\draw[dotted] (axis cs:18,0) -- (axis cs:18,16) node[pos=0,below] {$x_j+\tau$};
	\draw[red] (axis cs:15,13) -- (axis cs:18,16);
	\draw[red, dotted] (axis cs:18,16) -- (axis cs:19,18.3);

	\draw[dotted] (axis cs:4,2) -- (axis cs:0,2) node[pos=1, left] {$d_{\text{IR} \rightarrow j}(x_j)$};
	\draw[dotted] (axis cs:7,9) -- (axis cs:0,9) node[pos=1,left] {\makecell[r]{$d_{\text{IR} \rightarrow j}(x_j)$\\$+I$}};
	\draw[dotted] (axis cs:18,16) -- (axis cs:0,16) node[pos=1, left] {\makecell[r]{$d_{\text{IR} \rightarrow j}(x_j)$\\$+\tau$}};
	\draw[dotted] (axis cs:0,13) -- (axis cs:15,13) node[pos=0, left] {\makecell[r]{$d_{\text{IR} \rightarrow j}(x_j)$\\$+I+\frac{I}{s_1}$}};
	
	\draw[<->] (axis cs:17,17) -- (axis cs:17,19) node[pos=1, left] {$\Delta/2$};
	
	\end{axis}
\end{tikzpicture}} %
			\caption{\label{fig:shape-dj} Shape of the time function $d_{\text{IR} \rightarrow j}$. The shape is periodic, with period $\tau$. When clock $\h_{\text{IR}}$ reaches $x_j$, clock $\h_j$ starts counting the time faster (with a relative factor $s_1$) until $\h_{\text{IR}}$ counts $I/s_1$ more seconds. Then, $\h_j$ counts the time slower (with a relative factor $1/s_1$). When the time-error function between $\h_{\text{IR}}$ and $\h_j$ reaches $-\Delta/2$, $\h_j$ counts the time at the same speed as $\h_{\text{IR}}$ until the next period starts.}
		\end{figure}
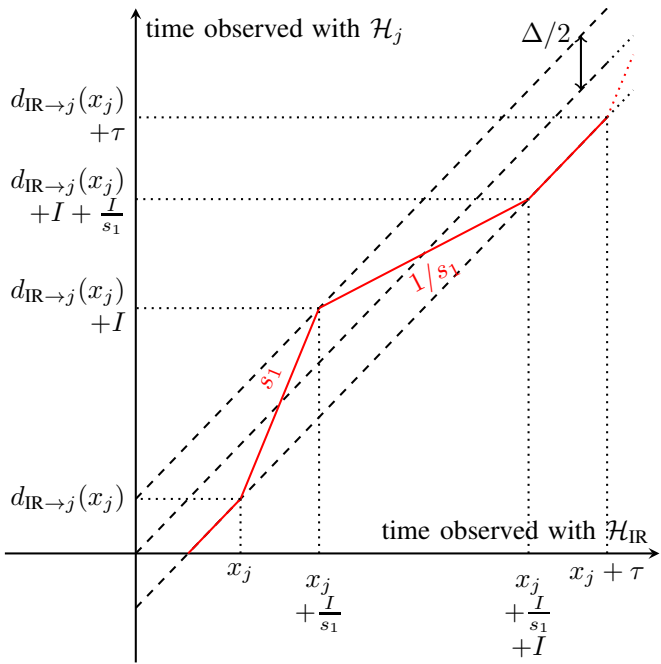
		\paragraph*{Adversarial Traffic Generation}:
		Let $l$ be any arbitrary data size. Each source $j$ is configured to send a packet of size $l$ when its local clock reaches $d_j(x_j)+k\tau$ and $d_j(x_j)+k\tau+I$ for all $k\in\mathbb{N}$. Figure \ref{fig:timeline-local} presents the traffic generation of source $j$ within one period, observed with its internal clock $\h_j$. When observed with $\h_j$, the traffic generation is periodic of period $\tau$.
		\begin{figure}\centering %
			\resizebox{\linewidth}{!}{% !TeX spellcheck = en_US
% !TeX root=article.tex

\begin{tikzpicture}

	\draw[->] (0,0) -- (11.5,0) node[pos=1, anchor=east] {\makecell[r]{time observed\\with $\mathcal{H}_j$}};
	
	\draw[->] (1,0) -- (1,1) node[pos=0, anchor=north] {$\begin{aligned}&d_j(x_j)\\&+ k\tau\end{aligned}$};
	\draw[->] (4,0) -- (4,1) node[pos=0, anchor=north] {$\begin{aligned}&d_j(x_j)\\&+ k\tau+I\end{aligned}$};
	\draw[->] (9,0) -- (9,1) node[pos=0, anchor=north] {$\begin{aligned}&d_j(x_j)\\&+ (k+1)\tau\end{aligned}$};
	
	\draw[<->] (1,1.2) -- (4,1.2) node[midway, above] {$I$};
	\draw[<->] (1,1.8) -- (9,1.8) node[midway, above] {$\tau$};
\end{tikzpicture}} %
			\caption{\label{fig:timeline-local} Generation of packets as observed with $\h_j$. The traffic profile is periodic of period $\tau$. Source $j$ sends a packet when the internal clock reaches $d_j(x_j) + k\tau$ for some $k\in\mathbb{N}$, then it sends another packet after a duration of $I$ counted using $\h_j$, and finally restarts at the next period.}
		\end{figure}
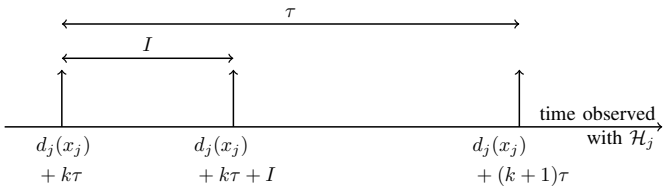
		
		As $n\ge3$ and $s_1 < 1.5$, $\tau\ge 2I + 2\epsilon \ge 2I$ and $\tau-I \ge I$. Hence, the minimum duration between two packets generated by source $j$, counted using $\h_j$ is $I$. This proves that each flow exits its respective source $j$ with a leaky-bucket arrival curve $\gamma_{\frac{l}{I},l}$ (rate $l/I$, burst $l$) when observed using $\mathcal{H}_j$.	We now assume that the interleaved regulator is configured with the same leaky-bucket arrival curve $\gamma_{\frac{l}{I},l}$ for all the flows.
		
		\begin{figure}\centering %
			\resizebox{\linewidth}{!}{% !TeX spellcheck = en_US
% !TeX root=article.tex

\begin{tikzpicture}

	\draw[->] (0,0) -- (11.5,0) node[pos=1, anchor=east] {\makecell[r]{time observed\\with $\mathcal{H}_{\text{IR}}$}};
	
	\draw[->] (1,0) -- (1,1) node[pos=0, anchor=north] {$\begin{aligned}&x_j\\+ &k\tau\end{aligned}$};
	\draw[->] (3,0) -- (3,1) node[pos=0, anchor=north] {$\begin{aligned}&x_j\\+ k\tau&+\frac{I}{s_1}\end{aligned}$};
	\draw[->] (9,0) -- (9,1) node[pos=0, anchor=north] {$\begin{aligned}&x_j\\+ (k&+1)\tau\end{aligned}$};
	
	\draw[<->] (1,1.2) -- (3,1.2) node[midway, above] {$I/s_1$};
	\draw[<->] (1,1.8) -- (9,1.8) node[midway, above] {$\tau$};
\end{tikzpicture}} %
			\caption{\label{fig:timeline-global} Generation of packets as observed with $\mathcal{H}_{\text{IR }}$. The traffic profile is periodic of period $\tau$. When observing with $\h_{\text{IR}}$, source $j$ sends packets at $x_j + k\tau$ and $x_j + k\tau + \frac{I}{s_1}$ for all $k\in\mathbb{N}$.}
		\end{figure}
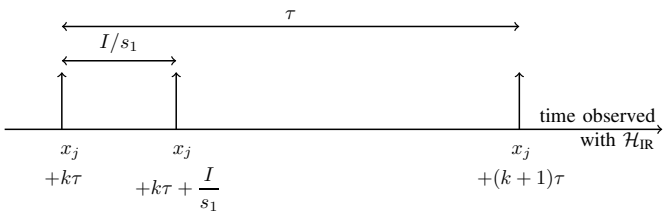
		
		Figure \ref{fig:timeline-global} presents the timeline of packets generated by source $j$ but as observed with $\h_{\text{IR}}$. The IR has to regulate the same timeline for all the $n$ inputs, based on its configuration. For  $j=1\ldots n$ %$j\in\llbracket0,n-1\rrbracket$ 
and for $k\in\mathbb{N}$ we note $A^1_{j,k}=x_j+k\tau$ the arrival time in the \ac{IR} of the first packet of the the $k$th period of source $j$ measured with $\mathcal{H}_{\text{IR}}$ and $A^2_{j,k}=x_j+k\tau+\frac{I}{s_1}$ the arrival time of the second packet, still measured with $\mathcal{H}_{\text{IR}}$. Also, note $D^1_{j,k}$ and $D^2_{j,k}$ their respective release time out of the IR, again measured using $\mathcal{H}_{\text{IR}}$.
		
		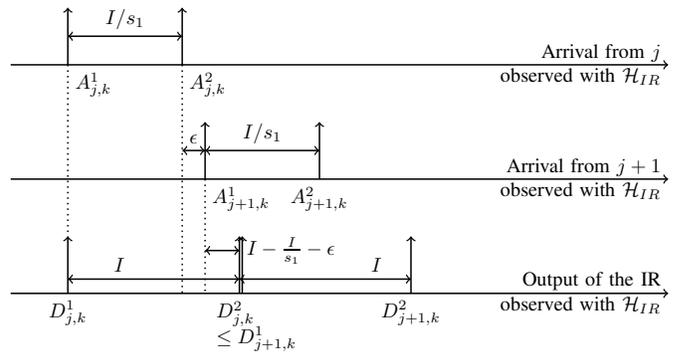
\begin{figure}\centering %
			\resizebox{\linewidth}{!}{% !TeX spellcheck = en_US
% !TeX root=article.tex

\begin{tikzpicture}

	\draw[->] (0,0) -- (11.5,0) node[pos=1, anchor=east] {\makecell[r]{Arrival from $j$\\observed with $\mathcal{H}_{IR}$}};
	\draw[->] (1,0) -- (1,1) node[pos=0, anchor=north west] {$A^1_{j,k}$};
	\draw[->] (3,0) -- (3,1) node[pos=0, anchor=north west] {$A^2_{j,k}$};
	\draw[<->] (1,0.5) -- (3,0.5) node[pos=0.5,above] {$I/s_1$};
	
	\draw[->] (0,-2) -- (11.5,-2) node[pos=1, anchor=east] {\makecell[r]{Arrival from $j+1$\\observed with $\mathcal{H}_{IR}$}};
	\draw[->] (3.4,-2) -- (3.4,-1) node[pos=0, anchor=north west] {$A^1_{j+1,k}$};
	\draw[->] (5.4,-2) -- (5.4,-1) node[pos=0, anchor=north] {$A^2_{j+1,k}$};
	\draw[<->] (3,-1.5) -- (3.4,-1.5) node[pos=0.5,above] {$\epsilon$};
	\draw[<->] (3.4,-1.5) -- (5.4,-1.5) node[pos=0.5,above] {$I/s_1$};
	
	\draw[->] (0,-4) -- (11.5,-4) node[pos=1, anchor=east] {\makecell[r]{Output of the IR\\observed with $\mathcal{H}_{IR}$}};
	\draw[->] (1,-4) -- (1,-3) node[pos=0, anchor=north] {$D^1_{j,k}$};
	\draw[->] (4,-4) -- (4,-3) node[pos=0, anchor=north, xshift=0.3cm] {\makecell[l]{$D^2_{j,k}$\\$\le D^1_{j+1,k}$}};
	\draw[->] (4.05,-4) -- (4.05,-3);
	\draw[->] (7,-4) -- (7,-3) node[pos=0, anchor=north] {$D^2_{j+1,k}$};
	
	\draw[dotted] (1,0) -- (1,-3);
	\draw[dotted] (3,0) -- (3,-4);
	\draw[dotted] (3.4,-2) -- (3.4,-4);
	\draw[<->] (1,-3.75) -- (4,-3.75) node [pos=0.3,above] {$I$};
	\draw[<->] (3.4,-3.25) -- (4,-3.25) node [pos=1,anchor=west] {$I-\frac{I}{s_1}-\epsilon$};
	\draw[<->] (4,-3.75) -- (7,-3.75) node [pos=0.8,above] {$I$};

\end{tikzpicture}} %
			\caption{\label{fig:timeline-two} Traffic arrival from two successive upstream sources, as observed with $\mathcal{H}_{\text{IR}}$ and release time of the packets, again observed with $\h_{\text{IR}}$.}
		\end{figure}
		Figure \ref{fig:timeline-two} presents the arrival and release times of the packets for two consecutive sources. Assume for instance that the source has been idle for a while, then $D^1_{j,k} = A^1_{j,k}$. The instant $A^2_{j,k}$ is, from the perspective of $\h_{\text{IR}}$, too soon by $I(1-\frac{1}{s_1})$. Using the IR equations \cite{leboudecTheoryTrafficRegulators2018}, the IR has to delay the packet and
		\begin{equation}\label{eq:d2vsd1}\forall j=1\ldots n, \forall k\in\mathbb{N}, \quad D^2_{j,k} \ge D^1_{j,k} + I\end{equation}
		As $x_{j+1}=x_j + \frac{I}{s_1} + \epsilon$, packet $A^1_{j+1,k}$ arrives $\epsilon$ seconds after $A^2_{j,k}$ (measured using $\h_{\text{IR}}$). As $\epsilon < I(1-\frac{1}{s_1})$, the packet at $A^1_{j+1,k}$ arrives before the previous packet of the previous source could be released out of the \ac{IR}. Because the \ac{IR} only looks at the head-of-line packet, and using the IR equations, we obtain
		\begin{equation}\label{eq:jp1vsj}\forall j=1\ldots n-1, \forall k\in\mathbb{N}, \quad D^1_{j+1,k} \ge D^2_{j,k}\end{equation}
		Combining Equations \ref{eq:d2vsd1} and \ref{eq:jp1vsj} gives, by induction,
		\begin{equation}\label{eq:firstDvsLastD} D^2_{n,k}\ge D^1_{1,k} + nI\end{equation}
		Now we can note that
		\begin{align*}
			A^1_{1,k+1}&=x_1+k\tau+\tau\\
			&=x_1+k\tau+n\frac{I}{s_1}+n\epsilon\\
			&=x_1 + (n-1)\frac{I}{s_1} + (n-1)\epsilon + k\tau + \frac{I}{s_1} + \epsilon\\
			&=x_{n} + k\tau + \frac{I}{s_1} + \epsilon\\
			&=A^2_{n,k} + \epsilon
		\end{align*}
		Hence, the first packet of the $(k+1)$th period of the first upstream source arrives $\epsilon$ seconds (counted with $\h_{\text{IR}}$) after the second packet of the $k$th period of the last source, so we also have
		\begin{equation}\label{eq:loop}
		D^1_{1,k+1} \ge D^2_{n,k}
		\end{equation}
		Combining equations \ref{eq:d2vsd1} and \ref{eq:loop} gives
		\begin{equation}
		D^1_{1,k}\ge	D^1_{1,1} + (k-1)nI = x_1 + (k-1)nI
		\end{equation}
		Because we have $D^1_{1,1} = x_1$, as the network was empty before.
		The delay suffered through the IR by the first packet of the $k$th period of the first source is, when measured with $\h_{\text{IR}}$
		\begin{align}
			D^1_{1,k} - A^1_{1,k} &\ge x_1 + (k-1)nI - x_1 - (k-1)\tau\\
			&\ge x_1 + (k-1)nI - x_1 - (k-1)n\frac{I}{s_1} - (k-1)n\epsilon\\
			&\ge (k-1)n\left(I\left(1-\frac{1}{s_1}\right)-\epsilon\right) \label{eq:diverge}
		\end{align}
		As we have arbitrary selected $\epsilon$ such that $\epsilon < I(1-\frac{1}{s_1})$, we obtain $I(1-\frac{1}{s_1}) - \epsilon > 0$ thus the above delay lower-bound diverges as $k$ increases, so the delay through the IR is unbounded when seen from $\h_{\text{IR}}$, which proves the instability.
		
		\rema Equation~(\ref{eq:diverge}) proves that at each period of duration $\tau$, the worst-case delay increases by $nI(1-\frac{1}{s_1})-n\epsilon$. The divergence of the worst-case delay per second is
		\begin{align*}
			\text{div} &= \frac{nI(1-\frac{1}{s_1})-n\epsilon}{\frac{nI}{s_1}+n\epsilon}\\
		\end{align*}
		This divergence is valid for any $\epsilon>0$, with $\epsilon < I(1-\frac{1}{s_1})$. Taking $\epsilon \rightarrow 0$, the divergence can be as large as
		\begin{align*}
			\lim_{\epsilon\rightarrow0}\text{div} &= \frac{nI(1-\frac{1}{s_1})}{\frac{nI}{s_1}}\\
			&= s_1-1
		\end{align*}
		$s_1$ can be as large as $\sqrt{\rho}$, so the divergence of the worst-case delay can be as large as $\sqrt{\rho}-1$, for any $n\ge3, \Delta>0$.
	\end{proof}

% !TeX spellcheck = en_US
% !TeX rootfile = article.tex
%\todo{acronyms for Holly - To be removed from final version}
\begin{acronym}[CP-OFDMX] 
	\acro{ADAM}{asynchronous dual arrival-curve method}
	\acro{AFDX}{avionics full-dupleX switched ethernet}
	\acro{ATS}{asynchronous traffic shaping}
	\acro{CAN}{Controller Area Network}
	\acro{CBS}{credit-based scheduler}
	\acro{CDT}{control-data traffic}
	\acro{CEV}{crew exploration vehicule}
	\acro{CQF}{cyclic queuing and forwarding}
	\acro{ETE}{end-to-end}
	\acro{FIFO}{first in, first out}
	\acro{FP}{fixed-priority}
	\acro{GCL}{gate control list}
	\acro{GNSS}{global navigation satellite system}
	\acro{GPS}{Global Positioning System}
	\acro{IEEE}{Institute of Electrical and Electronics Engineers}
	\acro{IETF}{Internet Engineering Task Force}
	\acro{IR}{interleaved regulator}
	\acro{ITU}{International Telecommunication Union}
	\acro{LCAN}{low-cost acyclic network}
	\acro{xTFA}[FP-TFA]{fixed-point total-flow analysis}
	\acro{MFAS}{minimum feedback arc set}
	\acro{MFVS}{minimum feedback vertex set}
	\acro{MTIE}{maximum time interval error}
	\acro{NC}{network calculus}
	\acro{NoC}{networks on chip}
	\acro{ns-3}{network simulator 3}
	\acro{NTP}{Network Time Protocol}
	\acro{PBOO}{pay burst only once}
	\acro{PFR}{per-flow regulator}
	\acro{PMOC}{pay multiplexing only at convergence points}
	\acro{PTP}{Precision Time Protocol}
	\acro{TAI}{international atomic time (\emph{temps atomique international})}
	\acro{TAS}{time-aware shaper}
	\acro{TDEV}{time deviation}
	\acro{TFA}{total-flow analysis}
	\acro{TIE}{time interval error}
	\acro{TP}{turn prohibition}
	\acro{TSN}{time-sensitive networking}
\end{acronym}

\end{document}